\documentclass{article}%
\usepackage{boxedminipage}
\usepackage{amssymb}
\usepackage{amsfonts}
\usepackage{amsmath}
\usepackage{revsymb}
\usepackage[portrait]{geometry}
\usepackage{hyperref}
\usepackage[square,comma,numbers,sort&compress]{natbib}
\usepackage{hypernat}
\usepackage{graphicx}%
\providecommand{\U}[1]{\protect\rule{.1in}{.1in}}
\newtheorem{theorem}{Theorem}

\newtheorem{algorithm}[theorem]{Observation}

\newtheorem{corollary}[theorem]{Corollary}

\newtheorem{definition}[theorem]{Definition}

\newtheorem{lemma}[theorem]{Lemma}

\newtheorem{problem}[theorem]{Problem}
\newtheorem{proposition}[theorem]{Proposition}

\newenvironment{proof}[1][Proof]{\noindent\textbf{#1.} }{\ \rule{0.5em}{0.5em}}
\hypersetup{
colorlinks=true,       linkcolor=blue,          citecolor=blue,        filecolor=blue,      urlcolor=blue           }
\begin{document}

\title{Two-sided bounds on minimum-error quantum measurement, on the reversibility of
quantum dynamics, and on maximum overlap using directional iterates}
\author{Jon Tyson\thanks{jonetyson@X.Y.Z, where X=post, Y=Harvard, Z=edu}\\Jefferson Lab, Harvard University}
\date{Posted July 10, 2009\\
Revised June 1, 2010}
\maketitle

\begin{abstract}
\noindent In a unified framework, we estimate the following quantities of
interest in quantum information theory:

\begin{enumerate}
\item The minimum-error distinguishability of arbitrary ensembles of mixed
quantum states.

\item The approximate reversibility of quantum dynamics in terms of
entanglement fidelity. (This is referred to as "channel-adapted quantum error
recovery" when applied to the composition of an encoding operation and a noise channel.)

\item The maximum overlap between a bipartite pure quantum state and a
bipartite mixed state that may be achieved by applying a local quantum
operation to one part of the mixed state.

\item The conditional min-entropy of bipartite quantum states.

\end{enumerate}

\noindent A refined version of the author's techniques
\href{http://link.aip.org/link/?JMAPAQ/50/032106/1}{[J. Math. Phys.
\textbf{50, }032016]} for bounding the first quantity is employed to give
two-sided estimates of the remaining three quantities.\vspace{0.02in}

We obtain a closed-form approximate reversal channel. Using a state-dependent
Kraus decomposition, our reversal may be interpreted as a
quadratically-weighted version of that of Barnum and Knill
\href{http://dx.doi.org/10.1063/1.1459754}{[J. Math. Phys. \textbf{43},
2097]}. The relationship between our reversal and Barnum and Knill's is
therefore similar to the relationship between Holevo's asymptotically-optimal
measurement [Theor. Probab. Appl. \textbf{23}, 411] and the \textquotedblleft
pretty good\textquotedblright\ measurement of Belavkin [Stochastics \textbf{1,
}315] and Hausladen \& Wootters
\href{http://dx.doi.org/10.1080/09500349414552221}{[J. Mod. Optic.
\textbf{41}, 2385]}. In particular, we obtain relatively simple reversibility
estimates without negative matrix powers at no cost in tightness of our
bounds. Our recovery operation is found to significantly outperform the
so-called \textquotedblleft transpose channel\textquotedblright\ in the simple
case of depolarizing noise acting on half of a maximally-entangled state.
Furthermore, our overlap results allow the entangled input state and the
output target state to differ, thus obtaining estimates in a somewhat more
general setting.\vspace{0.02in}

Using a result of K\"{o}nig, Renner, and Schaffner
\href{http://dx.doi.org/10.1109/TIT.2009.2025545}{[IEEE. Trans. Inf. Th. 55,
4337]}, our maximum overlap estimate is used to bound the conditional
min-entropy of arbitrary bipartite states.\vspace{0.02in}

Our primary tool is \textquotedblleft small angle\textquotedblright%
\ initialization of an abstract generalization of the iterative schemes of
Je\v{z}ek-\v{R}eh\'{a}\v{c}ek-Fiur\'{a}\v{s}ek
\href{http://dx.doi.org/10.1103/PhysRevA.65.060301}{[Phys. Rev. A \textbf{65,
}060301]}, Je\v{z}ek-Fiur\'{a}\v{s}ek-Hradil
\href{http://link.aps.org/doi/10.1103/PhysRevA.68.012305}{[Phys. Rev. A
\textbf{68}, 012305]}, and Reimpell-Werner
\href{http://link.aps.org/doi/10.1103/PhysRevLett.94.080501}{[Phys. Rev. Lett.
\textbf{94}, 080501]}. The monotonicity result of Reimpell
\href{http://deposit.ddb.de/cgi-bin/dokserv?idn=988217317}{[Ph.D. Thesis,
2007]} follows in greater generality.

\newpage

\end{abstract}
\tableofcontents

\newpage

\section{Introduction}

\noindent This paper considers the following problem of relevance in quantum
information theory:

\begin{quotation}
\noindent\textbf{The maximum overlap problem}: Let $\mu_{\mathcal{\mathcal{KH}%
}}$ be a positive semidefinite trace-class operator on $\mathcal{K}%
\otimes\mathcal{H}$, and let $M_{\mathcal{\mathcal{LH}}}$ be positive
semidefinite bounded operator on $\mathcal{L}\otimes\mathcal{H}$, where
$\mathcal{H},$ $\mathcal{K},$ and $\mathcal{L}$ are separable Hilbert spaces.
What is maximum overlap%
\begin{equation}
\operatorname{MO}\left(  \mu_{\mathcal{\mathcal{\mathcal{\mathcal{KH}}}}%
},M_{\mathcal{\mathcal{LH}}}\right)  =\sup_{\mathcal{R}}\operatorname*{Tr}%
_{\mathcal{\mathcal{LH}}}\left(  M_{\mathcal{\mathcal{LH}}}\mathcal{R}%
_{\mathcal{K}\rightarrow\mathcal{L}}\left(  \mu_{\mathcal{KH}}\right)
\right)  \text{,}\label{eq defining maximum overlap}%
\end{equation}
where the supremum is over all quantum operations $\mathcal{R}$ from
$\mathcal{K}$ to $\mathcal{L}$?
\end{quotation}

\noindent The maximum-overlap problem has the following important special cases:

\bigskip

\noindent\textbf{1. The minimum-error quantum detection problem} \cite{Holevo
remarks on optimal measurements, Yuen Ken Lax Optimum testing of multiple,
Holevo optimal measurement conditions 1, Barnett and Croke On the conditions
for discrimination between quantum states with minimum error}\textbf{:}

\begin{quotation}
\noindent If an unknown quantum state $\rho_{k}$ is randomly selected from
given ensemble of such states, with what probability may the value of $k$ be
determined by a carefully-chosen quantum measurement?
\end{quotation}

\noindent\textbf{2. \textbf{Approximate reversal of quantum dynamics }}%
\cite{Barnum Knill UhOh,
Schumacher Westmoreland Approximate Quantum error correction,
Jezek Fiurasek Hradil Quantum inference of states and processes,
Reimpell Werner,
Fletcher Thesis Channel Adapted quantum error correction,
Fletcher Shor Win Optimum quantum error recovery using semidefinite programming,
Reimpell Werner Audenaert,
Reimpell Thesis,
Fletcher Shor Win Channel-Adapted quantum error correction for the amplitude dampin channel,
Fletcher Shor Win Structured Near-Optimal Channel-Adapted Quantum Error Correction,
Kosut Shabani Lidar Robust Quantum error correction via convex optimization,
Taghavi Channel-Optimized quantum error correction,
Yamamoto Hara Tsumura suboptimal quantum-error-correcting procedure based on semidefinite programming,
Beny and Oreshkov general conditions for approximate quanutm error recovery and near optimal recovery channels,
Beny and Oreshkov in preparation,
Ng Mandayam simple approach to approximate QEC}%
:

\begin{quotation}
\noindent Suppose that an arbitrary quantum operation $\mathcal{A}$ acts on a
given quantum state $\rho$. How well may the action of $\mathcal{A}$ be
reversed by application of a recovery channel $\mathcal{R}$, so as to preserve
the entanglement of the original system with the environment? This problem is
one of \textbf{\textquotedblleft channel-adapted quantum error
recovery\textquotedblright}\ when the operation $\mathcal{A}$ is of the form
$\mathcal{A}=\mathcal{N}\circ\mathcal{E}$, where $\mathcal{E}$ is an encoding
operation designed to protect against a known noise process $\mathcal{N}$.
\end{quotation}

\noindent\textbf{3. Estimation of conditional min-/max-entropy of bipartite
quantum states }\cite{Renner thesis}:

\begin{quotation}
\noindent Let $\rho_{AB}$ be a bipartite quantum mixed state. Estimate the
conditional min-entropy $H_{\text{min}}\left(  A|B\right)  $ of $A$ given $B$.
\end{quotation}

Since all of these problems are believed to defy closed-form solution, the
purpose of this paper is to provide estimates. In section
\ref{section minimum error distinction a max prenorm problem} a refined proof
of the two-sided \textquotedblleft generalized
Holevo-Curlander\textquotedblright\ bounds of \cite{Tyson Holevo Curlander
Bounds, Ogawa and Nagoaka Strong converse to the quantum channel coding
theorem} for case $1$ is given. This method is extended in sections
$\ref{section max overlap}$ and \ref{section approximate channel reversals} to
yield simple two-sided estimates for cases $2$-$3$ and for $\operatorname{MO}%
\left(  \mu,M\right)  $ in the case of rank-1 $M$.

We briefly introduce each of the cases 1-3 before outlining our approach and
surveying closely-related work.

\subsection{Minimum-error detection}

The minimum-error quantum detection problem was first studied in the 1960's in
connection with the design of optical detectors \cite{Helstrom Quantum
Detection and Estimation Theory}, and it has since become of importance in
quantum Shannon theory (for example \cite{pure state HSW theorem, mixed state
HSW theorem, Holevo mixed state HSW theorem}) and in the design of quantum
algorithms
\cite{Ip Shor's algorithm is optimal,
Bacon,Childs from optimal to efficient algo, optimal alg for hidden shift,
Hayashi Kawachi Kobayashi quantum measurements for hidden subgroup problems with optimal sample complexity,
moore and russels distinguishing, Bacon new hidden subgroup, Radhakrishnan Rotteler Sen Random measurement bases quantum state distinction and the HSP,
Hunziker The geometry of quantum learning}%
. A generalization to the theory of wave pattern recognition may be found in
\cite{Belavkin Book}. Various general upper and/or lower bounds on quantum
distinguishability may be found in
\cite{pure state HSW theorem,Hayden Leung Multiparty
Hiding,Barnum Knill UhOh,Montanaro on the distinguishability of random quantum
states,Daowen Qui Minimum-error discrimination between mixed quantum
states,Montanaro a lower bound on the probability of error in quantum state
discrimination, Hayashi Kawachi Kobayashi quantum measurements for hidden subgroup problems with optimal sample complexity,
Qiu and Li bounds on the minimum error discrimination between
mixed quantum states,
Ogawa and Nagoaka Strong converse to the quantum channel coding theorem,
Tyson Holevo Curlander Bounds, Holovo Assym Opt Hyp Test, Curlander thesis MIT}%
.

The minimum-error quantum detection problem is precisely formulated by

\begin{definition}
\label{definition containing ensemble E}Let
\begin{equation}
\mathcal{E}=\left\{  \rho_{k}\right\}  _{k\in K}%
\label{a priori normed ensemble to distinguish}%
\end{equation}
be an ensemble of quantum states, represented as positive semidefinite
operators normalized by a-priori probability, setting
\begin{equation}
\operatorname*{Tr}\rho_{k}=p_{k},\label{a priori normalization}%
\end{equation}
where $p_{k}$ is the likelihood that $\rho_{k}$ will be drawn from
$\mathcal{E}$. A quantum measurement \cite{Mike and Ike} is described by a
\textbf{positive-operator-valued measure} (\textbf{POVM)}, which consists of a
vector $M=\left\{  M_{k}\right\}  _{k\in K}$ of positive semidefinite
operators satisfying $%
{\displaystyle\sum}
M_{k}\leq%
\openone
$.$^{\text{\cite{subpovm}}}$ (Throughout this paper the operator inequality
$A\leq B$ means $B-A$ is positive semidefinite.) The probability that the
value $k$ is measured when $M$ is applied to a unit-trace density matrix
$\rho$ is given by%
\[
\Pr\nolimits_{M}\left(  k\,|\,\rho\right)  =\operatorname*{Tr}M_{k}%
\rho\text{.}%
\]
The \textbf{success rate} for the POVM\ $M$ to correctly determine the value
of $k$ corresponding to a random element of the ensemble $\mathcal{E}$ is
given by%
\begin{equation}
P_{\text{succ}}\left(  M\right)  =%
{\displaystyle\sum_{k}}
~p_{k}\Pr\nolimits_{M}\left(  k\,|\,\frac{\rho_{k}}{p_{k}}\right)
=\operatorname*{Tr}%
{\displaystyle\sum_{k\in K}}
M_{k}\rho_{k}\text{.\label{formula for Psucc}}%
\end{equation}
The \textbf{minimum-error measurement problem} consists of finding a POVM
maximizing $\left(  \ref{formula for Psucc}\right)  $.
\end{definition}

\subsubsection{The relationship to \textquotedblleft
worst-case\textquotedblright\ detection}

Sometimes one is interested in the \textquotedblleft
worst-case\textquotedblright\ distinguishability
\begin{equation}
\max_{M}\min_{k}\operatorname*{Tr}%
{\displaystyle\sum}
M_{k}\hat{\rho}_{k},\label{worst case measurement}%
\end{equation}
of a collection of unit-trace states $\hat{\rho}_{k}$. As pointed out in
\cite{Harrow Winter How many copies are needed for state discrimination}, the
minimax theorem \cite{Morgenstern von Neumann minimax} implies that%
\begin{equation}
\max_{M}\min_{k}\operatorname*{Tr}\left(  M_{k}\hat{\rho}_{k}\right)
=\max_{M}\min_{\left\{  p_{k}\right\}  }\operatorname*{Tr}%
{\displaystyle\sum}
M_{k}p_{k}\hat{\rho}_{k}=\min_{\left\{  p_{k}\right\}  }\max_{M}%
\operatorname*{Tr}%
{\displaystyle\sum}
M_{k}p_{k}\hat{\rho}_{k},
\end{equation}
where $\left\{  p_{k}\right\}  $ represents a probability distribution. In
particular, single-instance bounds (for fixed $\left\{  p_{k}\right\}  $) may
in principle be minimized over all distributions $\left\{  p_{k}\right\}  $ to
give corresponding \textquotedblleft worst-case\textquotedblright\ bounds.

\subsection{Channel-adapted quantum error recovery}

The following problem is of importance in quantum information theory, quantum
communication, and quantum computing:

\begin{quotation}
\noindent\textit{Suppose that one wishes to store, process, or transmit
quantum data using a process that is subject to noise or loss. How well may
the effects of this noise be avoided, corrected, or eliminated by encoding the
data into a protected form, from which it may be later recovered unharmed by
this noise?}
\end{quotation}

This problem arises in any physical implementation of quantum communication or
computation, since unmitigated interactions with the environment tend to
corrupt quantum signals or memory. By the celebrated \textquotedblleft
threshold theorem\textquotedblright\ \cite{Aharonov Ben-Or F2, A. Kitaev
Russian Math. Surveys, Knill Laflamme Zurek, Aliferis Gosttesman Preskill,
Aharonov Kitaev Preskill}, one may in principle use error correction and
concatenated quantum codes to perform an arbitrary quantum computation in the
presence of noise below a fixed \textquotedblleft threshold\textquotedblright\ amount.

Standard quantum error correction seeks to design encoding and decoding maps
which \textit{exactly} correct for a given class of errors. Early successes of
this program were the first codes that could protect against arbitrary
single-qubit errors \cite{Shor Scheme for reducing decoherence in quantum
computer memory, Steane Error correcting codes in quantum theory, Calderbank
and Shor good quantum error-correcting codes exist}, followed by general
theoretical advances of \cite{Knill and Laflamme Theory of QEC}, and by the
construction of codes that correct for arbitrary single-qubit errors by
encoding a single qubit into five \cite{Bennett DiVincenzo Smolin Wootters
mixed state entanglement and QEC, Laflamme Miquel Paz Zurek Perfect QEC}.

Alternatively, one may consider approximate quantum error correction. For
example, Leung \textit{et al} \cite{Leung Nielsen Chuang Yamamoto} consider
relaxed error correction criteria to allow for efficient correction of a known
dominant noise process. Furthermore, Cr\'{e}peau, Gottesman, and Smith
\cite{Crepeau Gottesman Smith Approximate quantum error correcting codes and
secret sharing schemes} construct approximate error correcting codes which
asymptotically correct twice as many arbitrary local errors as would be
possible under exact error correction, even though they achieve fidelity
exponentially close to $1$ in the limit of long codes.

Under the banner of \textit{approximate channel adapted error correction}, a
number of authors
\cite{
Reimpell Werner,
Fletcher Shor Win Optimum quantum error recovery using semidefinite programming,
Fletcher Shor Win Channel-Adapted quantum error correction for the amplitude dampin channel,
Reimpell Werner Audenaert,
Reimpell Thesis,
Kosut Shabani Lidar Robust Quantum error correction via convex optimization,
Taghavi Channel-Optimized quantum error correction,
Barnum Knill UhOh,
Fletcher Thesis Channel Adapted quantum error correction,
Fletcher Shor Win Structured Near-Optimal Channel-Adapted Quantum Error Correction,
Yamamoto Hara Tsumura suboptimal quantum-error-correcting procedure based on semidefinite programming,
Ng Mandayam simple approach to approximate QEC,
Beny and Oreshkov general conditions for approximate quanutm error recovery and near optimal recovery channels,
Beny and Oreshkov in preparation}
alternatively have sought to treat quantum encoding and/or recovery as
optimization problems. Mathematically, given a \textquotedblleft
noise\textquotedblright\ channel $\mathcal{N}$ one seeks an encoding operation
$\mathcal{\xi}$ and a recovery operation $\mathcal{R}$ so that the
composition
\[
\Xi=\mathcal{R}\circ\mathcal{N}\circ\mathcal{\xi}%
\]
is as close to the identity channel as possible. Measures of \textquotedblleft
closeness\textquotedblright\ to the identity include

\begin{definition}
Let $\rho$ be a mixed quantum state over a Hilbert space $\mathcal{H}$, which
may be represented as a pure quantum state $\left\vert \psi_{\rho
}\right\rangle _{\mathcal{HE}}$ of the original system entangled with an
environment $\mathcal{E}$. The \textbf{entanglement fidelity} \cite{Schumacher
sending entanglement through noisy quantum channels} of the quantum operation
$\Xi:B^{1}\left(  \mathcal{H}\right)  \rightarrow B^{1}\left(  \mathcal{H}%
\right)  $ is given by%
\begin{equation}
F_{e}\left(  \rho,\Xi\right)  =\left\langle \psi_{\rho}\right\vert \Xi\left(
\left\vert \psi_{\rho}\right\rangle \left\langle \psi_{\rho}\right\vert
\right)  \left\vert \psi_{\rho}\right\rangle .\label{entanglement fidelity}%
\end{equation}
(Note that the choice of purification does not affect the defined quantity.)
The \textbf{channel fidelity }is the entanglement fidelity when $\rho$ is
taken to be maximally-mixed.\textbf{\ }Given a collection of states $\rho_{k}
$ with a-priori probabilities $p_{k}$, one defines the \textbf{average
entanglement fidelity }\cite{Barnum Knill UhOh}\textbf{\ }%
\begin{equation}
\bar{F}_{e}\left(  \left\{  \left(  \rho_{k},p_{k}\right)  \right\}
,\Xi\right)  =%
{\displaystyle\sum}
p_{k}F_{e}\left(  \rho_{k},\Xi\right)  \text{.}%
\label{average entanglement fidelity}%
\end{equation}

\end{definition}

Following
\cite{Barnum Knill UhOh,
Fletcher Shor Win Structured Near-Optimal Channel-Adapted Quantum Error Correction,
Fletcher Shor Win Channel-Adapted quantum error correction for the amplitude dampin channel,
Fletcher Shor Win Optimum quantum error recovery using semidefinite programming}%
, we shall fix the encoding operation $\xi$ and the noise process
$\mathcal{N}$. In particular, we focus on the problem of finding an
approximately optimal quantum recovery map, or channel reversal, for the
composed map%
\[
\mathcal{A}=\mathcal{N}\circ\mathcal{\xi}\text{,}%
\]
in the sense of entanglement fidelity.

\subsubsection{Other metrics for error
recovery\label{relationship to worst case bounds}}

A number of works have considered other measures of reversibility of quantum
channels. Kretschmann, Schlingermann, and Werner \cite{Kretschmann
Schlingemann and Werner Information distrubance tradeoff and the continuity of
stinesprings representation} have obtained two-sided bounds on the CB-norm
reversibility of channels in terms of the CB-distance between the
complementary channel and a depolarizing channel. Ng and Mandayam \cite{Ng
Mandayam simple approach to approximate QEC} have employed the transpose
channel (a special case of Barnum and Knill's \cite{Barnum Knill UhOh}
reversal) to study quantum error correction using the metric of worst-case
(non-entanglement) fidelity. Yamamoto, Hara, and Tsumura \cite{Yamamoto Hara
Tsumura suboptimal quantum-error-correcting procedure based on semidefinite
programming} considered a fixed encoding operation $\mathcal{E}$ and used
semidefinite programing to find a sub-optimal channel $\mathcal{R}$ to roughly
optimize the \textquotedblleft worst-case\textquotedblright\ entanglement
fidelity
\begin{equation}
\max_{\mathcal{R}}\min_{\rho}F_{e}\left(  \rho,\mathcal{R}\circ\mathcal{N}%
\circ\mathcal{E}\right)  .\label{minimax theorem applies}%
\end{equation}
More will be said about worst-case bounds in section
\ref{section bounds Beny Oreshkov}, below.

\subsection{Quantum conditional min- and max-entropy}

The following related quantities (and their $\varepsilon$-smooth counterparts)
are of interest in quantum cryptography (for example
\cite
{Schaffner Terhal Wehner robust cryptography in the noisy quantum storage model,
Konig Wehner Wullschleger unconditional security from noisy quantum storage,
Renner extracting classical randomness in a quantum world,
Renner thesis,
Schaffner cryptography in the bounded-quantum-storage model,
Leverrier Karpov Grangier Cerf Unconditional security of continuous variable QKD,
Konig Renner sampling of min entropy relative to quantum knowledge}%
) and/or in studies of non-identically distributed and/or non-asymptotic
problems in quantum information theory (for example
\cite{Renner Wolf Wullschleger single serving channel capacity,
Konig Renner Schaffner Operational meaning of min and max entropy,
Wehner Christandl Doherty A lower bound on the dimension of a quantum system given measured data,
Berta Single shot quantum state merging,
Berta Christandl Renner a conceptually simple proof of the quantum reverse shannon theorem}%
):

\begin{definition}
Let $\rho_{AB}$ be a bipartite density operator on $\mathcal{H}_{A}%
\otimes\mathcal{H}_{B}$. The \textbf{min-entropy of }$A$\textbf{\ conditioned
on }$B $ \cite{Renner thesis, Konig Renner Schaffner Operational meaning of
min and max entropy} is defined by%
\begin{equation}
H_{\text{min}}\left(  A|B\right)  _{\rho}:=-\log_{2}\inf_{\upsilon_{B}%
}\left\{  \left.  \operatorname*{Tr}\upsilon_{B}\mathbb{~}\right\vert
~\rho_{AB}\leq%
\openone
_{A}\otimes\upsilon_{B}\right\}  ,\label{eq defining conditional minentropy}%
\end{equation}
where the infimum ranges over positive semidefinite $\upsilon_{B}$. The
\textbf{max-entropy of A conditioned on B }\cite{Renner thesis, Konig Renner
Schaffner Operational meaning of min and max entropy} is defined by%
\begin{equation}
H_{\text{max}}\left(  A|B\right)  _{\rho}:=-H_{\text{min}}\left(  A|C\right)
_{\rho}\text{,}%
\end{equation}
where the min-entropy on the RHS is evaluated for a purification $\rho_{ABC}$
of $\rho_{AB}$. The \textbf{max-information }\cite{Berta Christandl Renner a
conceptually simple proof of the quantum reverse shannon theorem} that $B$ has
about $A$ is given by%
\[
I_{\text{max}}\left(  A:B\right)  _{\rho}=H_{\text{min}}\left(  A|B\right)
_{\rho_{A}^{-1/2}\rho_{AB}\rho_{A}^{-1/2}},
\]
where $\rho_{A}=\operatorname*{Tr}_{B}\rho_{AB}$.
\end{definition}

Estimates of $H_{\text{min}}\left(  A|B\right)  _{\rho}$ are obtained as a
corollary of our estimates for maximum overlap in conjunction with the
following recent theorem:

\begin{theorem}
[K\"{o}nig, Renner, Schaffner \cite{Konig Renner Schaffner Operational meaning
of min and max entropy}]%
\label{Theorem Konig et al operational min entropy theorem}Let the Hilbert
spaces $\mathcal{H}_{A}$ and $\mathcal{H}_{B}$ be finite-dimen\-sional. Then
the min-entropy of $A$ conditioned on $B$ for the state $\rho_{AB}$ may be
expressed as%
\begin{equation}
H_{\text{min}}\left(  A|B\right)  _{\rho}=-\log\left(  \dim\left(
\mathcal{H}_{A}\right)  \sup_{\mathcal{R}}\left(  \left\langle \Phi
_{AA^{\prime}}\right\vert \mathcal{R}_{B\rightarrow A^{\prime}}\left(
\rho_{AB}\right)  \left\vert \Phi_{AA^{\prime}}\right\rangle \right)  \right)
,\label{eq min entropy as overlap}%
\end{equation}
where $\Phi_{AA^{\prime}}$ is a bipartite maximally-entangled state between
$A$ and reference system $A^{\prime}\simeq A$, and where the supremum is over
quantum operations from $B$ to $A^{\prime}$.
\end{theorem}

\subsection{Directional Iterates: An abstract approach for deriving estimates
\label{abstract section}}

The first step in proving our estimates will be to recast all of the problems
of the first section as instances of

\begin{problem}
[Maximal seminorm problem]Let $S$ be a subset of a real or complex
semidefinite inner product space $V$. Find a maximal-seminorm element of $S$.
(A semidefinite inner product has all the usual properties of an inner
product, except that one may have $\left\langle x,x\right\rangle =0$ for
nonzero $x$.)
\end{problem}

The following generalization of the iterative schemes of \cite{Jezek Rehacek
and Fiurasek Finding optimal strategies for minimum error quantum state
discrimination, Jezek Fiurasek Hradil Quantum inference of states and
processes, Reimpell Thesis, Reimpell Werner}, will prove useful for analyzing
this class of problems:

\begin{definition}
\label{def of directional iterate}An \textbf{abstract
Je\v{z}ek-\v{R}eh\'{a}\v{c}ek-Fiur\'{a}\v{s}ek-Hradil-Reimpell-Werner iterate}
of $g\in V$ is an element $g^{\left(  +\right)  }\in S$ which maximizes
$\operatorname{Re}\left\langle g^{\left(  +\right)  },g\right\rangle $. Such
$g^{\left(  +\right)  }$ will also be called \textbf{directional iterates}.
\end{definition}%

\[
\fbox{$%
\begin{array}
[c]{cc}%
{\parbox[b]{2.0297in}{\begin{center}
\includegraphics[
height=1.3396in,
width=2.0297in
]%
{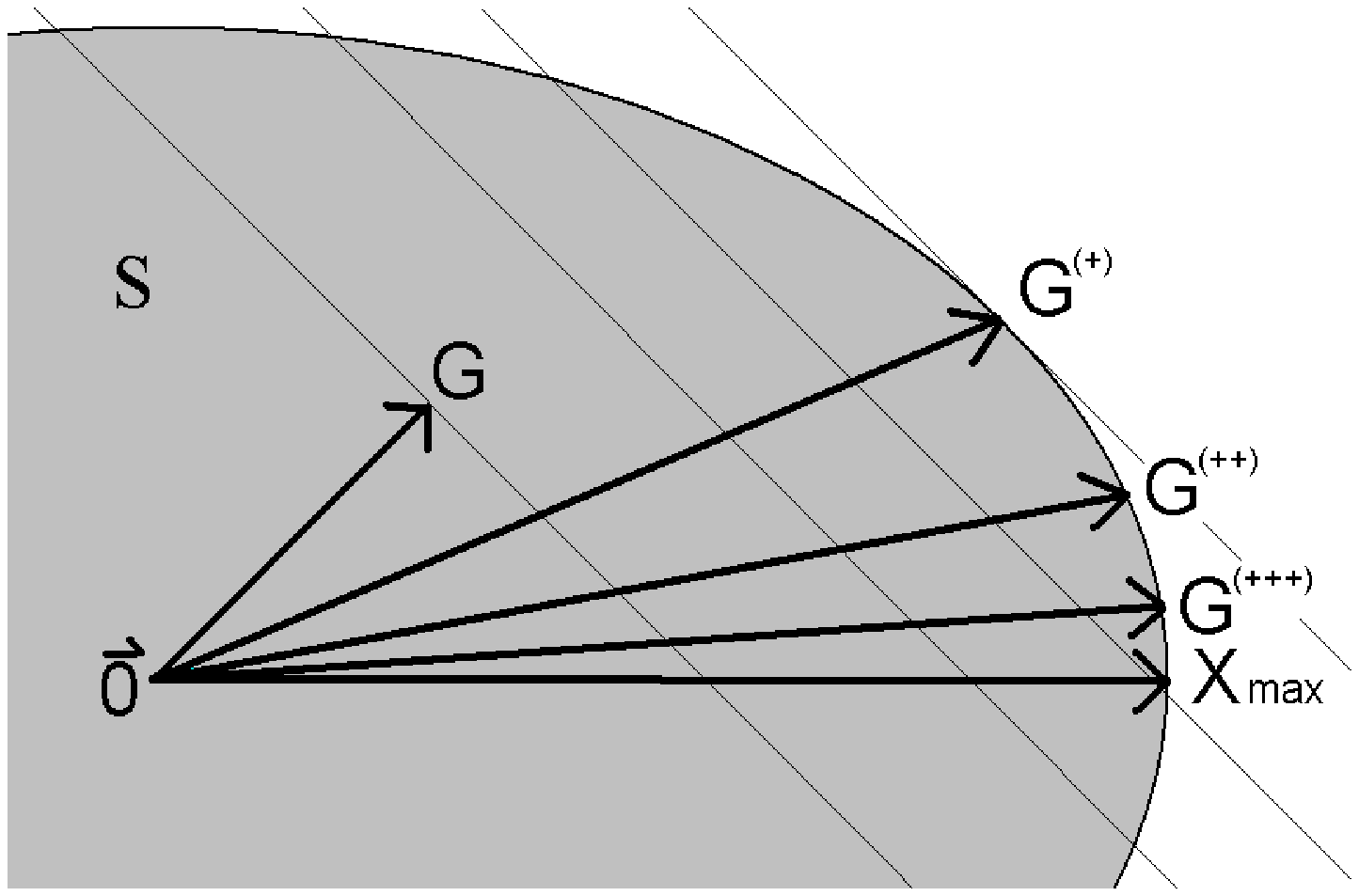}%
\\
{}%
\end{center}}}%
&
{\parbox[b]{1.7824in}{\begin{center}
\includegraphics[
height=1.3854in,
width=1.7824in
]%
{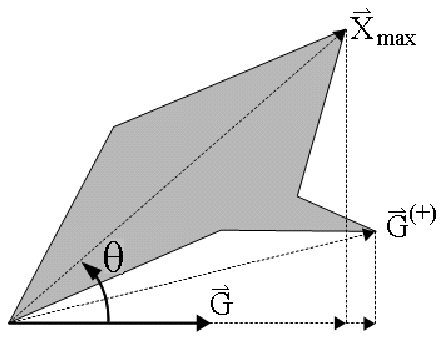}%
\\
{}%
\end{center}}}%
\\%
\begin{array}
[t]{l}%
\text{\textbf{Fig 1a:} \textit{Iterates converging on }}x_{\text{max}}%
\text{.}\\
\text{\textit{(contours drawn orthogonal to} }G\text{\textit{.)}}%
\end{array}
&
\begin{array}
[t]{l}%
\text{\textbf{Fig 1b: }\textit{Note that} }\\
\left\Vert X_{\text{max}}\right\Vert \geq\left\Vert G^{\left(  +\right)
}\right\Vert \geq\left\Vert \Pi G^{\left(  +\right)  }\right\Vert
\geq\left\Vert \Pi X_{\text{max}}\right\Vert \text{,}\\
\text{\textit{with approximate equality for reasonably small} }\theta\text{.
}\\
\text{\textit{Here }}\Pi=\left\vert G\right\rangle \left\langle G\right\vert
\text{, \textit{and }}V\text{ \textit{has real scalars. A }}\\
\text{\textit{complex variant appears as inequality} }%
\ref{key abstract estimate}\text{.}%
\end{array}
\end{array}
$}%
\]

\pagebreak

Useful properties of these iterates are given by

\begin{lemma}
[Geometric properties of directional iterates]%
\label{directional iterate lemma}Suppose that $S\subseteq V$ has a
maximal-seminorm vector $x_{\text{max}}$, and assume that each $g\in V$ admits
a directional iterate $g^{\left(  +\right)  }$. Then

\begin{enumerate}
\item[\textbf{G1.}] One has the following inequalities%
\begin{equation}
\left\Vert x_{\text{max}}\right\Vert \geq\left\Vert g^{\left(  +\right)
}\right\Vert \geq\Lambda\left(  g\right)  \geq\left\Vert x_{\text{max}%
}\right\Vert \cos\left(  \theta\right)  \text{,}\label{key abstract estimate}%
\end{equation}
where%
\begin{align}
\Lambda\left(  g\right)   &  :=\operatorname{Re}\left\langle g^{\left(
+\right)  },\frac{g}{\left\Vert g\right\Vert }\right\rangle
\label{def of lamda in abstract case}\\
\cos\theta &  :=\operatorname{Re}\frac{\left\langle g,x_{\text{max}%
}\right\rangle }{\left\Vert g\right\Vert \left\Vert x_{\text{max}}\right\Vert
}\text{.}\label{def of angle theta in abstract case}%
\end{align}

\item[\textbf{G2.}] The map $g\mapsto g^{\left(  +\right)  }$ is
seminorm-increasing on $S$. In particular, if $g\in S$ then
\begin{equation}
\left\Vert g^{\left(  +\right)  }\right\Vert ^{2}\geq\left\Vert g\right\Vert
^{2}+\left\Vert g^{\left(  +\right)  }-g\right\Vert ^{2}\text{.}%
\label{iteration increase inequality}%
\end{equation}

\end{enumerate}
\end{lemma}

\noindent\textbf{Note:} The importance of property $G1$ is this: \textit{If
one can construct a guess }$g$\textit{\ subtending a reasonably small angle
with }$x_{\text{max}}$\textit{\ then both }$\Lambda\left(  g\right)
$\textit{\ and }$\left\Vert g^{\left(  +\right)  }\right\Vert $\textit{\ are
reasonably good estimates for }$\left\Vert x_{\text{max}}\right\Vert
$\textit{. }(Note that although $\left\Vert g^{\left(  +\right)  }\right\Vert
$ is a closer approximation to $\left\Vert x_{\text{max}}\right\Vert $, in our
applications\textit{\ }$\Lambda\left(  g\right)  $ will have a much simpler expression.)

\bigskip

\begin{proof}
To prove property $G1$, note that
\[
\left\Vert x_{\text{max}}\right\Vert \geq\left\Vert g^{\left(  +\right)
}\right\Vert \geq\frac{\operatorname{Re}\left\langle g^{\left(  +\right)
},g\right\rangle }{\left\Vert g\right\Vert }\geq\frac{\operatorname{Re}%
\left\langle x_{\text{max}},g\right\rangle }{\left\Vert g\right\Vert
}=\left\Vert x_{\text{max}}\right\Vert \cos\left(  \theta\right)  \text{.}%
\]
The first inequality is trivial, the second is Schwarz's, and the third is by
the definition of $g^{\left(  +\right)  }$.

To prove property $G2$, write
\[
\left\Vert g^{\left(  +\right)  }\right\Vert ^{2}=\left\Vert g^{\left(
+\right)  }-g\right\Vert ^{2}+\left\Vert g\right\Vert ^{2}+2\operatorname{Re}%
\left(  \left\langle g^{\left(  +\right)  },g\right\rangle -\left\langle
g,g\right\rangle \right)  \text{.}%
\]
The last term on the RHS is nonnegative by the definition of $g^{\left(
+\right)  }$.
\end{proof}

\medskip

We now may set forth the following:

\bigskip

\noindent%
\begin{boxedminipage}{5.8in}%

\smallskip

\noindent\textbf{General strategy for estimating maximal seminorms:}

\begin{enumerate}
\item \textit{Find a \textquotedblleft small angle guess\textquotedblright%
}\ $g$\textit{, such that the angle defined by }$\left(
\ref{def of angle theta in abstract case}\right)  $\textit{\ is provably small
in\hspace*{0.15in}\ }\newline\textit{some approximate sense.}

\item \textit{Obtain two-sided bounds for }$\left\Vert x_{\text{max}%
}\right\Vert $ \textit{using this bound on} $\theta$ \textit{in conjunction
with }$\left(  \ref{key abstract estimate}\right)  $.

\item \textit{Make this bound explicit by computing }$g^{\left(  +\right)  }%
$\textit{\ and }$\Lambda\left(  g\right)  $\textit{.}\smallskip
\end{enumerate}%

\end{boxedminipage}%

\bigskip

By property $G2$ of Lemma \ref{directional iterate lemma}, one may have some
hope of obtaining a maximal element as the limit of repeated iteration, as
occurs in Fig. 1a. In sections \ref{section JRF iteration from POVMs}%
-\ref{section reimpell werner iteration} we review numerical schemes in the
literature which may be seen as examples of this process. (These sections may
be skimmed on first reading.)

\subsubsection{Example 1: Je\v{z}ek-\v{R}eh\'{a}\v{c}ek-Fiur\'{a}\v{s}ek
iteration for POVMs\label{section JRF iteration from POVMs}}

Je\v{z}ek, \v{R}eh\'{a}\v{c}ek, and Fiur\'{a}\v{s}ek (JRF) \cite{Jezek Rehacek
and Fiurasek Finding optimal strategies for minimum error quantum state
discrimination, Hradil et al Maximum Likelihood methods in quantum mechanics}
proposed an \textit{unproven} numerical method for computing optimal
POVMs,$^{\text{\cite{other numerical schemes}}}$ using iteration of the
mapping $M\mapsto M^{\left(  \oplus\right)  }$ given by

\begin{definition}
\label{def of JRF iteration}The
\textbf{Je\v{z}ek-\v{R}eh\'{a}\v{c}ek-Fiur\'{a}\v{s}ek (JRF) iterate of a POVM
}$M=\left\{  M_{k}\right\}  _{k\in K}$ \cite{Jezek Rehacek and Fiurasek
Finding optimal strategies for minimum error quantum state discrimination,
Hradil et al Maximum Likelihood methods in quantum mechanics} is the POVM
defined by%
\begin{equation}
M_{k}^{\left(  \oplus\right)  }=\left(
{\displaystyle\sum_{\ell\in K}}
\rho_{\ell}M_{\ell}\rho_{\ell}\right)  ^{-1/2^{+}}\rho_{k}M_{k}\rho_{k}\left(
%
{\displaystyle\sum_{\ell\in K}}
\rho_{\ell}M_{\ell}\rho_{\ell}\right)  ^{-1/2^{+}}\text{.}%
\label{eq JRF successor}%
\end{equation}
Here the negative matrix power is defined by%
\begin{equation}
A^{-s^{+}}=%
{\displaystyle\sum_{\lambda_{j}>0}}
\lambda_{j}^{-s}\Pi_{j}\label{eq defining minus 1/2 plus exponent}%
\end{equation}
for $s\geq0$ and self-adjoint $A$ with spectral decomposition $A=%
{\textstyle\sum}
\lambda_{j}\Pi_{j}$.
\end{definition}

Je\v{z}ek, \v{R}eh\'{a}\v{c}ek, and Fiur\'{a}\v{s}ek made the following:%

\newtheorem{numericalobs}[theorem]{Numerical Observation}
\begin{numericalobs}
[JRF \cite
{ Jezek Rehacek and Fiurasek Finding optimal strategies for minimum error quantum state discrimination,
Hradil et al Maximum Likelihood methods in quantum mechanics}]%
JRF iteration monotonically increases success rate:%
\[
P_{\text{succ}}\left(  M^{\left(  \oplus\right)  }\right)  \geq P_{\text{succ}%
}\left(  M\right)  .
\]
Furthermore, iteration of this map starting from $\left\{  M_{k}=%
\openone
\right\}  $ converges to an optimal measurement%
\begin{equation}
\lim_{j\rightarrow\infty}P_{\text{succ}}\left(  M^{\left(  \oplus\right)
^{j}}\right)  =P_{\text{succ}}\left(  M^{\text{opt}}\right)  \text{.}%
\end{equation}%
\end{numericalobs}%

In section \ref{section JRF iteration revisited}, JRF iteration is exhibited
as a disguised form of directional iteration. JRF's numerically-observed
monotonicity then follows immediately from property $G2$ of lemma
\ref{directional iterate lemma}.

\subsubsection{Example 2: Je\v{z}ek-Fiur\'{a}\v{s}ek-Hradil and
Reimpell-Werner iterates \label{section reimpell werner iteration}}

Je\v{z}ek, Fiur\'{a}\v{s}ek, and Hradil (JFH) \cite{Jezek Fiurasek Hradil
Quantum inference of states and processes, Hradil et al Maximum Likelihood
methods in quantum mechanics} proposed an \textit{unproven} numerical scheme
for the maximum-likelihood problem \cite{Hradil et al Maximum Likelihood
methods in quantum mechanics, tomography1, Sacchi maximum likelihood
reconstruction of CP maps, tomography2} in quantum process tomography, which
contains the maximum-overlap problem $\left(
\ref{eq defining maximum overlap}\right)  $ as a special
case.$^{\text{\cite{maximum overlap as tomography}}}$

Reimpell and Werner \cite{Reimpell Werner, Reimpell Thesis} introduced a mild
generalization of this special case of JRH's algorithm, for use in finding
maximizers of the following:

\begin{definition}
A \textbf{Reimpell-Werner functional }\cite{Reimpell Werner, Reimpell Thesis}
$\mathcal{R}\mapsto f\left(  \mathcal{R}\right)  $ is a linear functional of
linear transformations $\mathcal{R}:B^{1}\left(  \mathcal{K}\right)
\rightarrow B^{1}\left(  \mathcal{L}\right)  $ such that $f\left(
\mathcal{R}\right)  \geq0$ for all completely positive $\mathcal{R}$.
\end{definition}

Reimpell and Werner were interested in the special cases of approximate
quantum error recovery and quantum encoding in the sense of channel fidelity.
In particular, setting%
\begin{equation}
f_{\mathcal{N}}\left(  \mathcal{E},\mathcal{R}\right)  =F_{e}\left(
\openone
/\dim\mathcal{H},\mathcal{R}\circ\mathcal{N}\circ\mathcal{E}\right)  \text{,}%
\end{equation}
where $\mathcal{N}$ is a known noise map, they alternatively optimized the
encoder $\mathcal{E}$ and decoder $\mathcal{R}$ in a seesaw fashion.

By analogy with the matrix-power method \cite{Reimpell Werner, Reimpell
Thesis}, they proposed an unproven numerical method for maximizing $f\left(
\mathcal{R}\right)  $ by iteration of the following map:

\begin{definition}
\label{def of reimpell werner iterates}Let $\mathcal{L}$ and $\mathcal{K}$ be
finite-dimensional, and represent the Reimpell-Werner functional $f$ as
\begin{equation}
f\left(  \mathcal{R}\right)  =\operatorname*{Tr}_{\mathcal{LK}^{\ast}}\left(
F\mathcal{\tilde{R}}\right)  \text{,}%
\label{eq reimpell functional represented by F}%
\end{equation}
where $\mathcal{\tilde{R}}\in B^{1}\left(  \mathcal{L}\otimes\mathcal{K}%
^{\ast}\right)  $ is the Choi matrix of $\mathcal{R}$ (see Definition
\ref{definition of canonical purif of state}) and $F$ is a positive operator
on $\mathcal{LK}^{\ast}$. The \textbf{Reimpell-Werner iterate} $\mathcal{R}%
^{\oplus}$ of $\mathcal{R}$ \cite{Reimpell Werner, Reimpell Thesis} is the
quantum operation with Choi matrix%
\begin{equation}
\widetilde{\mathcal{R}}^{\oplus}=\Gamma^{-1/2^{+}}~F\tilde{R}F~\Gamma
^{-1/2^{+}}\text{,}\label{Reimpell def of iterate}%
\end{equation}
where $\Gamma:\mathcal{LK}^{\ast}\rightarrow\mathcal{LK}^{\ast}$ is given by%
\begin{equation}
\Gamma=%
\openone
_{\mathcal{L}}\otimes\operatorname*{Tr}_{\mathcal{L}}\left(  F\tilde
{R}F\right)  \text{.}\label{Reimpell def of gamma for iterate}%
\end{equation}

\end{definition}

Reimpell \cite{Reimpell Thesis} proved the monotonicity property $f\left(
\mathcal{R}^{\oplus}\right)  \geq f\left(  \mathcal{R}\right)  $ using a
clever matrix analysis argument. In particular, the optimal map $\mathcal{R}$
is a fixed point of this iteration.

In Appendix B we show that Reimpell-Werner iteration (and the special case of
restricted JRH iteration) may be viewed as directional iteration on the
corresponding space of Stinespring dilations. In particular, Reimpell's
monotonicity result is exhibited as a special case of Lemma
\ref{directional iterate lemma}.

\subsection{Relevant existing bounds, suboptimal measurements, and approximate
reversals\label{section reviewing measurements}}

\subsubsection{Quadratic measurements and Generalized Holevo-Curlander bounds
\label{section generalized holevo curlander in intro}}

\begin{definition}
Let $\mathcal{E}=\left\{  p_{k}\left\vert \psi_{k}\right\rangle \left\langle
\psi_{k}\right\vert \right\}  _{k\in K}$ be an ensemble of pure states. Then
\textbf{Holevo's pure state measurement} \textbf{\cite{Holovo Assym Opt Hyp
Test}} is given by $M_{k}=\left\vert e_{k}\right\rangle \left\langle
e_{k}\right\vert $, where%
\begin{equation}
e_{k}=e_{k}^{\text{Holevo}}:=\left(
{\displaystyle\sum}
p_{k}^{2}\left\vert \psi_{k}\right\rangle \left\langle \psi_{k}\right\vert
\right)  ^{-1/2^{+}}p_{k}\psi_{k}\text{.}\label{holevo's pure state meas}%
\end{equation}

\end{definition}

Holevo constructed this measurement using an approximate minimal principle,
and proved

\begin{theorem}
[Holevo's asymptotic optimality theorem \cite{Holovo Assym Opt Hyp Test}%
]Holevo's measurement is \linebreak asymptotically-optimal for distinguishing
pure states in the sense that for fixed probabilities $\left\{  p_{k}\right\}
$ one has%
\begin{equation}
\frac{P_{\text{fail}}\left(  \left\{  e_{k}^{\text{Holevo}}\right\}  \right)
}{P_{\text{fail}}^{\text{optimal}}}\rightarrow1\label{limit holevo asymp}%
\end{equation}
as the $\psi_{k}$ are varied so that $\left\langle \psi_{i},\psi
_{j}\right\rangle \rightarrow\delta_{ij}$. Here $P_{\text{fail}}%
=1-P_{\text{succ}}$ represents the failure rate.
\end{theorem}

A natural mixed-state generalization of Holevo's measurement is given by

\begin{definition}
\label{def of QW measurement}The \textbf{quadratically-weighted measurement}
\cite{Jezek Rehacek and Fiurasek Finding optimal strategies for minimum error
quantum state discrimination, Tyson Holevo Curlander Bounds} for
distinguishing the ensemble $\left(
\ref{a priori normed ensemble to distinguish}\right)  $ is the first
Je\v{z}ek-\v{R}eh\'{a}\v{c}ek-Fiur\'{a}\v{s}ek iterate
\begin{equation}
M_{k}^{\text{QW}}=\left(
{\displaystyle\sum\nolimits_{\ell}}
\rho_{\ell}^{2}\right)  ^{-1/2^{+}}\rho_{k}^{2}\left(
{\displaystyle\sum\nolimits_{\ell}}
\rho_{\ell}^{2}\right)  ^{-1/2^{+}}\text{.}\label{formula for MQW}%
\end{equation}

\end{definition}

\noindent\textbf{Remark: }The quadratically-weighted measurement is an example
of a Belavkin-Maslov measurement (see page 39 of \cite{Belavkin Book}).

Generalizing the pure-state results of Holevo \cite{Holovo Assym Opt Hyp Test}
and Curlander \cite{Curlander thesis MIT}, the author proved the following:

\begin{theorem}
[Generalized Holevo-Curlander bounds \cite{Tyson Holevo Curlander Bounds}%
]\label{Theorem gen Holevo Curlander bounds}One has the following bounds on
the success rate of the optimal measurement $M^{\text{opt}}$ for
distinguishing the ensemble $\mathcal{E}$ of Definition
\ref{definition containing ensemble E}:%
\begin{equation}
\Lambda^{2}\leq P_{\text{succ}}\left(  M^{\text{QW}}\right)  \leq
P_{\text{succ}}\left(  M^{\text{opt}}\right)  \leq\Lambda\text{,}%
\label{Tyson bounds}%
\end{equation}
where%
\begin{equation}
\Lambda=\operatorname*{Tr}\sqrt{%
{\displaystyle\sum}
\rho_{k}^{2}}\leq1\text{.}\label{gamma in zero 1}%
\end{equation}

\end{theorem}

\noindent\textbf{Note:} The upper bound of $\left(  \ref{Tyson bounds}\right)
$ was essentially a special case of a pre-existing bound of Ogawa and Nagaoka
\cite{Ogawa and Nagoaka Strong converse to the quantum channel coding
theorem}, which is a simple consequence of matrix monotonicity.

\subsubsection{The \textquotedblleft pretty good\textquotedblright%
\ measurement and Barnum \& Knill's distinguishability bound}

Another approximately-optimal measurement is the linearly-weighted measurement
given by

\begin{definition}
The \textbf{Belavkin-Hausladen-Wootters \textquotedblleft pretty
good\textquotedblright\ measurement} \textbf{(PGM)} \cite{Belavkin optimal
distinction of non-orthogonal quantum signals, Belavkin Optimal multiple
quantum statistical hypothesis testing, HausladenThesis, HausWootPGM} is given
by%
\begin{equation}
M_{k}^{\text{PGM}}=\left(
{\displaystyle\sum}
\rho_{\ell}\right)  ^{-1/2^{+}}\rho_{k}\left(
{\displaystyle\sum}
\rho_{\ell}\right)  ^{-1/2^{+}}\text{.}\label{PGM}%
\end{equation}

\end{definition}

A comparison of the PGM with Holevo's pure state measurement was conducted in
\cite{Tyson Error rates of quantum Belavkin measurements}. It was found that
Holevo's measurement outperforms the PGM\ for ensembles of two pure states,
and that the PGM does NOT satisfy Holevo's asymptotic optimality property
$\left(  \ref{limit holevo asymp}\right)  $.

The PGM\ is approximately-optimal for \textquotedblleft
reasonably-distinguishable\textquotedblright\ ensembles in the following
precise sense:

\begin{theorem}
[Barnum-Knill \cite{Barnum Knill UhOh}]%
\label{theorem barnum mill measurement bound}The success rate of the PGM
satisfies%
\begin{equation}
\frac{P_{\text{succ}}\left(  M^{\text{PGM}}\right)  }{P_{\text{succ}}\left(
M^{\text{opt}}\right)  }\geq P_{\text{succ}}\left(  M^{\text{opt}}\right)
\text{,}\label{BK meas estimate}%
\end{equation}
where $M^{\text{opt}}$ is an optimal measurement.
\end{theorem}

Re-expressing this inequality in terms of $P_{\text{fail}}=1-P_{\text{succ}}$,
one sees that the PGM has a failure rate within a factor of two of the
optimal:%
\begin{equation}
P_{\text{fail}}\left(  M^{\text{opt}}\right)  \leq P_{\text{fail}}\left(
M^{\text{PGM}}\right)  \leq\left(  1+P_{\text{succ}}\left(  M^{\text{opt}%
}\right)  \right)  P_{\text{fail}}\left(  M^{\text{opt}}\right)  \leq2\times
P_{\text{fail}}\left(  M^{\text{opt}}\right)  \text{.}%
\label{BK factor of two chain}%
\end{equation}
The relationship between Barnum and Knill's bound $\left(
\ref{BK meas estimate}\right)  $ and the bounds of Theorem
\ref{Theorem gen Holevo Curlander bounds} is explained by the following proposition:

\begin{proposition}
[Comparison with the Barnum-Knill bounds]Both of the lower bounds of
inequality \ref{Tyson bounds} are sufficiently tight to also satisfy Barnum
and Knill's tightness relation $\left(  \ref{BK meas estimate}\right)  $:%
\begin{equation}
\frac{P_{\text{succ}}\left(  M^{\text{QW}}\right)  }{P_{\text{succ}}\left(
M^{\text{opt}}\right)  }\geq\frac{\Lambda^{2}}{P_{\text{succ}}\left(
M^{\text{opt}}\right)  }\geq P_{\text{succ}}\left(  M^{\text{opt}}\right)
\text{.}\label{eq also as tight}%
\end{equation}
In particular, $P_{\text{fail}}\left(  M^{\text{QW}}\right)  ,$ $2\left(
1-\Lambda\right)  $, and $1-\Lambda^{2}$ all lie in the interval $\left[
P_{\text{fail}}\left(  M^{\text{opt}}\right)  ,2\times P_{\text{fail}}\left(
M^{\text{opt}}\right)  \right]  $.
\end{proposition}

\begin{proof}
Equation $\ref{eq also as tight}$ follows immediately by double application of
inequality \ref{Tyson bounds}. The claimed inclusions follow as in inequality
\ref{BK factor of two chain}, where one additionally uses the inequality
$1-\Lambda^{2}\leq2\left(  1-\Lambda\right)  $.
\end{proof}

\subsubsection{Barnum and Knill's approximate reversal map}

Generalizing the \textquotedblleft pretty good\textquotedblright%
\ measurement$^{\text{\cite{barnum knill mistake}}}$, Barnum and Knill have
constructed a reversal of an arbitrary quantum operation $\mathcal{A}%
:B^{1}\left(  \mathcal{H}\right)  \rightarrow B^{1}\left(  \mathcal{K}\right)
$ that is approximately optimal for reasonably reversible $\mathcal{A}$ in a
precise sense:

\begin{theorem}
[Barnum-Knill \cite{Barnum Knill UhOh}]%
\label{Theorem barnum knill reversal bound}Assume that the density operators
$\rho_{k}\in B^{1}\left(  \mathcal{H}\right)  $ of equation
$\ref{average entanglement fidelity}$ commute, set $\rho=%
{\textstyle\sum}
p_{k}\rho_{k}$, and let $\mathcal{A}^{\dag}:B\left(  \mathcal{K}\right)
\rightarrow B\left(  \mathcal{H}\right)  $ be the adjoint of $\mathcal{A}$
(see Def. \ref{definition of adjoint as an actual displayed defintion},
below). Then the recovery operation%
\begin{equation}
\mathcal{R}^{\text{BK}}\left(  \upsilon\right)  =\sqrt{\rho}\mathcal{A^{\dag}%
}\left(  \left(  \mathcal{A}\left(  \rho\right)  \right)  ^{-1/2^{+}}%
\upsilon\left(  \mathcal{A}\left(  \rho\right)  \right)  ^{-1/2^{+}}\right)
\sqrt{\rho}\label{Barnum Knill reversal}%
\end{equation}
is approximately optimal in the sense that%
\begin{equation}
\frac{\bar{F}_{e}\left(  \left\{  \rho_{k},p_{k}\right\}  ,\mathcal{R}%
^{\text{BK}}\circ\mathcal{N}\right)  }{\max_{\mathcal{R}}\bar{F}_{e}\left(
\left\{  \rho_{k},p_{k}\right\}  ,\mathcal{R}\circ\mathcal{N}\right)  }%
\geq\max_{\mathcal{R}}\bar{F}_{e}\left(  \left\{  \rho_{k},p_{k}\right\}
,\mathcal{R}\circ\mathcal{N}\right)  \text{,}\label{Barnum Knill estimate}%
\end{equation}
where $\bar{F}_{e}$ is the average entanglement fidelity of equation
$\ref{average entanglement fidelity}$.
\end{theorem}

A special case of eq. \ref{Barnum Knill reversal} is of recent \cite{Ng
Mandayam simple approach to approximate QEC} interest in the literature:

\begin{definition}
The \textbf{transpose channel} \cite{Petz quantum entropy and its use} is the
special case of the Barnum-Knill reversal $\mathcal{R}^{\text{BK}}$ for
maximally-mixed $\rho$.
\end{definition}

A reversal of approximately optimal entanglement fidelity which is closely
related to the quadratic measurement will be constructed in section
$\ref{section approximate channel reversals}$.

\subsubsection{The bounds of B\'{e}ny and
Oreshkov\label{section bounds Beny Oreshkov}}

Generalizing the problem of quantum error-recovery, B\'{e}ny and Oreshkov
\cite{Beny and Oreshkov general conditions for approximate quanutm error
recovery and near optimal recovery channels} have more-generally considered
channel simulation. In particular, they consider the \textquotedblleft
worst-case\textquotedblright\ entanglement fidelity
\[
\max_{\mathcal{R}}\min_{\rho}F_{\rho}\left(  \mathcal{RA},\mathcal{M}\right)
\]
with which the channel $\mathcal{A}$ may be used to simulate the channel
$\mathcal{M}$. Here one has%
\[
F_{\rho}\left(  \mathcal{N},\mathcal{M}\right)  =\min_{\rho}f\left(
\mathcal{N}_{\mathcal{H}\rightarrow\mathcal{K}}\left(  \left\vert \psi_{\rho
}\right\rangle _{\mathcal{\mathcal{HH}}_{R}}\left\langle \psi_{\rho
}\right\vert \right)  ,\mathcal{M}_{\mathcal{H}\rightarrow\mathcal{K}}\left(
\left\vert \psi_{\rho}\right\rangle _{\mathcal{\mathcal{HH}}_{R}}\left\langle
\psi_{\rho}\right\vert \right)  \right)  \text{,}%
\]
where $\psi_{\rho}$ is a purification of $\rho$ and (changing their
conventions slightly) $f\left(  \rho,\sigma\right)  =\left(
\operatorname*{Tr}\sqrt{\sqrt{\rho}\sigma\sqrt{\rho}}\right)  ^{2}$ is the
fidelity between the states $\rho$ and $\sigma$. Note that quantum error
recovery is the $\mathcal{M}=%
\openone
$ special case. Employing the min-max Theorem and a beautiful (and short!)
duality argument involving complementary channels, they obtain the following theorem:

\begin{theorem}
[B\'{e}ny-Oreshkov \cite{Beny and Oreshkov general conditions for approximate
quanutm error recovery and near optimal recovery channels}]%
\label{corrolary 3 of beny}One has the worst-case recovery bounds%
\begin{equation}
\left(  \frac{3}{4}+\frac{1}{4}\tilde{\Lambda}_{\sigma}\right)  ^{2}\geq
\max_{\mathcal{R}}\min_{\rho}\mathcal{F}_{e}\left(  \rho,\mathcal{R}%
\circ\mathcal{A}\right)  \geq\tilde{\Lambda}_{\sigma}^{2}\text{,}%
\end{equation}
where the state $\sigma$ is an adjustable parameter, $\mathcal{A}$ has Kraus
decomposition $\mathcal{A}\left(  \mu\right)  =%
{\displaystyle\sum}
E_{i}\mu E_{i}^{\dag}$, and%
\begin{equation}
\tilde{\Lambda}_{\sigma}:=\min_{\rho}\operatorname*{Tr}\sqrt{%
{\displaystyle\sum}
E_{i}\rho^{2}E_{j}^{\dag}\times\operatorname*{Tr}\left(  E_{j}\sigma
E_{i}^{\dag}\right)  }\text{.}\label{Beny lamda}%
\end{equation}
Furthermore, if $\rho$ is fixed then one obtains%
\begin{equation}
\left(  \frac{3}{4}+\frac{1}{4}F_{\rho}\left(  \mathcal{\hat{A}},S\right)
\right)  ^{2}\geq\max_{\mathcal{R}}F_{e}\left(  \rho,\mathcal{R}%
\circ\mathcal{A}\right)  \geq\left(  F_{\rho}\left(  \mathcal{\hat{A}%
},S\right)  \right)  ^{2}\text{,}%
\end{equation}
where $\mathcal{\hat{A}}$ is a channel complementary to $\mathcal{A}$ and
$S\left(  \sigma\right)  =\mathcal{\hat{A}}\left(  \rho\right)  \times
\operatorname*{Tr}\sigma$.
\end{theorem}

\noindent\textbf{Remarks:}

\begin{enumerate}
\item There is an apparent, but unexplained, relationship between our work
below and the results of B\'{e}ny-Oreshkov, which appeared in arXiv preprint
form almost-simultaneously to ours. Further will be said on this matter in
\cite{Beny and Oreshkov in preparation}. (See Theorem
\ref{theorem quadratic recovery estimates} and Proposition
\ref{prop form near beny}, below.)

\item It is important to note that in the finite-dimensional case that one has
the identity%
\begin{equation}
\max_{\mathcal{R}}\min_{\rho}F_{e}\left(  \rho,\mathcal{R}\circ\mathcal{N}%
\circ\mathcal{E}\right)  =\min_{\rho}\max_{\mathcal{R}}F_{e}\left(
\rho,\mathcal{R}\circ\mathcal{N}\circ\mathcal{E}\right)
.\label{can swap in worst case}%
\end{equation}
This follows from the min-max Theorem \cite{Ky Fan minimax}, where the
convexity of the mapping $\rho\mapsto F_{e}\left(  \rho,\mathcal{R}%
\circ\mathcal{N}\circ\mathcal{E}\right)  $ is evident from equation $1.10$ of
\cite{Fletcher Thesis Channel Adapted quantum error correction} and where one
may take the recovery $\mathcal{R}$ to range over the convex set of quantum
operations (trace non-increasing completely positive maps). \textit{In
particular, one may obtain \textquotedblleft worst-case\textquotedblright%
\ recovery bounds (albeit with unevaluated minimization over }$\rho$\textit{)
from single-instance bounds on }$F_{e}\left(  \rho,\mathcal{R}\circ
\mathcal{N}\right)  $\textit{,} which we exclusively consider below.
\end{enumerate}

\subsection{Results}

Section \ref{abstract section} has already introduced directional iteration as
an abstract method for estimating solutions of maximal-seminorm problems. This
incorporates several explicitly-defined numerical iterative schemes, including:

\begin{itemize}
\item The iteration of Je\v{z}ek, \v{R}eh\'{a}\v{c}ek, and Fiur\'{a}\v{s}ek
for computing optimal quantum measurements.

\item The iteration of M. Je\v{z}ek, J. Fiur\'{a}\v{s}ek, and Z. Hradil as
restricted to the maximum-overlap problem.

\item The iteration of Reimpell and Werner for numerically optimizing quantum
error correction (both encoding and recovery).
\end{itemize}

\noindent Defined by a minimal-principle, directional iteration monotonically
increases seminorm essentially \textit{by construction}. In particular:

\begin{itemize}
\item Je\v{z}ek, \v{R}eh\'{a}\v{c}ek, and Fiur\'{a}\v{s}ek's numerical
observation that their iteration only increases success rate is proven in
greater generality.

\item This gives a short proof of Reimpell's monotonicity Theorem (pp. 39-42
of \cite{Reimpell Thesis}) for iterative optimization of quantum error correction.
\end{itemize}

Section $\ref{section minimum error distinction a max prenorm problem}$
introduces our techniques by presenting a new proof of the generalized
Holevo-Curlander bounds (Theorem \ref{Theorem gen Holevo Curlander bounds}) on
the distinguishability of arbitrary ensembles of mixed quantum states.

In section \ref{section max overlap}, Theorem
\ref{Theorem two sided overlap estimates} gives concise two-sided bounds for
the maximum overlap problem $\left(  \ref{eq defining maximum overlap}\right)
$, in the restricted case that $M_{\mathcal{\mathcal{LH}}}$ is rank $1$.
Corollary \ref{corollary bounding min entropy} bounds the quantum conditional
min-entropy. Appendix C shows how one may apply these bounds to recover the
bounds of section
$\ref{section minimum error distinction a max prenorm problem}$.

Theorem \ref{theorem quadratic recovery estimates} of section
\ref{section approximate channel reversals} applies our overlap bounds to
estimate approximate channel reversibility in the sense of entanglement
fidelity. (The bounds of section \ref{subsection overlap estimates} more
generally allow the entangled input and output states to differ, however.) Our
channel-reversibility estimates apply to the case of channel-adapted
approximate quantum error recovery.

Section $\ref{subsection comparison with barnum knill reversal}$ compares our
reversibility estimates and approximate reversal map to those of Barnum and
Knill. Although our bounds take a particularly simple form, they are still
sufficiently accurate to satisfy the tightness relation $\left(
\ref{Barnum Knill estimate}\right)  $ satisfied by the bounds of Barnum and
Knill. The relationship between our recovery map and Barnum and Knill's is
found to be analogous to the relationship between Holevo's asymptotically
optimal measurement and the so-called \textquotedblleft pretty
good\textquotedblright\ measurement. Furthermore, our recovery operation is
found to significantly outperform the transpose channel in the case of
depolarizing noise acting on half of a maximally-entangled state.

The conclusion points out directions for future research.

\section{Notation, conventions, and mathematical background}

The reader who is only interested in minimum-error distinguishability bounds
should proceed directly to section
\ref{section minimum error distinction a max prenorm problem}, referring back
only as directed.

\begin{definition}
Let $\mathcal{H}$ and $\mathcal{K}$ be Hilbert spaces, and let $A:\mathcal{H}%
\rightarrow\mathcal{K}$ be a bounded linear operator. The \textbf{absolute
value }is $\left\vert A\right\vert =\sqrt{A^{\dag}A}$. The space $B^{1}\left(
\mathcal{H}\rightarrow\mathcal{K}\right)  $ consists of all operators of
finite \textbf{trace norm} $\left\Vert A\right\Vert _{1}:=\operatorname*{Tr}%
\left\vert A\right\vert $. The space $B^{2}\left(  \mathcal{H}\rightarrow
\mathcal{K}\right)  $ consists of all operators of finite
\textbf{Hilbert-Schmidt norm} $\left\Vert A\right\Vert _{2}:=\sqrt
{\operatorname*{Tr}A^{\dag}A}$. This space has the inner product%
\begin{equation}
\left\langle A,B\right\rangle =\operatorname*{Tr}A^{\dag}B.
\end{equation}
The space $B\left(  \mathcal{H}\rightarrow\mathcal{K}\right)  $ consists of
all operators of finite \textbf{operator norm}, given by%
\begin{equation}
\left\Vert A\right\Vert =\left\Vert A\right\Vert _{\infty}=\sup_{0\neq\psi
\in\mathcal{H}}\frac{\left\Vert A\psi\right\Vert }{\left\Vert \psi\right\Vert
}\text{.}%
\end{equation}
When $\mathcal{H}=\mathcal{K}$, these spaces will be denoted by $B\left(
\mathcal{H}\right)  $, $B^{1}\left(  \mathcal{H}\right)  $, and $B^{2}\left(
\mathcal{H}\right)  ,$ for short. An operator $A$ is a \textbf{contraction }if
$\left\Vert A\right\Vert \leq1$.
\end{definition}

It is assumed that the reader is familiar with the following trace-norm
inequalities, which may be found in \cite{Reed and Simon I}:%

\begin{align}
\left\vert \operatorname*{Tr}A\right\vert  &  \leq\left\Vert A\right\Vert
_{1}=\left\Vert A^{\dag}\right\Vert _{1}\text{ if }\mathcal{K}=\mathcal{H}%
\text{.}\label{ineq tr abs less than tr norm}\\
\left\Vert WA\right\Vert _{1} &  \leq\left\Vert W\right\Vert _{\infty}%
\times\left\Vert A\right\Vert _{1}\text{.}\label{ineq Holder I1 I infinity}%
\end{align}
Furthermore,
\begin{equation}
\sup_{\left\Vert U\right\Vert \leq1}\operatorname{Re}\left(
\operatorname*{Tr}A^{\dag}U\right)  =\left\Vert A\right\Vert _{1}%
,\label{equation sup trace A dag U}%
\end{equation}
where $A:\mathcal{H}\rightarrow\mathcal{K}$ and the supremum is over
contractions $U:\mathcal{H}\rightarrow\mathcal{K}$. It follows simply from the
singular value decomposition that $U$ is a maximizer of $\left(
\ref{equation sup trace A dag U}\right)  $ iff%
\begin{equation}
\left.  U\right\vert _{\operatorname{Ran}\left(  A^{\dag}A\right)  }=A\left(
A^{\dag}A\right)  ^{-1/2^{+}}\text{,}\label{unitary maximizer equation}%
\end{equation}
where $\left(  A^{\dag}A\right)  ^{-1/2^{+}}$ is defined by $\left(
\ref{eq defining minus 1/2 plus exponent}\right)  $.

\begin{definition}
Let $A$ be a self-adjoint operator. The \textbf{positive projection }$\Pi
_{+}\left(  A\right)  $ is the projection onto the closure of the range of the
positive part of $A$. In particular, if $A$ has spectral decomposition $A=%
{\displaystyle\sum}
\lambda_{i}\left\vert \psi_{i}\right\rangle \left\langle \psi_{i}\right\vert $
then%
\begin{equation}
\Pi_{+}\left(  A\right)  =%
{\displaystyle\sum_{\lambda_{i}>0}}
\left\vert \psi_{i}\right\rangle \left\langle \psi_{i}\right\vert
\text{.}\label{eq defining positive projection}%
\end{equation}

\end{definition}

A more thorough discussion of most of the following terms may be found in
\cite{Mike and Ike}:

\begin{definition}
\label{definition of adjoint as an actual displayed defintion} A
\textbf{quantum state} is a trace-class positive semidefinite operator $\rho$
on a Hilbert space $\mathcal{H}$. (Generally states are of unit trace,
although in section
\ref{section minimum error distinction a max prenorm problem} it will be
convenient to normalize them by a-priori probability.) The \textbf{support}
$\operatorname*{supp}\left(  A\right)  $ of the transformation $A:\mathcal{H}%
\rightarrow\mathcal{K}$ is the closure of the range of $A^{\dag}A$, or
equivalently the orthogonal complement of the null-space of $A$. A
\textbf{quantum channel} is a trace preserving completely positive map. A
\textbf{quantum operation} is a trace non-increasing completely positive map.
A linear operator $U:\mathcal{K}\rightarrow\mathcal{L}\otimes\mathcal{E}$ is a
\textbf{Stinespring dilation }\cite{steinspring dialation} of a completely
positive map $\mathcal{R}:B^{1}\left(  \mathcal{K}\right)  \rightarrow
B^{1}\left(  \mathcal{L}\right)  $ if%
\begin{equation}
\mathcal{R}\left(  \rho\right)  =\operatorname*{Tr}_{\mathcal{E}}\left(  U\rho
U^{\dag}\right) \label{defining property of a purification}%
\end{equation}
for all $\rho\in B^{1}\left(  \mathcal{K}\right)  $. The \textbf{adjoint}
$\mathcal{R}^{\dag}:B\left(  \mathcal{L}\right)  \rightarrow B\left(
\mathcal{K}\right)  $ has the defining property that%
\begin{equation}
\operatorname*{Tr}_{\mathcal{L}}\left(  X_{\mathcal{L}}\mathcal{R}%
_{\mathcal{K}\rightarrow\mathcal{L}}\left(  Y_{\mathcal{K}}\right)  \right)
=\operatorname*{Tr}_{\mathcal{K}}\left(  \mathcal{R}_{\mathcal{L}%
\rightarrow\mathcal{K}}^{\dag}\left(  X_{\mathcal{L}}\right)  Y_{\mathcal{K}%
}\right) \label{eq defining adjoint}%
\end{equation}
for $X\in B\left(  \mathcal{L}\right)  $ and $Y\in B^{1}\left(  \mathcal{K}%
\right)  .$
\end{definition}

It is important to note that if $\mathcal{R}$ and $U$ are related by $\left(
\ref{defining property of a purification}\right)  $ then $\mathcal{R}$ is a
channel iff $U$ is an isometry $\left(  U^{\dag}U=%
\openone
\right)  $, and $\mathcal{R}:B^{1}\left(  \mathcal{K}\right)  \rightarrow
B^{1}\left(  \mathcal{L}\right)  $ is a quantum operation iff $U$ is a
contraction. Furthermore, it is observed in Appendix A that each quantum
operation $\mathcal{A}$ has a canonical dilation with the canonical
environment%
\begin{equation}
\mathcal{E}=\mathcal{L}_{\mathcal{E}}^{\ast}\otimes\mathcal{K}_{\mathcal{E}%
},\label{eq canonical environment}%
\end{equation}
where $\mathcal{K}_{\mathcal{E}}$ and $\mathcal{L}_{\mathcal{E}}^{\ast}$ are
copies of $\mathcal{K}$ and the dual space of $\mathcal{L}$, respectively.

\bigskip

\textbf{Tensor product notation:} A linear operator $A:\mathcal{H}%
\rightarrow\mathcal{K}$ will often be denoted as $A_{\mathcal{H}%
\rightarrow\mathcal{K}}$, and will be identified without further comment with
any operator of the form $A\otimes%
\openone
_{\mathcal{L}}$ where $%
\openone
_{\mathcal{L}}$ is the identity operator on some other Hilbert space
$\mathcal{L}$. If $\left\vert \psi\right\rangle \in\mathcal{L}$, the
transformation $\left\vert \psi\right\rangle \otimes A:\mathcal{H}%
\rightarrow\mathcal{K}\otimes\mathcal{L}$ is defined by%
\begin{equation}
\left(  \left\vert \psi\right\rangle _{\mathcal{L}}\otimes A_{\mathcal{H}%
\rightarrow\mathcal{K}}\right)  \left\vert \phi\right\rangle _{\mathcal{H}%
}=\left\vert \psi\right\rangle _{\mathcal{L}}\otimes\left\vert A\phi
\right\rangle _{\mathcal{K}}%
\label{define tensor product of vector and transformation}%
\end{equation}
When $\mathcal{A}:B^{1}\left(  \mathcal{H}\right)  \rightarrow B^{1}\left(
\mathcal{K}\right)  $ is a quantum operation, it will often be denoted as
$\mathcal{A}_{\mathcal{H}\rightarrow\mathcal{K}}$.

\subsection{Basis-free constructions using dual spaces and double
kets\label{section basis free}}

As in a number of previous works on channel-adapted quantum error recovery
\cite{Yamamoto Hara Tsumura suboptimal quantum-error-correcting procedure
based on semidefinite programming, Fletcher Shor Win Structured Near-Optimal
Channel-Adapted Quantum Error Correction, Fletcher Thesis Channel Adapted
quantum error correction}, in Section
\ref{section approximate channel reversals} and in the Appendices A and B it
will prove convenient to treat Hilbert spaces and their duals on equal footing:

\begin{definition}
The \textbf{dual space} $\mathcal{H}^{\ast}$ of the Hilbert space
$\mathcal{H}$ is the set of linear functionals $\bar{\psi}:\mathcal{H}%
\rightarrow\mathbb{C}$ of the form
\begin{equation}
\bar{\psi}\left(  \phi\right)  :=\left\langle \psi,\phi\right\rangle
:\mathcal{H}\rightarrow\mathbb{C},
\end{equation}
where $\psi\in\mathcal{H}$. The space $\mathcal{H}^{\ast}$ is a Hilbert space
in its own right with inner product%
\begin{equation}
\left\langle \bar{\psi}_{1},\bar{\psi}_{2}\right\rangle _{\mathcal{H}^{\ast}%
}:=\overline{\left(  \left\langle \psi_{1},\psi_{2}\right\rangle
_{\mathcal{H}}\right)  }=\left\langle \psi_{2},\psi_{1}\right\rangle
_{\mathcal{H}}\text{,}%
\end{equation}
where the bar in the middle denotes complex conjugation.
\end{definition}

Use of the dual space as a Hilbert space in its own right has pleasant
computational properties which are amenable to Dirac
notation.$^{\text{\cite{footnote citing antiparticles}}}$ For example, if
$\psi\in\mathcal{H}$ has the coordinate expansion%
\begin{equation}
\psi=%
{\displaystyle\sum}
\psi_{i}\left\vert i\right\rangle _{\mathcal{H}}%
\end{equation}
then the dual vector $\bar{\psi}\in\mathcal{H}^{\ast}$ has the expansion%
\begin{equation}
\bar{\psi}=%
{\displaystyle\sum}
\bar{\psi}_{i}\left\vert \bar{\imath}\right\rangle _{\mathcal{H}^{\ast}%
}\text{,}%
\end{equation}
where the coordinates $\bar{\psi}_{i}:=\overline{\left\langle i\right\vert
}_{\mathcal{H}^{\ast}}\overline{\left\vert \psi\right\rangle }_{\mathcal{H}%
^{\ast}}$ are simply the complex conjugates of the coordinates $\psi_{i}$:%
\begin{equation}
\bar{\psi}_{i}=\overline{\left(  \psi_{i}\right)  }\text{.}%
\end{equation}

Given the linear transformation $A\in B^{2}\left(  \mathcal{H}\rightarrow
\mathcal{K}\right)  $%
\begin{equation}
A_{\mathcal{H}\rightarrow\mathcal{K}}=%
{\displaystyle\sum}
A_{ij}\left\vert i\right\rangle _{\mathcal{K}}\,\left\langle j\right\vert
_{\mathcal{H}}%
\end{equation}
one may form the \textbf{conjugate operator }$\bar{A}:\mathcal{H}^{\ast
}\rightarrow\mathcal{K}^{\ast}$, the \textbf{transpose }$A^{\text{tr}%
}:\mathcal{K}^{\ast}\rightarrow\mathcal{H}^{\ast}$, and the
\textbf{basis-free} \textbf{double ket }$\left.  \left\vert A\right\rangle
\!\right\rangle _{\mathcal{\mathcal{KH}}^{\ast}}\in\mathcal{K}\otimes
\mathcal{H}^{\ast}$ by
\begin{align}
\bar{A}_{\mathcal{H}^{\ast}\rightarrow\mathcal{K}^{\ast}}  & =%
{\displaystyle\sum}
\bar{A}_{ij}\left\vert \bar{\imath}\right\rangle _{\mathcal{K}^{\ast}%
}\,\left\langle \bar{j}\right\vert _{\mathcal{H}^{\ast}}%
\label{stupid def of bar}\\
A_{\mathcal{K}^{\ast}\rightarrow\mathcal{H}^{\ast}}^{\text{tr}}  & =%
{\displaystyle\sum}
A_{ij}\left\vert \bar{j}\right\rangle _{\mathcal{H}^{\ast}}\,\left\langle
\bar{\imath}\right\vert _{\mathcal{K}^{\ast}}\label{stupid def of transpose}\\
\left.  \left\vert A\right\rangle \!\right\rangle _{\mathcal{\mathcal{KH}%
}^{\ast}}  & =%
{\displaystyle\sum}
A_{ij}\left\vert i\right\rangle _{\mathcal{K}}\,\left\vert \bar{j}%
\right\rangle _{\mathcal{H}^{\ast}}\text{.}\label{stupid def of double ket}%
\end{align}
These equations may be replaced by basis-independent definitions, since they
are uniquely-specified by the identities $\bar{A}\bar{\phi}=\overline{\left(
A\phi\right)  }$, $A^{\text{tr}}=\bar{A}^{\dag}$, and $\left\langle
\psi_{\mathcal{K}}\right\vert \left\langle \bar{\phi}_{\mathcal{H}^{\ast}%
}\right\vert \!\left.  \left\vert A\right\rangle \!\right\rangle
_{\mathcal{\mathcal{KH}}^{\ast}}=\left\langle \psi,A\phi\right\rangle $, for
$\phi\in\mathcal{H}$ and $\psi\in\mathcal{K}$, respectively.

The \textbf{basis-free double bra}
\begin{equation}
\left\langle \!\left\langle A\right\vert \right.  _{\mathcal{\mathcal{KH}%
}^{\ast}}=%
{\displaystyle\sum}
\bar{A}_{ij}\left\langle i\right\vert _{\mathcal{K}}\,\left\langle \bar
{j}\right\vert _{\mathcal{H}^{\ast}}\label{stupid def of double bra}%
\end{equation}
denotes the linear functional on $\mathcal{K}\otimes\mathcal{H}^{\ast}$
corresponding to $\left.  \left\vert A\right\rangle \!\right\rangle $. The
\textbf{partial transpose }is the isometric extension of the mapping
$A_{\mathcal{H}\rightarrow\mathcal{K}}\otimes B_{\mathcal{L}\rightarrow
\mathcal{M}}\mapsto A_{\mathcal{K}^{\ast}\rightarrow\mathcal{H}^{\ast}%
}^{\text{tr}}\otimes B_{\mathcal{L}\rightarrow\mathcal{M}}$, i.e.%
\begin{equation}
\operatorname*{PT}_{B^{2}\left(  \mathcal{H}\rightarrow\mathcal{K}\right)
\rightarrow B^{2}\left(  \mathcal{K}^{\ast}\rightarrow\mathcal{H}^{\ast
}\right)  }\left(
{\displaystyle\sum}
X_{mkh\ell}\left\vert m_{\mathcal{M}}\right\rangle \left\vert k_{\mathcal{K}%
}\right\rangle \left\langle h_{\mathcal{H}}\right\vert \left\langle
\ell_{\mathcal{L}}\right\vert \right)  =%
{\displaystyle\sum}
X_{mkh\ell}\left\vert m_{\mathcal{M}}\right\rangle \left\vert \bar
{h}_{\mathcal{H}^{\ast}}\right\rangle \left\langle \bar{k}_{\mathcal{K}^{\ast
}}\right\vert \left\langle \ell_{\mathcal{L}}\right\vert
,\label{stupid def of partial transpose}%
\end{equation}
which maps $B^{2}\left(  \mathcal{H}\otimes\mathcal{L}\rightarrow
\mathcal{K}\otimes\mathcal{M}\right)  \rightarrow B^{2}\left(  \mathcal{K}%
^{\ast}\otimes\mathcal{L}\rightarrow\mathcal{H}^{\ast}\otimes\mathcal{M}%
\right)  $, where $\mathcal{L}$ and $\mathcal{M}$ are arbitrary Hilbert spaces.

We collect some useful identities involving basis-free double-kets:

\begin{lemma}
\begin{enumerate}
\item If $A,B:\mathcal{H}\rightarrow\mathcal{K}$ then
\begin{equation}
\left\langle \!\left\langle A,B\right\rangle \!\right\rangle _{\mathcal{KH}%
^{\ast}}=\operatorname*{Tr}A^{\dag}B\label{eq double ket is an isometry}%
\end{equation}

\item Let $A:\mathcal{K}\rightarrow\mathcal{L},$ $B:\mathcal{H}\rightarrow
\mathcal{M}$, and $C:\mathcal{H}\rightarrow\mathcal{K}$. Then%
\begin{equation}
\left(  A_{\mathcal{K}\rightarrow\mathcal{L}}\otimes\bar{B}_{\mathcal{H}%
^{\ast}\rightarrow\mathcal{M}^{\ast}}\right)  \left.  \left\vert
C\right\rangle \!\right\rangle _{\mathcal{KH}^{\ast}}=\left.  \left\vert
ACB^{\dag}\right\rangle \!\right\rangle _{\mathcal{LM}^{\ast}}\text{.}%
\label{basic eq for double kets}%
\end{equation}

\item Let $A:\mathcal{H}\rightarrow\mathcal{L}$ and $B:\mathcal{K}%
\rightarrow\mathcal{L}$. Then%
\begin{equation}
\left\langle \!\left\langle A\right\vert \right.  _{\mathcal{LH}^{\ast}}%
\times\left.  \left\vert B\right\rangle \!\right\rangle _{\mathcal{LK}^{\ast}%
}=\operatorname*{Tr}_{\mathcal{L}}\left.  \left\vert B\right\rangle
\!\right\rangle _{\mathcal{LK}^{\ast}}\left\langle \!\left\langle A\right\vert
\right.  _{\mathcal{LH}^{\ast}}=\overline{B^{\dag}A}%
\label{eq inner product out first factor}%
\end{equation}

\item Let $A:\mathcal{H}\rightarrow\mathcal{K}$ and $B:\mathcal{H}%
\rightarrow\mathcal{L}$. Then%
\begin{align}
\operatorname*{Tr}_{\mathcal{H}^{\ast}}\left.  \left\vert A\right\rangle
\!\right\rangle _{\mathcal{K\mathcal{H}}^{\ast}}~\left\langle \!\left\langle
B\right\vert \right.  _{\mathcal{\mathcal{LH}}^{\ast}} &  =A_{\mathcal{H}%
\rightarrow\mathcal{K}}B_{\mathcal{L}\rightarrow\mathcal{H}}^{^{\dag}%
}\label{eq used to construct canonical purif}\\
\operatorname*{PT}_{B^{2}\left(  \mathcal{H}^{\ast}\right)  \rightarrow
B^{2}\left(  \mathcal{H}\right)  }\left(  \left.  \left\vert A\right\rangle
\!\right\rangle _{\mathcal{\mathcal{KH}}^{\ast}}~\left\langle \!\left\langle
B\right\vert \right.  _{\mathcal{\mathcal{LH}}^{\ast}}\right)   &
=B_{\mathcal{L}\rightarrow\mathcal{H}}^{\dag}\otimes A_{\mathcal{H}%
\rightarrow\mathcal{K}}:\mathcal{H}\otimes\mathcal{L}\rightarrow
\mathcal{H}\otimes\mathcal{K}\label{Basis free PT identity}%
\end{align}

\end{enumerate}
\end{lemma}

Note that by multilinearity it is enough to check these identities for rank-1 operators.

\begin{definition}
\label{definition of canonical purif of state}The \textbf{canonical
purification }\cite{footnote basis dependent versions} of a quantum state
$\rho\in B^{1}\left(  \mathcal{H}\right)  $ is given by
\begin{equation}
\left\vert \psi_{\rho}\right\rangle =\left.  \left\vert \sqrt{\rho
}\right\rangle \!\right\rangle _{\mathcal{\mathcal{\mathcal{\mathcal{HH}%
^{\ast}}}}}.\label{eq for canonical purification}%
\end{equation}
When $\mathcal{K}$ is finite-dimensional,$^{\text{\cite{footnote infinite dim
choi}}}$ the \textbf{Choi matrix} \cite{Choi Matrix probably} of a
transformation $\mathcal{R}:B^{1}\left(  \mathcal{K}\right)  \rightarrow
B^{1}\left(  \mathcal{L}\right)  $ is given by%
\begin{equation}
\mathcal{\tilde{R}}=\mathcal{R}\left(  \left.  \left\vert
\openone
\right\rangle \!\right\rangle _{\mathcal{KK}^{\ast}}\left\langle
\!\left\langle
\openone
\right\vert \right.  \right)  \in B^{1}\left(  \mathcal{LK}^{\ast}\right)
.\label{eq for basis free choi matrix}%
\end{equation}

\end{definition}

Note that by $\left(  \ref{eq used to construct canonical purif}\right)  $ and
$\left(  \ref{eq inner product out first factor}\right)  $, $\psi_{\rho}$ has
the standard defining property%
\begin{equation}
\rho=\operatorname*{Tr}_{\mathcal{H}^{\ast}}\left\vert \psi_{\rho
}\right\rangle \left\langle \psi_{\rho}\right\vert
\label{eq psi sub rho purifies rho}%
\end{equation}
of a purification of $\rho,$ and also of $\bar{\rho}$%
\begin{equation}
\bar{\rho}=\operatorname*{Tr}_{\mathcal{H}}\left\vert \psi_{\rho}\right\rangle
\left\langle \psi_{\rho}\right\vert .\label{eq psi sub rho purifies rhobar}%
\end{equation}
In particular, if $\mathcal{H}$ is finite-dimensional then the state $\left(
\dim\mathcal{H}\right)  ^{-1/2}\left.  \left\vert
\openone
\right\rangle \!\right\rangle _{\mathcal{\mathcal{HH}^{\ast}}}$ is
maximally-entangled, and indeed the singular value decomposition of an
operator%
\[
A=%
{\displaystyle\sum}
\lambda_{i}\left\vert f_{i}\right\rangle \left\langle g_{i}\right\vert
\]
corresponds precisely to the Schmidt decomposition of its double-ket
\[
\left.  \left\vert A\right\rangle \!\right\rangle =%
{\displaystyle\sum}
\lambda_{i}\left\vert f_{i}\right\rangle \left\vert \bar{g}_{i}\right\rangle .
\]

\noindent\textbf{Remark: }A basis-free construction of the Stinespring
dilation may be found in Appendix A.

\section{Minimum-error distinction as a maximal seminorm
problem\label{section minimum error distinction a max prenorm problem}}

The minimum-error quantum detection problem of Definition
\ref{definition containing ensemble E} may be reformulated as a
maximal-seminorm problem using the identity%
\begin{equation}
P_{\text{succ}}\left(  M\right)  =\left\Vert E\right\Vert _{\mathcal{E}}%
^{2}\text{,}\label{eq making min error dectection into seminorm}%
\end{equation}
per the following definition:

\begin{definition}
\label{defintion including ensemble inner product} Let $\mathcal{E}=\left\{
\rho_{k}\right\}  _{k\in K}$ be the ensemble of Definition
$\ref{definition containing ensemble E}$. A vector of operators $E=\left\{
E_{k}:\mathcal{H}\rightarrow\mathcal{H}\right\}  _{k\in K}$ is a
\textbf{generalized measurement (GM) }\cite{Helstrom Quantum Detection and
Estimation Theory} corresponding to the POVM\ $M=\left\{  M_{k}\right\}
_{k\in K}$ if one has the decomposition%
\begin{equation}
M_{k}=E_{k}^{\dag}E_{k}\text{.}\label{EQ for POVM in terms of GM}%
\end{equation}
The $\mathcal{E}$\textbf{-semi-inner product} is defined for vectors of
operators $F=\left\{  F_{k}:\mathcal{H}\rightarrow\mathcal{H}\right\}  _{k\in
K}$ and $G=\left\{  G_{k}:\mathcal{H}\rightarrow\mathcal{H}\right\}  _{k\in K}
$ by%
\begin{equation}
\left\langle F,G\right\rangle _{\mathcal{E}}=\operatorname*{Tr}%
{\displaystyle\sum_{k\in K}}
F_{k}^{\dag}G_{k}\rho_{k}\text{.}\label{eq defining ensemble inner product}%
\end{equation}
The $\mathcal{E}$-\textbf{semi-inner product space }is the space
$V_{\mathcal{E}}=\left\{  E~\left\vert ~\left\Vert E\right\Vert _{\mathcal{E}%
}<\infty\right.  \right\}  $, on which $\left\langle \bullet,\bullet
\right\rangle _{\mathcal{E}}$ is well-defined. The set $S_{\mathcal{E}%
}\subseteq V_{\mathcal{E}}$ will denote the set of generalized measurements of
$\mathcal{E}$.
\end{definition}

\noindent\textbf{Remark:} It is important to note that if $\mathcal{E}$ is a
perfectly distinguishable ensemble of more than one element then $\left\Vert
\bullet\right\Vert _{\mathcal{E}}$ is only a seminorm. In particular, any
cyclic permutation $E^{\prime}$ of a perfectly-distinguishing generalized
measurement must satisfy $\left\Vert E^{\prime}\right\Vert =0$.

\subsection{Computation of directional
iterates\label{section JRF iteration revisited}}

Our first step is to compute directional iterates for generalized measurements:

\begin{theorem}
[Directional iteration for generalized measurements]%
\label{Theorem JRF iteration is an example of abstract JRFH iteration}Take
$S=S_{\mathcal{E}}$ and $V=V_{\mathcal{E}}$ as in Def.
\ref{def of directional iterate}. Then a directional iterate of $E\in
V_{\mathcal{E}}$ is given by%
\begin{equation}
E_{k}^{\left(  +\right)  }=E_{k}\rho_{k}\left(
{\displaystyle\sum}
\rho_{\ell}E_{\ell}^{\dag}E_{\ell}\rho_{\ell}\right)  ^{-1/2^{+}%
}\text{,\label{formula for iterate of a generalized measurement}}%
\end{equation}
where the exponent is given by equation
\ref{eq defining minus 1/2 plus exponent}. Furthermore, one has the identity%
\begin{equation}
\left\langle E^{\left(  +\right)  },E\right\rangle _{\mathcal{E}%
}=\operatorname*{Tr}\sqrt{%
{\displaystyle\sum}
\rho_{k}M_{k}\rho_{k}}\text{.}\label{formula Eplus inner product E}%
\end{equation}

\end{theorem}

\noindent\textbf{Remark: }\emph{It follows by comparison of eq. }$\left(
\ref{formula for iterate of a generalized measurement}\right)  $\emph{\ with
eq. }$\left(  \ref{eq JRF successor}\right)  $\emph{\ that the iteration}
$E\mapsto E^{\left(  +\right)  }$\emph{\ for GMs corresponds to Je\v{z}ek,
\v{R}eh\'{a}\v{c}ek, and Fiur\'{a}\v{s}ek's iteration} $M\mapsto M^{\oplus}$
\emph{for POVMs. In particular, }$P_{\text{succ}}\left(  M^{\oplus}\right)
\geq P_{\text{succ}}\left(  M\right)  $,\emph{\ as was observed numerically
in} \cite{Jezek Rehacek and Fiurasek Finding optimal strategies for minimum
error quantum state discrimination}.

\bigskip

\begin{proof}
The proof is an easy modification of that of Theorem 9 of \cite{Tyson Holevo
Curlander Bounds}, which employs the $E_{k}=%
\openone
$ special case of $\left(
\ref{formula for iterate of a generalized measurement}\right)  $. One has the
identity%
\begin{equation}
\operatorname{Re}\left\langle E,F\right\rangle _{\mathcal{E}}%
=\operatorname{Re}\operatorname*{Tr}V_{E}^{\dag}U_{F}%
\text{,\label{tr UV in Jezek iteration}}%
\end{equation}
where $V_{E},U_{F}:\mathcal{H}\rightarrow\mathcal{H}\otimes\mathbb{C}^{K}$ are
defined by%
\begin{align*}
V_{E}\psi &  =%
{\displaystyle\sum_{k\in K}}
\left(  E_{k}\rho_{k}\psi\right)  \otimes\left\vert k\right\rangle
_{\mathbb{C}^{K}}\\
U_{F}\psi &  =%
{\displaystyle\sum_{k\in K}}
\left(  F_{k}\psi\right)  \otimes\left\vert k\right\rangle _{\mathbb{C}^{K}%
}\text{,}%
\end{align*}
where $\left\vert k\right\rangle _{\mathbb{C}^{K}}$ is the standard basis of
$\mathbb{C}^{K}$. Then $F$ is a generalized measurement iff $U_{F}$ is a
contraction, with $\left\Vert U_{F}\right\Vert \leq1$. But a contraction
$U_{F}$ maximizing $\left(  \ref{tr UV in Jezek iteration}\right)  $ is
computed using equation $\ref{unitary maximizer equation}$%
\[
U_{F}\psi=V_{E}\left(  V_{E}^{\dag}V_{E}\right)  ^{-1/2^{+}}\psi=%
{\displaystyle\sum}
\left\vert k\right\rangle _{\mathbb{C}^{K}}\otimes E_{k}\rho_{k}\left(
{\displaystyle\sum}
\rho_{\ell}E_{\ell}^{\dag}E_{\ell}\rho_{\ell}\right)  ^{-1/2^{+}}\psi\text{.}%
\]
Equations $\left(  \ref{formula for iterate of a generalized measurement}%
\right)  $ and $\left(  \ref{formula Eplus inner product E}\right)  $ follow.
\end{proof}

\subsection{A \textquotedblleft small-angle\textquotedblright%
\ guess\label{section small angle guess for detection}}

In order to use Lemma \ref{directional iterate lemma} to prove
distinguishability bounds, one must construct a guess $G$ subtending a
provably-small angle with an optimal generalized measurement $E^{\text{opt}}$.
As a hint of how to proceed, consider the case that the ensemble
$\mathcal{E}=\mathcal{E}^{\text{PD}}$ is a perfectly-distinguishable ensemble,
consisting of states $\rho_{k}$ of mutually-orthogonal support. An optimal GM
is simply given by%
\begin{equation}
E_{k}^{\text{opt}}=\Pi_{+}\left(  \rho_{k}\right)  \text{,}%
\end{equation}
where the positive projection on the right was defined in $\left(
\ref{eq defining positive projection}\right)  $. Use of spectral theory may be
avoided, however, if one notes that \emph{the semi-inner product
}$\left\langle E,F\right\rangle _{\mathcal{E}}$ \emph{of equation}
$\ref{eq defining ensemble inner product}$ \emph{is sensitive to the action of
the} $E_{k}$\emph{\ and} $F_{k}$ \emph{only on the ranges of the
corresponding} $\rho_{k}$. In particular, the simplest-possible
\textquotedblleft guess\textquotedblright%
\begin{equation}
G_{k}=%
\openone
\text{ for all }k\label{G for JRF abstract}%
\end{equation}
satisfies $G\equiv E^{\text{opt}}$ mod $\left\langle \bullet,\bullet
\right\rangle _{\mathcal{E}}$, since%
\begin{equation}
\left\Vert E^{\text{opt}}-G\right\Vert _{\mathcal{E}^{\text{PD}}}%
^{2}=\operatorname*{Tr}%
{\displaystyle\sum}
\left(  \Pi_{+}\left(  \rho_{k}\right)  -%
\openone
\right)  ^{2}\rho_{k}=0\text{.}\label{eq justifying guess}%
\end{equation}

This equation suggests that the guess $\left(  \ref{G for JRF abstract}%
\right)  $ will remain appropriate for \textquotedblleft reasonably
distinguishable\textquotedblright\ ensembles. Indeed, equation
$\ref{formula for iterate of a generalized measurement}$ shows that the
iterate $G^{\left(  +\right)  }$ corresponds to the mixed-state generalization
of Holevo's asymptotically optimal measurement $\left(
\ref{holevo's pure state meas}\right)  $.$^{\text{\cite{footnote how to get
Holevo cost}}}$ Furthermore, one obtains the following \textquotedblleft small
angle\textquotedblright\ estimates:

\begin{lemma}
\label{Lemma G for meas small angle}Define $G\in V_{\mathcal{E}}$ by equation
$\ref{G for JRF abstract}$, and let $M$ be a POVM of non-zero success rate.
Then one can decompose $M_{k}=E_{k}^{\dag}E_{k}$ in such a way that
$\left\langle G,E\right\rangle _{\mathcal{E}}\in\mathbb{R}$ and%
\begin{equation}
\cos\left(  \theta\right)  :=\frac{\left\langle G,E\right\rangle
_{\mathcal{E}}}{\left\Vert G\right\Vert _{\mathcal{E}}\left\Vert E\right\Vert
_{\mathcal{E}}}\geq\sqrt{P_{\text{succ}}\left(  M\right)  }\text{.}%
\label{ineq cos theta for measurement guess}%
\end{equation}

\end{lemma}

\begin{proof}
Chose $M_{k}=\tilde{E}_{k}^{\dag}\tilde{E}_{k}$ arbitrarily. By the polar
decomposition, there exist unitary $U_{k}:\mathcal{H}\rightarrow\mathcal{H}$
so that $U_{k}\tilde{E}_{k}\rho_{k}\geq0$ for all $k$. Setting
\[
E_{k}=U_{k}\tilde{E}_{k}\text{,}%
\]
it follows from H\"{o}lder inequality's $\left(
\ref{ineq Holder I1 I infinity}\right)  $ that%
\begin{equation}
\left\langle G,E\right\rangle _{\mathcal{E}}=\operatorname*{Tr}%
{\displaystyle\sum}
E_{k}\rho_{k}=%
{\displaystyle\sum}
\left\Vert E_{k}\rho_{k}\right\Vert _{1}\geq%
{\displaystyle\sum}
\left\vert \operatorname*{Tr}E_{k}^{\dag}E_{k}\rho_{k}\right\vert
=P_{\text{succ}}\left(  M\right)  \text{.}%
\label{inequal without theta for measurement guess}%
\end{equation}
Using the fact that $\left\Vert G\right\Vert _{\mathcal{E}}=1$, the conclusion
follows by dividing both sides by $\left\Vert E\right\Vert _{\mathcal{E}%
}=\sqrt{P_{\text{succ}}\left(  M\right)  }$.
\end{proof}

\noindent\textbf{Remark: }Note that if one rescales $G$ into generalized
measurement, as is only possible when the index set is finite, then one
obtains the \textquotedblleft random guessing\textquotedblright\ measurement
$\tilde{G}_{k}=%
\openone
/\sqrt{\left\vert K\right\vert }$. The reader may therefore be surprised that
the first iterate $\tilde{G}^{\left(  +\right)  }$ (which is unaffected by
rescaling of $G$) corresponds to Holevo's asymptotically-optimal measurement,
since this guess isn't even \textquotedblleft smart\textquotedblright\ enough
to take into account the a-priori probabilities of the $\rho_{k}$!

The resolution of this paradox is that the unrescaled guess $G_{k}=%
\openone
$ does NOT correspond to random guessing: \textit{it guesses every value }$k$.
Equation $\ref{eq justifying guess}$ is therefore analogous to the statement
that the teacher who grades a multiple-choice test by means of a punched
overlay (with holes only for the correct answers) will give top marks to the
daring schoolboy who marks ALL of the ovals on his exam.

\subsection{A simple proof of the generalized Holevo-Curlander
bounds\label{section simple proof Holevo Curlander}}

We have assembled all the pieces necessary to apply Lemma
\ref{directional iterate lemma}:

\bigskip

\begin{proof}
[Proof of Theorem \ref{Theorem gen Holevo Curlander bounds}]Take
$V=V_{\mathcal{E}}$ and $S=S_{\mathcal{E}}$ to be as in Definition
\ref{defintion including ensemble inner product}, and let $G$ be given by
$\left(  \text{\ref{G for JRF abstract}}\right)  $. By Theorem
\ref{Theorem JRF iteration is an example of abstract JRFH iteration}, one has
\begin{equation}
\left(  G^{\left(  +\right)  }\right)  ^{\dag}G^{\left(  +\right)
}=M^{\text{QW}}\text{.}\label{eq QW meas is iterate of guess}%
\end{equation}
By Lemma \ref{Lemma G for meas small angle} we may decompose $M_{k}%
^{\text{opt}}=\left(  E_{k}^{\text{opt}}\right)  ^{\dag}E_{k}^{\text{opt}}$ in
such a way that the \textquotedblleft small angle\textquotedblright\ estimate
$\left(  \ref{ineq cos theta for measurement guess}\right)  $ holds in the
equivalent form%
\begin{equation}
\left\langle E^{\text{opt}},G\right\rangle _{\mathcal{E}}\geq P_{\text{succ}%
}\left(  M^{\text{opt}}\right) \label{meas ineq step from small angle}%
\end{equation}
given by $\left(  \ref{inequal without theta for measurement guess}\right)  $.
Replacing $x_{\text{max}}$ by $E_{k}^{\text{opt}}$ in $\left(
\ref{key abstract estimate}\right)  $ gives%
\begin{equation}
\sqrt{P_{\text{succ}}\left(  M^{\text{opt}}\right)  }\geq\sqrt{P_{\text{succ}%
}\left(  M^{\text{QW}}\right)  }\geq\Lambda\left(  G\right)  \geq\left\langle
E^{\text{opt}},G\right\rangle _{\mathcal{E}}\text{,}%
\label{inquality subst and take sq root}%
\end{equation}
where we have used the fact that $\left\Vert G\right\Vert _{\mathcal{E}}=1$.
But by equations $\left(  \ref{def of lamda in abstract case}\right)  $,
$\left(  \ref{formula Eplus inner product E}\right)  $, and $\left(
\ref{gamma in zero 1}\right)  $,%
\[
\Lambda\left(  G\right)  =\operatorname*{Tr}\sqrt{%
{\displaystyle\sum}
\rho_{k}^{2}}=\Lambda\text{.}%
\]
The last inequality of $\left(  \ref{Tyson bounds}\right)  $ follows by
appending $\left(  \ref{meas ineq step from small angle}\right)  $ to $\left(
\ref{inquality subst and take sq root}\right)  $. The remaining three
inequalities of $\left(  \ref{Tyson bounds}\right)  $ follow by squaring
$\left(  \ref{inquality subst and take sq root}\right)  $. The inequality
$\Lambda\leq1$ follows by $\left(  \ref{Tyson bounds}\right)  $.
\end{proof}

\section{\ Maximum overlap as a maximal-seminorm
problem\label{section max overlap}}

The maximum overlap problem of equation $\ref{eq defining maximum overlap}$
may be expressed as a maximal seminorm problem using the identity%
\begin{equation}
\operatorname*{Tr}_{\mathcal{LH}}\left(  M_{\mathcal{\mathcal{LH}}%
}R_{\mathcal{K}\rightarrow\mathcal{L}}\left(  \mu_{\mathcal{\mathcal{KH}}%
}\right)  \right)  =\left\Vert U\right\Vert _{\mu,M}^{2}\text{,}%
\label{eq MO as maximum seminorm}%
\end{equation}
where $U$ is a Stinespring dilation of $\mathcal{R}$ and the seminorm is from
the following definition:

\begin{definition}
\label{definition overlap inner product}Let $\mathcal{E}=\mathcal{L}%
_{\mathcal{E}}^{\ast}\otimes\mathcal{K}_{\mathcal{E}}$ be the the canonical
environment $\left(  \ref{eq canonical environment}\right)  $ of
$\mathcal{R}_{\mathcal{K}\rightarrow\mathcal{L}}$. For operators $A,B:$
$\mathcal{H}\rightarrow\mathcal{\mathcal{L}}\otimes\mathcal{\mathcal{E}}$, the
$\mathbf{\mu}$-$\mathbf{M}$ \textbf{semi-inner product }is defined by
\begin{equation}
\left\langle A_{\mathcal{K}\rightarrow\mathcal{LE}},B_{\mathcal{K}%
\rightarrow\mathcal{LE}}\right\rangle _{\mu,M}=\operatorname*{Tr}%
_{\mathcal{H\mathcal{LE}}}\left(  M_{\mathcal{\mathcal{LH}}}B_{\mathcal{K}%
\rightarrow\mathcal{LE}}\mu_{\mathcal{KH}}\left(  A^{\dag}\right)
_{\mathcal{LE}\rightarrow\mathcal{K}}\right)  \text{.}%
\label{eq inner product for overlap problem}%
\end{equation}
The $\mathbf{\mu}$-$\mathbf{M}$ \textbf{semi-inner product space} is the space
$V_{\mu,M}=\left\{  U:\mathcal{K}\rightarrow\mathcal{\mathcal{LE}}~\left\vert
~\left\Vert U\right\Vert _{\mu,M}<\infty\right.  \right\}  $, on which
$\left\langle \bullet,\bullet\right\rangle _{\mu,M}$ is well-defined. The
\textbf{purification ball} is the set $S=\left\{  \left.  U:\mathcal{K}%
\rightarrow\mathcal{\mathcal{LE}}~\right\vert ~\left\Vert U\right\Vert
\leq1\right\}  \subseteq V_{\mu,M}$, where $\left\Vert \bullet\right\Vert $ is
the operator-norm.
\end{definition}

\subsection{Computation of directional iterates}

As in the case of measurements, it is not difficult to compute directional iterates:

\begin{theorem}
[Directional iteration for maximum overlap is JFH iteration]%
\label{Theorem iteration for overlap purifications}Let $V=V_{\mu,M}$ and $S$
be as in Definition \ref{definition overlap inner product}. Then the operator
$U_{\mathcal{K}\rightarrow\mathcal{LE}}\in V_{\mu,M}$ has the directional
iterate%
\begin{equation}
U_{\mathcal{K}\rightarrow\mathcal{LE}}^{\left(  +\right)  }=Q\left(  Q^{\dag
}Q\right)  ^{-1/2^{+}},\label{eq for JRF iterate for overlap problem}%
\end{equation}
where%
\begin{equation}
Q_{\mathcal{K}\rightarrow\mathcal{LE}}=\operatorname*{Tr}_{\mathcal{H}}\left(
M_{\mathcal{\mathcal{LH}}}U_{\mathcal{K}\rightarrow\mathcal{\mathcal{LE}}}%
\mu_{\mathcal{KH}}\right)  \text{.}\label{formula Q for overlap problem}%
\end{equation}
Furthermore, one has%
\begin{equation}
\left\langle U,U^{\left(  +\right)  }\right\rangle _{\mu,M}=\left\Vert
Q_{\mathcal{K}\rightarrow\mathcal{LE}}\right\Vert _{1}\text{.}%
\label{inner prod between U and iterate}%
\end{equation}

\end{theorem}

\noindent\textbf{Remark: }Let $U_{\mathcal{K}\rightarrow\mathcal{LE}}$ be a
Stinespring dilation of a CP map $\mathcal{R}_{\mathcal{K}\rightarrow
\mathcal{L}}$. Then the operator $U_{\mathcal{K}\rightarrow
\mathcal{\mathcal{LE}}}^{\left(  +\right)  }$ of the above Theorem is a
dilation of the Je\v{z}ek-Fiur\'{a}\v{s}ek-Hradil iterate \cite{Jezek Fiurasek
Hradil Quantum inference of states and processes, Hradil et al Maximum
Likelihood methods in quantum mechanics}, already mentioned in section
$\ref{section reimpell werner iteration}$.

\bigskip

\begin{proof}
By cyclicity of the trace and equations $\left(
\ref{equation sup trace A dag U}\right)  $-$\left(
\ref{unitary maximizer equation}\right)  $,%
\begin{align}
\max_{W_{\mathcal{K}\rightarrow\mathcal{\mathcal{LE}}}\in S}\operatorname{Re}%
\left\langle U,W\right\rangle _{\mu,M} &  =\max_{\left\Vert W\right\Vert
\leq1}\operatorname{Re}\operatorname*{Tr}_{\mathcal{H\mathcal{LE}}}\left(
M_{\mathcal{\mathcal{LH}}}W_{\mathcal{K}\rightarrow\mathcal{LE}}%
\mu_{\mathcal{KH}}\left(  U^{\dag}\right)  _{\mathcal{LE}\rightarrow
\mathcal{K}}\right) \nonumber\\
&  =\max_{\left\Vert W\right\Vert \leq1}\operatorname{Re}\operatorname*{Tr}%
_{\mathcal{K}}\left(  \left(  Q^{\dag}\right)  _{\mathcal{\mathcal{LE}%
}\rightarrow\mathcal{K}}W_{\mathcal{K}\rightarrow\mathcal{LE}}\right)
\nonumber\\
&  =\left\Vert Q_{\mathcal{K}\rightarrow\mathcal{\mathcal{LE}}}\right\Vert
_{1}\text{,}\label{equation sup having Q in it for JFH}%
\end{align}
with maximizer $W=U^{\left(  +\right)  }$ given by $\left(
\ref{eq for JRF iterate for overlap problem}\right)  $.
\end{proof}

\subsection{The restricted maximum-overlap
problem\label{section the restricted maximum overlap}}

The remainder of this work restricts consideration to the case that%
\begin{equation}
M_{\mathcal{LH}}=\left\vert \phi\right\rangle _{\mathcal{\mathcal{LH}}%
}\left\langle \phi\right\vert \label{eq M proj in restricted MO prob}%
\end{equation}
is a rank $1$ projection, seeking to estimate%
\begin{equation}
\operatorname{MO}\left(  \mu_{\mathcal{\mathcal{KH}}},\phi
_{\mathcal{\mathcal{LH}}}\right)  :=\sup_{\mathcal{R}}\left\langle
\phi_{\mathcal{\mathcal{LH}}}\right\vert \mathcal{R}_{\mathcal{K}%
\rightarrow\mathcal{L}}\left(  \mu_{\mathcal{\mathcal{KH}}}\right)  \left\vert
\phi_{\mathcal{LH}}\right\rangle \text{,}%
\label{eq defining mu psi maximum overlap problem}%
\end{equation}
where the supremum is over quantum operations $\mathcal{R}$ from $\mathcal{K}
$ to $\mathcal{L}$. For convenience, we denote
\begin{equation}
\left\langle A_{\mathcal{K}\rightarrow\mathcal{LE}},B_{\mathcal{K}%
\rightarrow\mathcal{LE}}\right\rangle _{\mu,\phi}:=\left\langle
A,B\right\rangle _{\mu,\left\vert \phi\right\rangle \left\langle
\phi\right\vert }=\left\langle \phi_{\mathcal{\mathcal{LH}}}\right\vert
\operatorname*{Tr}_{\mathcal{\mathcal{E}}}\left(  B_{\mathcal{K}%
\rightarrow\mathcal{LE}}\mu_{\mathcal{KH}}\left(  A^{\dag}\right)
_{\mathcal{LE}\rightarrow\mathcal{K}}\right)  \left\vert \phi
_{\mathcal{\mathcal{LH}}}\right\rangle \label{reduced inner product}%
\end{equation}
and $V_{\mu,\phi}=V_{\mu,\left\vert \phi\right\rangle \left\langle
\phi\right\vert }$.

It is worth mentioning that by Theorems 1 and 2 of \cite{Konig Renner
Schaffner Operational meaning of min and max entropy} (see also equation
\ref{KRS version of optimal success} of the appendix), the minimum-error
detection problem is a special case of the restricted maximum overlap problem.
The importance of this fact for this work is as follows:\ \textit{One may use
the study of quantum measurements as a testing ground to for techniques for
the study of the maximum overlap problem and its special cases, including
quantum error recovery. }Furthermore, as we have already seen, Barnum and
Knill \cite{Barnum Knill UhOh} considered channel reversibility in the sense
of average entanglement fidelity by generalizing the \textquotedblleft pretty
good\textquotedblright\ measurement.

\subsubsection{A minor simplification}

We use the following notation for the partial traces of $\left\vert
\phi_{\mathcal{\mathcal{LH}}}\right\rangle $:%
\begin{align}
\phi_{\mathcal{L}} &  =\operatorname*{Tr}_{\mathcal{H}}\left\vert
\phi\right\rangle _{\mathcal{\mathcal{LH}}}\left\langle \phi\right\vert
\label{phi sub L defined}\\
\phi_{\mathcal{H}} &  =\operatorname*{Tr}_{\mathcal{L}}\left\vert
\phi\right\rangle _{\mathcal{\mathcal{LH}}}\left\langle \phi\right\vert
\label{phi sub H defined}%
\end{align}
Using the identity
\begin{equation}
\left\vert \phi_{\mathcal{\mathcal{LH}}}\right\rangle =\Pi_{+}\left(
\phi_{\mathcal{H}}\right)  \left\vert \phi_{\mathcal{\mathcal{LH}}%
}\right\rangle \text{,}\label{phi proj identity}%
\end{equation}
where the positive projection $\Pi_{+}\left(  \phi_{\mathcal{H}}\right)  $ is
given by equation $\ref{eq defining positive projection}$, one has the following

\begin{algorithm}
One has the identity%
\begin{equation}
\left\langle \phi_{\mathcal{\mathcal{LH}}}\right\vert \mathcal{R}%
_{\mathcal{K}\rightarrow\mathcal{L}}\left(  \mu_{\mathcal{\mathcal{KH}}%
}\right)  \left\vert \phi_{\mathcal{\mathcal{LH}}}\right\rangle =\left\langle
\phi_{\mathcal{\mathcal{LH}}}\right\vert \mathcal{R}_{\mathcal{K}%
\rightarrow\mathcal{L}}\left(  \hat{\mu}_{\mathcal{\mathcal{KH}}}\right)
\left\vert \phi_{\mathcal{\mathcal{LH}}}\right\rangle ,
\end{equation}
for any $\mathcal{R}$, where $\hat{\mu}$ is defined by%
\begin{equation}
\hat{\mu}_{\mathcal{\mathcal{KH}}}=\Pi_{+}\left(  \phi_{\mathcal{H}}\right)
\times\mu_{\mathcal{KH}}\times\Pi_{+}\left(  \phi_{\mathcal{H}}\right)
.\label{def muhat after switch to phis}%
\end{equation}

\end{algorithm}

\subsubsection{A \textquotedblleft small angle\textquotedblright\ guess}

The strategy of Sec. \ref{abstract section} calls for construction of a guess
$G\in V_{\mu,\phi}$ subtending a provably small angle with some dilation
$W^{\text{opt}}:\mathcal{K}\rightarrow\mathcal{L}\otimes\mathcal{E}$ of an
optimal overlap operation $\mathcal{R}^{\text{opt}}$. We will be most
concerned with the \textquotedblleft reasonably overlappable\textquotedblright%
\ case, for which one has the crude approximation
\begin{equation}
\sup_{\mathcal{R}}\left\langle \phi_{\mathcal{\mathcal{LH}}}\right\vert
\mathcal{R}_{\mathcal{K}\rightarrow\mathcal{L}}^{\text{opt}}\left(  \hat{\mu
}_{\mathcal{\mathcal{KH}}}\right)  \left\vert \phi_{\mathcal{\mathcal{LH}}%
}\right\rangle \approx\operatorname*{Tr}\hat{\mu}_{\mathcal{\mathcal{KH}}%
}\text{.}%
\end{equation}

Our choice of guess will therefore be motivated by the case in which exact
equality is obtained:

\begin{proposition}
[The perfectly overlappable case]%
\label{proposition perfectly overlappable case}Let $\left\vert \phi
_{\mathcal{\mathcal{LH}}}\right\rangle $ be a unit vector and let
$\mathcal{R}$ be a quantum operation. Then one has perfect overlap%
\begin{equation}
\left\langle \phi_{\mathcal{\mathcal{LH}}}\right\vert \mathcal{R}%
_{\mathcal{K}\rightarrow\mathcal{L}}\left(  \hat{\mu}_{\mathcal{\mathcal{KH}}%
}\right)  \left\vert \phi_{\mathcal{\mathcal{LH}}}\right\rangle
=\operatorname*{Tr}\hat{\mu}\label{perfect overlap}%
\end{equation}
if and only if%
\begin{equation}
\left.  \mathcal{R}_{\mathcal{L}\rightarrow\mathcal{K}}^{\dag}\left(
\left\vert \phi_{\mathcal{\mathcal{LH}}}\right\rangle \left\langle
\phi_{\mathcal{\mathcal{LH}}}\right\vert \right)  \right\vert
_{\operatorname{Ran}\left(  \hat{\mu}_{_{\mathbb{\mathcal{KH}}}}\right)  }=%
\openone
\text{,}\label{eq consequence of perfect overlap}%
\end{equation}
where the adjoint $\mathcal{R}^{\dag}$ is given by $\left(
\ref{eq defining adjoint}\right)  $.
\end{proposition}

\begin{proof}
Since
\[
\mathcal{R}^{\dag}\left(  \left\vert \phi_{\mathcal{\mathcal{LH}}%
}\right\rangle \left\langle \phi_{\mathcal{\mathcal{LH}}}\right\vert \right)
\leq\mathcal{R}^{\dag}\left(
\openone
\right)  \leq%
\openone
,
\]
the conclusion follows from $\left(  \ref{eq defining adjoint}\right)  $ and
$\left(  \ref{equation sup trace A dag U}\right)  $-$\left(
\ref{unitary maximizer equation}\right)  $.
\end{proof}

\bigskip

In section \ref{section small angle guess for detection} we saw for finite
ensembles that a properly-rescaled version of the \textquotedblleft daring
schoolboy's\textquotedblright\ guess $\left\{  G_{k}=%
\openone
\right\}  $ could be implemented by \textquotedblleft random
guessing,\textquotedblright\ without use of any measurement apparatus. This
suggests consideration of a guess $G_{\mathcal{K}\rightarrow
\mathcal{\mathcal{LE}}}$ for which the corresponding (possibly
trace-increasing) CP map
\begin{equation}
\mathcal{R}^{G}\left(  \rho\right)  =\operatorname*{Tr}_{\mathcal{E}}\left(
G\rho G^{\dag}\right) \label{improper dilation}%
\end{equation}
is independent of $\rho$.

The following Lemma shows that an analogue of equation
$\ref{eq justifying guess}$ is satisfied by a guess of this kind:

\begin{lemma}
[Construction of a guess]\label{lemma construction of overlap guess}Let
$\phi_{\mathcal{\mathcal{LH}}}$ be a unit vector and let $G_{\mathcal{K}%
\rightarrow\mathcal{\mathcal{LE}}}$ be a dilation $\left(
\ref{improper dilation}\right)  $ of the (usually trace-increasing) CP map%
\begin{equation}
\mathcal{R}_{\mathcal{K}\rightarrow\mathcal{L}}^{G}\left(  \rho\right)
:=\phi_{\mathcal{L}}^{-1^{+}}\times\operatorname*{Tr}\left(  \rho\right)
\text{,}%
\end{equation}
where $\mathcal{E}$ is the canonical environment $\left(
\ref{eq canonical environment}\right)  $. Here we use the notation introduced
in equations $\ref{eq defining minus 1/2 plus exponent}$ and
$\ref{phi sub L defined}$. Then

\begin{enumerate}
\item \label{part of guess lemma which used to say isometry}One has the
identity%
\begin{equation}
\left(  \mathcal{R}^{G}\right)  _{\mathcal{L}\rightarrow\mathcal{K}}^{\dag
}\left(  \left\vert \phi_{\mathcal{\mathcal{LH}}}\right\rangle \left\langle
\phi_{\mathcal{\mathcal{LH}}}\right\vert \right)  =\Pi_{+}\left(
\phi_{\mathcal{H}}\right)  \otimes%
\openone
_{\mathcal{K}\rightarrow\mathcal{K}}\text{.}%
\label{adjoint eq in rev lemma using double kets}%
\end{equation}

\item[2.] \label{part avoid state hyp in proof}If $\mu$ and $\phi$ are
\textquotedblleft perfectly overlappable\textquotedblright\ by a quantum
operation $\mathcal{R}=\mathcal{R}^{\text{opt}}$,\ as in equation
$\ref{perfect overlap}$, then $\mathcal{R}^{\text{opt}}$ has a dilation
$W_{\mathcal{K}\rightarrow\mathcal{L}\otimes\mathcal{E}}^{\text{opt}}$ such
that%
\begin{equation}
\left\Vert G-W^{\text{opt}}\right\Vert _{\mu,\phi}=0\text{.}%
\label{eq justify overlap guess}%
\end{equation}

\end{enumerate}
\end{lemma}

\bigskip

\noindent\textbf{Remark: }Note that the choice of dilation $G$ does not affect
the operation%
\[
\mathcal{R}^{\left(  +\right)  }\left(  \rho\right)  :=\operatorname*{Tr}%
_{\mathcal{E}}G^{\left(  +\right)  }\rho\left(  G^{\left(  +\right)  }\right)
^{\dag},
\]
since the replacement $G\rightarrow U_{\mathcal{E}\rightarrow\mathcal{E}%
}G_{\mathcal{K}\rightarrow\mathcal{\mathcal{LE}}}$, where $U_{\mathcal{E}%
\rightarrow\mathcal{E}}$ is unitary, simply induces the replacement
$G^{\left(  +\right)  }\rightarrow U_{\mathcal{E}\rightarrow\mathcal{E}%
}G^{\left(  +\right)  }$.

\begin{proof}
Equation $\ref{adjoint eq in rev lemma using double kets}$ is trivial. To
prove the equation \ref{eq justify overlap guess}, note that equations
$\ref{eq defining adjoint}$ and
$\ref{adjoint eq in rev lemma using double kets}$ imply%
\begin{equation}
G_{\mathcal{\mathcal{LE}}\rightarrow\mathcal{K}}^{\dag}\left\vert
\phi_{\mathcal{\mathcal{LH}}}\right\rangle \left\langle \phi
_{\mathcal{\mathcal{LH}}}\right\vert G_{\mathcal{K}\rightarrow
\mathcal{\mathcal{LE}}}=\left(  \mathcal{R}^{G}\right)  _{\mathcal{L}%
\rightarrow\mathcal{K}}^{\dag}\left(  \left\vert \phi_{\mathcal{\mathcal{LH}}%
}\right\rangle \left\langle \phi_{\mathcal{\mathcal{LH}}}\right\vert \right)
=\Pi_{+}\left(  \phi_{\mathcal{H}}\right)  \otimes%
\openone
_{\mathcal{K}\rightarrow\mathcal{K}}\text{.}%
\label{eq in rev lemma using the double kets}%
\end{equation}
In particular, $\left\langle \phi_{\mathcal{\mathcal{LH}}}\right\vert
G_{\mathcal{K}\rightarrow\mathcal{L}\otimes\mathcal{E}}$ restricts to an
isometry from $\operatorname*{ran}\left(  \phi_{\mathcal{H}}\right)
\otimes\mathcal{K}\supseteq\operatorname*{ran}\left(  \hat{\mu}%
_{\mathcal{\mathcal{KH}}}\right)  $ into $\mathcal{E}$. Let $W_{\mathcal{K}%
\rightarrow\mathcal{\mathcal{LE}}}^{\prime}$ be a dilation of $\mathcal{R}%
^{\text{opt}}$. By Proposition $\ref{perfect overlap}$ it similarly follows
that $\left\langle \phi_{\mathcal{\mathcal{LH}}}\right\vert W_{\mathcal{K}%
\rightarrow\mathcal{\mathcal{LE}}}^{\prime}$ is also an isometry on
$\operatorname*{ran}\left(  \hat{\mu}_{\mathcal{\mathcal{KH}}}\right)  $,
implying that there exists a unitary $X_{\mathcal{E}\rightarrow\mathcal{E}}$
such that%
\[
X_{\mathcal{E}\rightarrow\mathcal{E}}\left\langle \phi_{\mathcal{\mathcal{LH}%
}}\right\vert W_{\mathcal{K}\rightarrow\mathcal{L}\otimes\mathcal{E}}^{\prime
}=\left\langle \phi_{\mathcal{\mathcal{LH}}}\right\vert G_{\mathcal{K}%
\rightarrow\mathcal{\mathcal{L}}\otimes\mathcal{E}}%
\]
on $\operatorname*{ran}\left(  \hat{\mu}_{\mathcal{\mathcal{KH}}}\right)  $.
Equation $\ref{eq justify overlap guess}$ follows from $\left(
\ref{reduced inner product}\right)  $ by setting $W_{\mathcal{K}%
\rightarrow\mathcal{L}\otimes\mathcal{E}}^{\text{opt}}=X_{\mathcal{E}%
\rightarrow\mathcal{E}}W_{\mathcal{K}\rightarrow\mathcal{L}\otimes\mathcal{E}%
}^{\prime}$.
\end{proof}

\subsubsection{Angle Estimates}

The following estimate shows that $G_{\mathcal{K}\rightarrow\mathcal{LE}}$
remains a reasonably-good guess when $\phi$ and $\mu$ are only
reasonably-overlappable, c.f. inequality
\ref{ineq cos theta for measurement guess}:

\begin{lemma}
[Angle estimates]\label{lemma inequalitites for maximum overlap}Take
$\phi_{\mathcal{\mathcal{LH}}}$ to be a unit vector, take the CP map
$\mathcal{R}_{\mathcal{K}\rightarrow\mathcal{L}}^{G}$ and the guess
$G_{\mathcal{K}\rightarrow\mathcal{L}\otimes\mathcal{E}}$ to be as in Lemma
\ref{lemma construction of overlap guess}, and let $\mathcal{R}_{\mathcal{K}%
\rightarrow\mathcal{L}}$ be any quantum operation for which%
\[
\left\langle \phi_{\mathcal{\mathcal{LH}}}\right\vert \mathcal{R}%
_{\mathcal{K}\rightarrow\mathcal{L}}\left(  \mu_{\mathcal{\mathcal{KH}}%
}\right)  \left\vert \phi_{\mathcal{\mathcal{LH}}}\right\rangle >0\text{.}%
\]
Then $\mathcal{R}$ has a Stinespring dilation $W_{\mathcal{K}\rightarrow
\mathcal{\mathcal{LE}}}$ such that $\left\langle W,G\right\rangle _{\mu,\phi
}\in\mathbb{R}$ and%
\begin{equation}
\cos\left(  \theta\right)  :=\frac{\left\langle W,G\right\rangle _{\mu,\phi}%
}{\left\Vert W\right\Vert _{\mu,\phi}\left\Vert G\right\Vert _{\mu,\phi}}%
\geq\sqrt{\frac{\left\langle \phi_{\mathcal{\mathcal{LH}}}\right\vert
\mathcal{R}_{\mathcal{K}\rightarrow\mathcal{L}}\left(  \hat{\mu}%
_{\mathcal{\mathcal{KH}}}\right)  \left\vert \phi_{\mathcal{LH}}\right\rangle
}{\left\Vert \mathcal{R}_{\mathcal{L}\rightarrow\mathcal{K}}^{\dag}\left(
\left\vert \phi_{\mathcal{\mathcal{LH}}}\right\rangle \left\langle
\phi_{\mathcal{\mathcal{LH}}}\right\vert \right)  \right\Vert _{\infty
}\operatorname*{Tr}\left(  \hat{\mu}_{\mathcal{KH}}\right)  }}\text{.}%
\label{angle estimate for projective overlap guess}%
\end{equation}
Here $\hat{\mu}_{\mathcal{\mathcal{KH}}}$ is given by $\left(
\ref{def muhat after switch to phis}\right)  $ and the adjoint $\mathcal{R}%
^{\dag}$ is from equation \ref{eq defining adjoint}. Furthermore, one has the
identities%
\begin{align}
\left\Vert G\right\Vert _{\mu,\phi}^{2} &  =\operatorname*{Tr}\left(  \hat
{\mu}_{\mathcal{KH}}\right) \label{zero order guess has norm 1}\\
\left\langle G^{\left(  +\right)  },G\right\rangle _{\mu,\phi} &
=\operatorname*{Tr}_{\mathcal{K}}\sqrt{\left\langle \phi\right\vert
_{\mathcal{\mathcal{LH}}}\hat{\mu}_{\mathcal{KH}}^{2}\left\vert \phi
\right\rangle _{\mathcal{\mathcal{LH}}}}%
\label{inner prod of G and Gplus for zero order guess}%
\end{align}
where $G^{\left(  +\right)  }$ is the iterate of $G$ given by Theorem
\ref{Theorem iteration for overlap purifications}.
\end{lemma}

\noindent\textbf{Remark: }Note that $\cos\left(  \theta\right)  =1$ if perfect
overlap $\left\langle \phi\right\vert _{\mathcal{\mathcal{LH}}}\mathcal{R}%
_{\mathcal{K}\rightarrow\mathcal{L}}\left(  \hat{\mu}_{\mathcal{\mathcal{KH}}%
}\right)  \left\vert \phi\right\rangle _{\mathcal{LH}}=\operatorname*{Tr}%
\hat{\mu}_{\mathcal{\mathcal{KH}}}$ is achieved.

\bigskip

\begin{proof}
Equation $\ref{zero order guess has norm 1}$ follows from part
\ref{part of guess lemma which used to say isometry} of Lemma
\ref{lemma construction of overlap guess}:%
\[
\left\Vert G\right\Vert _{\mu,\phi}^{2}=\left\langle \phi
_{\mathcal{\mathcal{LH}}}\right\vert \mathcal{R}_{\mathcal{K}\rightarrow
\mathcal{L}}^{G}\left(  \hat{\mu}_{\mathcal{\mathcal{KH}}}\right)  \left\vert
\phi_{\mathcal{\mathcal{LH}}}\right\rangle =\operatorname*{Tr}\left(  \left(
\mathcal{R}^{G}\right)  _{\mathcal{L}\rightarrow\mathcal{K}}^{\dag}\left(
\left\vert \phi_{\mathcal{\mathcal{LH}}}\right\rangle \left\langle
\phi_{\mathcal{\mathcal{LH}}}\right\vert \right)  \hat{\mu}%
_{\mathcal{\mathcal{KH}}}\right)  =\operatorname*{Tr}\hat{\mu}%
_{\mathcal{\mathcal{KH}}}%
\]

To prove the angle estimate $\left(
\ref{angle estimate for projective overlap guess}\right)  $, note that one has
the identity%
\begin{equation}
\left\langle W_{\mathcal{K}\rightarrow\mathcal{\mathcal{LE}}},G\right\rangle
_{\mu,\phi}=\operatorname*{Tr}_{\mathcal{E}}P_{\mathcal{E}\rightarrow
\mathcal{E}},
\end{equation}
where
\begin{equation}
P_{\mathcal{E}\rightarrow\mathcal{E}}=\left\langle \phi_{\mathcal{\mathcal{LH}%
}}\right\vert G_{\mathcal{K}\rightarrow\mathcal{\mathcal{LE}}}\mu
_{\mathcal{\mathcal{KH}}}\left(  W^{\dag}\right)  _{\mathcal{\mathcal{LE}%
}\rightarrow\mathcal{K}}\left\vert \phi_{\mathcal{\mathcal{LH}}}\right\rangle
\text{.}\label{pos op for overlap}%
\end{equation}
Starting with any dilation $W_{\mathcal{K}\rightarrow\mathcal{\mathcal{LE}}} $
of $\mathcal{R}$, we may assure that $P_{\mathcal{E}\rightarrow\mathcal{E}}$
is positive semidefinite by a replacement%
\[
W_{\mathcal{K}\rightarrow\mathcal{\mathcal{LE}}}\rightarrow X_{\mathcal{E}%
\rightarrow\mathcal{E}}W_{\mathcal{K}\rightarrow\mathcal{\mathcal{LE}}},
\]
where the unitary operator $X_{\mathcal{E}\rightarrow\mathcal{E}}$ comes from
the polar decomposition of $P_{\mathcal{E}\rightarrow\mathcal{E}}$. It follows
that the LHS\ of $\left(  \ref{angle estimate for projective overlap guess}%
\right)  $ is real and maximized over the choice of dilation of $\mathcal{R}$.
We claim that there exists a an operator $Z:\mathcal{E}\rightarrow\mathcal{E}$
such that
\begin{align}
Z_{\mathcal{E}\rightarrow\mathcal{E}}P_{\mathcal{E}\rightarrow\mathcal{E}} &
=\left\langle \phi_{\mathcal{\mathcal{LH}}}\right\vert W_{\mathcal{K}%
\rightarrow\mathcal{\mathcal{LE}}}\mu_{\mathcal{\mathcal{KH}}}\left(  W^{\dag
}\right)  _{\mathcal{\mathcal{LE}}\rightarrow\mathcal{K}}\left\vert
\phi_{\mathcal{\mathcal{LH}}}\right\rangle \label{ZCl}\\
\left\Vert Z_{\mathcal{E}\rightarrow\mathcal{E}}\right\Vert _{\infty} &
=\left\Vert \mathcal{R}_{\mathcal{L}^{\prime}\rightarrow\mathcal{K}}^{\dag
}\left(  \left\vert \phi_{\mathcal{L}^{\prime}\mathcal{H}}\right\rangle
\left\langle \phi_{\mathcal{\mathcal{L}}^{\prime}\mathcal{\mathcal{H}}%
}\right\vert \right)  \right\Vert _{\infty}^{1/2}\label{ZC2}%
\end{align}
Assuming this claim, H\"{o}lder's inequality $\left(
\ref{ineq Holder I1 I infinity}\right)  $ implies that
\begin{align}
\left\langle W_{\mathcal{K}\rightarrow\mathcal{\mathcal{LE}}},G\right\rangle
_{\mu,\phi} &  =\operatorname*{Tr}_{\mathcal{E}}P_{\mathcal{E}\rightarrow
\mathcal{E}}=\left\Vert P_{\mathcal{E}\rightarrow\mathcal{E}}\right\Vert
_{1}\nonumber\\
&  \geq\frac{1}{\left\Vert Z\right\Vert }\left\vert \operatorname*{Tr}%
_{\mathcal{E}}Z_{\mathcal{E}\rightarrow\mathcal{E}}P_{\mathcal{E}%
\rightarrow\mathcal{E}}\right\vert \nonumber\\
&  =\frac{\left\langle \phi_{\mathcal{\mathcal{LH}}}\right\vert \mathcal{R}%
_{\mathcal{K}\rightarrow\mathcal{L}}\left(  \mu_{\mathcal{\mathcal{KH}}%
}\right)  \left\vert \phi_{\mathcal{\mathcal{LH}}}\right\rangle }{\left\Vert
\mathcal{R}_{\mathcal{L}\rightarrow\mathcal{K}}^{\dag}\left(  \left\vert
\phi_{\mathcal{LH}}\right\rangle \left\langle \phi_{\mathcal{\mathcal{LH}}%
}\right\vert \right)  \right\Vert _{\infty}^{1/2}}\text{.}%
\label{reduced angle estimate for maximum overlap}%
\end{align}
The angle estimate $\left(  \ref{angle estimate for projective overlap guess}%
\right)  $ follows by dividing both sides by%
\[
\left\Vert W\right\Vert _{\mu,\phi}=\sqrt{\left\langle \phi
_{\mathcal{\mathcal{LH}}}\right\vert \mathcal{R}_{\mathcal{K}\rightarrow
\mathcal{L}}\left(  \mu_{\mathcal{\mathcal{KH}}}\right)  \left\vert
\phi_{\mathcal{\mathcal{LH}}}\right\rangle }%
\]
and by the square root of equation $\ref{zero order guess has norm 1}$.

To prove the claims $\left(  \ref{ZCl}\right)  -\left(  \ref{ZC2}\right)  $,
define
\begin{equation}
Z_{\mathcal{E}\rightarrow\mathcal{E}}=\left\langle \phi_{\mathcal{\mathcal{L}%
}^{\prime}\mathcal{\mathcal{H}}}\right\vert W_{\mathcal{K}\rightarrow
\mathcal{\mathcal{L}}^{\prime}\mathcal{\mathcal{E}}}\left(  G^{\dag}\right)
_{\mathcal{\mathcal{LE}}\rightarrow\mathcal{K}}\left\vert \phi
_{\mathcal{\mathcal{LH}}}\right\rangle ,\label{define mystery z}%
\end{equation}
where $\mathcal{L}^{\prime}$ is a copy of $\mathcal{L}$. Then%
\begin{align*}
\left\Vert Z_{\mathcal{E}\rightarrow\mathcal{E}}\right\Vert _{\infty} &
=\left\Vert \left(  Z^{\dag}Z\right)  _{\mathcal{E}\rightarrow\mathcal{E}%
}\right\Vert _{\infty}^{1/2}=\left\Vert \left\langle \phi
_{\mathcal{\mathcal{LH}}}\right\vert G_{\mathcal{K}\rightarrow
\mathcal{\mathcal{LE}}}W_{\mathcal{L}^{\prime}\mathcal{E}\rightarrow
\mathcal{K}}^{\dag}\left\vert \phi_{\mathcal{L}^{\prime}\mathcal{H}%
}\right\rangle \left\langle \phi_{\mathcal{\mathcal{L}}^{\prime}%
\mathcal{\mathcal{H}}}\right\vert W_{\mathcal{K}\rightarrow
\mathcal{\mathcal{L}}^{\prime}\mathcal{\mathcal{E}}}\left(  G^{\dag}\right)
_{\mathcal{\mathcal{LE}}\rightarrow\mathcal{K}}\left\vert \phi
_{\mathcal{\mathcal{LH}}}\right\rangle \right\Vert _{\infty}^{1/2}\\
&  =\left\Vert \left\langle \phi_{\mathcal{\mathcal{LH}}}\right\vert
G_{\mathcal{K}\rightarrow\mathcal{\mathcal{LE}}}\mathcal{R}_{\mathcal{L}%
^{\prime}\rightarrow\mathcal{K}}^{\dag}\left(  \left\vert \phi_{\mathcal{L}%
^{\prime}\mathcal{H}}\right\rangle \left\langle \phi_{\mathcal{\mathcal{L}%
}^{\prime}\mathcal{\mathcal{H}}}\right\vert \right)  \left(  G^{\dag}\right)
_{\mathcal{\mathcal{LE}}\rightarrow\mathcal{K}}\left\vert \phi
_{\mathcal{\mathcal{LH}}}\right\rangle \right\Vert _{\infty}^{1/2}\text{.}%
\end{align*}
But as in the proof of Lemma \ref{lemma construction of overlap guess}, the
operator $\left\langle \phi_{\mathcal{\mathcal{LH}}}\right\vert G_{\mathcal{K}%
\rightarrow\mathcal{\mathcal{LE}}}$ is an isometry from $\mathcal{K}%
\otimes\operatorname*{ran}\left(  \phi_{\mathcal{H}}\right)  $ into
$\mathcal{E}$, proving $\left(  \ref{ZC2}\right)  $.

To prove $\left(  \ref{ZCl}\right)  $, note that equations
\ref{define mystery z}, $\ref{eq in rev lemma using the double kets},$ \&
\ref{phi proj identity} imply that
\begin{equation}
Z_{\mathcal{E}\rightarrow\mathcal{E}}\left\langle \phi_{\mathcal{\mathcal{LH}%
}}\right\vert G_{\mathcal{K}\rightarrow\mathcal{\mathcal{LE}}}=\left\langle
\phi_{\mathcal{\mathcal{L}}^{\prime}\mathcal{\mathcal{H}}}\right\vert
W_{\mathcal{K}\rightarrow\mathcal{\mathcal{L}}^{\prime}\mathcal{\mathcal{E}}%
}\Pi_{+}\left(  \phi_{\mathcal{H}}\right)  =\left\langle \phi
_{\mathcal{\mathcal{L}}^{\prime}\mathcal{\mathcal{H}}}\right\vert
W_{\mathcal{K}\rightarrow\mathcal{\mathcal{L}}^{\prime}\mathcal{\mathcal{E}}%
}\text{.}\label{strange z identity}%
\end{equation}
Equation $\ref{ZCl}$ follows, since by $\left(  \ref{pos op for overlap}%
\right)  $%
\begin{align*}
Z_{\mathcal{E}\rightarrow\mathcal{E}}P_{\mathcal{E}\rightarrow\mathcal{E}} &
=Z_{\mathcal{E}\rightarrow\mathcal{E}}\left\langle \phi_{\mathcal{\mathcal{L}%
}^{\prime}\mathcal{\mathcal{H}}}\right\vert G_{\mathcal{K}\rightarrow
\mathcal{\mathcal{L}}^{\prime}\mathcal{\mathcal{E}}}\,\mu
_{\mathcal{\mathcal{KH}}}\left(  W^{\dag}\right)  _{\mathcal{\mathcal{LE}%
}\rightarrow\mathcal{K}}\left\vert \phi_{\mathcal{\mathcal{LH}}}\right\rangle
\\
&  =\left\langle \phi_{\mathcal{\mathcal{LH}}}\right\vert W_{\mathcal{K}%
\rightarrow\mathcal{\mathcal{LE}}}\mu_{\mathcal{\mathcal{KH}}}\left(  W^{\dag
}\right)  _{\mathcal{\mathcal{LE}}\rightarrow\mathcal{K}}\left\vert
\phi_{\mathcal{\mathcal{LH}}}\right\rangle \text{.}%
\end{align*}

It remains to prove equation
$\ref{inner prod of G and Gplus for zero order guess}$. By equation
$\ref{inner prod between U and iterate}$,%
\begin{equation}
\left\langle G^{\left(  +\right)  },G\right\rangle _{\mu,\phi}=\left\Vert
Q_{\mathcal{K}\rightarrow\mathcal{LE}}\right\Vert _{1},
\end{equation}
where by equations $\ref{formula Q for overlap problem}$ and
$\ref{eq M proj in restricted MO prob}$,%
\begin{align}
Q_{\mathcal{K}\rightarrow\mathcal{LE}} &  =\operatorname*{Tr}_{\mathcal{H}%
}\left(  \left\vert \phi_{\mathcal{\mathcal{LH}}}\right\rangle \left\langle
\phi_{\mathcal{\mathcal{L}}^{\prime}\mathcal{\mathcal{H}}}\right\vert
~G_{\mathcal{K}\rightarrow\mathcal{\mathcal{L}}^{\prime}\mathcal{\mathcal{E}}%
}\mu_{\mathcal{KH}}\right) \nonumber\\
&  =\left\langle \phi_{\mathcal{\mathcal{L}}^{\prime}\mathcal{\mathcal{H}}%
}\right\vert ~G_{\mathcal{K}\rightarrow\mathcal{\mathcal{L}}^{\prime
}\mathcal{\mathcal{E}}}~\mu_{\mathcal{KH}}\left\vert \phi
_{\mathcal{\mathcal{LH}}}\right\rangle \text{,}%
\label{simplified Q for overlap problem}%
\end{align}
where $\mathcal{L}^{\prime}$ is a copy of $\mathcal{L}$. It follows by
equations $\ref{eq in rev lemma using the double kets}$ and
$\ref{phi proj identity}$ that%
\begin{align}
Q^{\dag}Q &  =\left\langle \phi_{\mathcal{\mathcal{LH}}}\right\vert
\mu_{\mathcal{\mathcal{\mathcal{\mathcal{\mathcal{\mathcal{K}}}}}H}}\chi
_{+}\left(  \phi_{\mathcal{H}}\right)  \mu
_{\mathcal{\mathcal{\mathcal{\mathcal{\mathcal{\mathcal{K}}}}}H}}\left\vert
\phi_{\mathcal{\mathcal{LH}}}\right\rangle \nonumber\\
&  =\left\langle \phi_{\mathcal{\mathcal{LH}}}\right\vert \hat{\mu
}_{\mathcal{\mathcal{\mathcal{\mathcal{\mathcal{\mathcal{K}}}}}H}}%
^{2}\left\vert \phi_{\mathcal{\mathcal{LH}}}\right\rangle
.\label{Q dag Q formula}%
\end{align}
The conclusion follows.
\end{proof}

\subsubsection{The \textquotedblleft Quadratic Overlapper\textquotedblright}

The following operation is analogous to the quadratically-weighted measurement:

\begin{definition}
The \textbf{\textquotedblleft quadratic overlapper\textquotedblright\ }is the
operation $\mathcal{R}^{\text{QO}}:B^{1}\left(  \mathcal{K}\right)
\rightarrow B^{1}\left(  \mathcal{L}\right)  $ defined by%
\begin{equation}
\mathcal{R}^{\text{QO}}\left(  \upsilon_{\mathcal{K}}\right)
=\operatorname*{Tr}_{\mathcal{E}}G_{\mathcal{K}\rightarrow
\mathcal{\mathcal{LE}}}^{\left(  +\right)  }\upsilon_{\mathcal{K}}\left(
G_{\mathcal{K}\rightarrow\mathcal{\mathcal{LE}}}^{\left(  +\right)  }\right)
^{\dag}\text{,}\label{eq def PGO}%
\end{equation}
where $G_{\mathcal{K}\rightarrow\mathcal{\mathcal{LE}}}\in V_{\mu,\phi}$ is
the \textquotedblleft small-angle guess\textquotedblright\ of Lemma
\ref{lemma construction of overlap guess}, with directional iterate
$G_{\mathcal{K}\rightarrow\mathcal{\mathcal{LE}}}^{\left(  +\right)  }$ given
by equations $\ref{eq for JRF iterate for overlap problem}$%
-$\ref{formula Q for overlap problem}$ of Theorem
$\ref{Theorem iteration for overlap purifications}$.
\end{definition}

Alternatively, one may express $\mathcal{R}^{\text{QO}}$ using

\begin{theorem}
[Computation of the quadratic overlapper]%
\label{theorem compute quadratic overlapper}One has the formula%
\begin{equation}
\mathcal{R}_{\mathcal{K}\rightarrow\mathcal{L}}^{\text{QO}}\left(
\upsilon_{\mathcal{K}}\right)  =\operatorname*{Tr}_{\mathcal{KH}}\left(
\hat{\mu}_{\mathcal{KH}}^{2}\left(  \left(  Y^{-1/2^{+}}\upsilon Y^{-1/2^{+}%
}\right)  _{\mathcal{K}\rightarrow\mathcal{K}}\otimes\left\vert \phi
\right\rangle _{\mathcal{LH}}\left\langle \phi\right\vert \right)  \right)
,\label{eq for PGO}%
\end{equation}
where%
\begin{equation}
Y_{\mathcal{K}\rightarrow\mathcal{K}}=\left\langle \phi_{\mathcal{\mathcal{LH}%
}}\right\vert \hat{\mu}_{\mathcal{KH}}^{2}\left\vert \phi
_{\mathcal{\mathcal{LH}}}\right\rangle \text{,}\label{M in maximum overlap}%
\end{equation}
and $\hat{\mu}_{\mathcal{\mathcal{KH}}}$ is given by $\left(
\ref{def muhat after switch to phis}\right)  $.
\end{theorem}

\begin{proof}
By equation $\ref{eq for JRF iterate for overlap problem}$, the guess $G$ has
the iterate%
\[
G_{\mathcal{K}\rightarrow\mathcal{\mathcal{LE}}}^{\left(  +\right)  }=Q\left(
Q^{^{\dag}}Q\right)  ^{-1/2^{+}},
\]
where the operator $Q$ is defined by $\left(
\ref{formula Q for overlap problem}\right)  $. By equations $\left(
\ref{simplified Q for overlap problem}\right)  $-$\left(
\ref{Q dag Q formula}\right)  $, one has%
\begin{align*}
Q_{\mathcal{K}\rightarrow\mathcal{\mathcal{LE}}} &  =\left\langle
\phi_{\mathcal{\mathcal{L}}^{\prime}\mathcal{\mathcal{H}}}\right\vert
~G_{\mathcal{K}\rightarrow\mathcal{\mathcal{L}}^{\prime}\mathcal{\mathcal{E}}%
}~\mu_{\mathcal{KH}}\left\vert \phi_{\mathcal{\mathcal{LH}}}\right\rangle \\
Q^{\dag}Q &  =\left\langle \phi_{\mathcal{\mathcal{LH}}}\right\vert \hat{\mu
}_{\mathcal{\mathcal{\mathcal{\mathcal{\mathcal{\mathcal{K}}}}}H}%
\rightarrow\mathcal{\mathcal{KH}}}^{2}\left\vert \phi_{\mathcal{\mathcal{LH}}%
}\right\rangle =Y_{\mathcal{K}\rightarrow\mathcal{K}}.
\end{align*}
Setting%
\begin{equation}
\tilde{\upsilon}_{\mathcal{K}\rightarrow\mathcal{K}}=Y^{-1/2^{+}}\upsilon
Y^{-1/2^{+}},
\end{equation}
it follows that%
\begin{align}
&  \operatorname*{Tr}_{\mathcal{E}}G_{\mathcal{K}\rightarrow
\mathcal{\mathcal{LE}}}^{\left(  +\right)  }~\upsilon_{\mathcal{\mathcal{K}}%
}~\left(  G^{\left(  +\right)  }\right)  _{\mathcal{\mathcal{LE}}%
\rightarrow\mathcal{K}}^{\dag}\nonumber\\
&  =\operatorname*{Tr}_{\mathcal{E}}\left(  \left\langle \phi
_{\mathcal{\mathcal{\mathcal{L}}}^{\prime}\mathcal{\mathcal{\mathcal{H}}}%
}\right\vert G_{\mathcal{K}\rightarrow\mathcal{\mathcal{L}}^{\prime
}\mathcal{\mathcal{E}}}~\mu_{\mathcal{KH}}\left\vert \phi
_{\mathcal{\mathcal{LH}}}\right\rangle \tilde{\upsilon}_{\mathcal{K}%
\rightarrow\mathcal{K}}\left\langle \phi_{\mathcal{LH}}\right\vert
\mu_{\mathcal{KH}}\left(  G^{\dag}\right)  _{\mathcal{L}^{\prime}%
\mathcal{E}\rightarrow\mathcal{K}}\left\vert \phi_{\mathcal{L}^{\prime
}\mathcal{H}}\right\rangle \right) \nonumber\\
&  =\operatorname*{Tr}_{\mathcal{\mathcal{KH}}}\left(  \mu_{\mathcal{KH}%
}\left\vert \phi_{\mathcal{\mathcal{LH}}}\right\rangle \tilde{\upsilon
}_{\mathcal{K}\rightarrow\mathcal{K}}\left\langle \phi_{\mathcal{LH}%
}\right\vert \mu_{\mathcal{KH}}\left(  G^{\dag}\right)  _{\mathcal{L}^{\prime
}\mathcal{E}\rightarrow\mathcal{K}}\left\vert \phi_{\mathcal{L}^{\prime
}\mathcal{H}}\right\rangle \left\langle \phi_{\mathcal{\mathcal{\mathcal{L}}%
}^{\prime}\mathcal{\mathcal{\mathcal{H}}}}\right\vert G_{\mathcal{K}%
\rightarrow\mathcal{\mathcal{L}}^{\prime}\mathcal{\mathcal{E}}}\right)
\label{was causing latex error}\\
&  =\operatorname*{Tr}_{\mathcal{\mathcal{KH}}}\left(  \mu_{\mathcal{KH}%
}\left(  \tilde{\upsilon}_{\mathcal{K}\rightarrow\mathcal{K}}\otimes\left\vert
\phi\right\rangle _{\mathcal{LH}}\left\langle \phi\right\vert \right)
\mu_{\mathcal{KH}}\Pi_{+}\left(  \phi_{\mathcal{H}}\right)  \right)
\label{uses nice identity overlap proof}\\
&  =\operatorname*{Tr}_{\mathcal{\mathcal{KH}}}\left(  \hat{\mu}%
_{\mathcal{KH}}^{2}\left(  \tilde{\upsilon}_{\mathcal{K}\rightarrow
\mathcal{K}}\otimes\left\vert \phi\right\rangle _{\mathcal{LH}}\left\langle
\phi\right\vert \right)  \right)  .\label{uses cyclicity and something else}%
\end{align}
Here equation $\ref{was causing latex error}$ uses cyclicity of the trace,
equation $\ref{uses nice identity overlap proof}$ uses $\left(
\ref{eq in rev lemma using the double kets}\right)  $, and equation
$\ref{uses cyclicity and something else}$ uses $\left(
\ref{phi proj identity}\right)  $, $\left(
\ref{def muhat after switch to phis}\right)  $, and cyclicity of the trace.
\end{proof}

\subsubsection{Estimates for the restricted maximum overlap
problem\label{subsection overlap estimates}}

On now may apply our angle estimates in a manner similar to that of section
$\ref{section simple proof Holevo Curlander}$:

\begin{theorem}
[Two-sided estimates for the maximum overlap problem]%
\label{Theorem two sided overlap estimates}Let $\mu_{\mathcal{\mathcal{KH}}} $
be positive semidefinite on $\mathcal{K}\otimes\mathcal{H}$ and let
$\left\vert \phi_{\mathcal{\mathcal{LH}}}\right\rangle $ be a unit vector.
Then
\begin{align}
\frac{\Lambda^{2}}{\operatorname*{Tr}\hat{\mu}_{\mathcal{\mathcal{KH}}}} &
\leq\left\langle \phi_{\mathcal{\mathcal{LH}}}\right\vert \mathcal{R}%
_{\mathcal{K}\rightarrow\mathcal{L}}^{\text{QO}}\left(  \mu
_{\mathcal{\mathcal{KH}}}\right)  \left\vert \phi_{\mathcal{\mathcal{LH}}%
}\right\rangle \leq\max_{\mathcal{R}_{\mathcal{K}\rightarrow\mathcal{L}}%
}\left\langle \phi_{\mathcal{\mathcal{LH}}}\right\vert \mathcal{R}%
_{\mathcal{K}\rightarrow\mathcal{L}}\left(  \mu_{\mathcal{\mathcal{KH}}%
}\right)  \left\vert \phi_{\mathcal{\mathcal{LH}}}\right\rangle \nonumber\\
&  \leq\Lambda\times\left\Vert \left(  \mathcal{R}^{\text{opt}}\right)
_{\mathcal{L}\rightarrow\mathcal{K}}^{\dag}\left(  \left\vert \phi
_{\mathcal{LH}}\right\rangle \left\langle \phi_{\mathcal{\mathcal{LH}}%
}\right\vert \right)  \right\Vert _{\infty}^{1/2}\leq\Lambda\text{,}%
\label{eq overlap bound}%
\end{align}
where the maximum is over quantum operations $\mathcal{R}:B^{1}\left(
\mathcal{K}\right)  \rightarrow B^{1}\left(  \mathcal{L}\right)  $, where
$\mathcal{R}^{\text{opt}}$ attains this maximum, and where%
\begin{equation}
\Lambda=\operatorname*{Tr}_{\mathcal{K}}\sqrt{\left\langle \phi
_{\mathcal{\mathcal{LH}}}\right\vert \hat{\mu}_{\mathcal{KH}}^{2}\left\vert
\phi_{\mathcal{\mathcal{LH}}}\right\rangle }.\label{gamma for overlap bound}%
\end{equation}
Here $\hat{\mu}_{\mathcal{\mathcal{KH}}}$ is given by $\left(
\ref{def muhat after switch to phis}\right)  $, $\mathcal{R}^{\dag}$ is given
by $\left(  \ref{eq defining adjoint}\right)  $, and one interprets
$0^{2}/0=0$.
\end{theorem}

\noindent\textbf{Remark: }It follows from $\left(  \ref{eq overlap bound}%
\right)  $ that
\[
\Lambda\leq\operatorname*{Tr}\left(  \hat{\mu}_{\mathcal{\mathcal{KH}}%
}\right)  \times\left\Vert \left(  \mathcal{R}^{\text{opt}}\right)
_{\mathcal{L}\rightarrow\mathcal{K}}^{\dag}\left(  \left\vert \phi
_{\mathcal{LH}}\right\rangle \left\langle \phi_{\mathcal{\mathcal{LH}}%
}\right\vert \right)  \right\Vert _{\infty}^{1/2}\leq\operatorname*{Tr}%
\hat{\mu}_{\mathcal{\mathcal{KH}}}\text{.}%
\]

\noindent\textbf{Note: }Given an arbitrary invertible operator $X:\mathcal{H}%
\rightarrow\mathcal{H}$, one may obtain potentially sharper estimates from
inequality $\left(  \ref{eq overlap bound}\right)  $ using a replacement of
the form
\begin{align}
\phi_{\mathcal{\mathcal{LH}}}  &  \rightarrow\frac{X_{\mathcal{H}%
\rightarrow\mathcal{H}}\phi_{\mathcal{\mathcal{LH}}}}{\left\Vert
X_{\mathcal{H}\rightarrow\mathcal{H}}\phi_{\mathcal{\mathcal{LH}}}\right\Vert
}\\
\mu_{\mathcal{\mathcal{KH}}}  &  \rightarrow\left\Vert X_{\mathcal{H}%
\rightarrow\mathcal{H}}\phi_{\mathcal{\mathcal{LH}}}\right\Vert ^{2}\left(
X^{-1}\right)  ^{\dag}\mu_{\mathcal{\mathcal{KH}}}X^{-1},
\end{align}
which does not change the overlap $\left\langle \phi_{\mathcal{\mathcal{LH}}%
}\right\vert \mathcal{R}_{\mathcal{K}\rightarrow\mathcal{L}}\left(
\mu_{\mathcal{\mathcal{KH}}}\right)  \left\vert \phi_{\mathcal{\mathcal{LH}}%
}\right\rangle $.

\bigskip

\begin{proof}
Since all quantities in $\left(  \ref{eq overlap bound}\right)  $ scale
linearly in $\mu$, set $\operatorname*{Tr}\hat{\mu}_{\mathcal{\mathcal{KH}}%
}=1$. (The case $\hat{\mu}=0$ is trivial.)

Let $\mathcal{R}^{\text{opt}}$ attain the maximum in $\left(
\ref{eq overlap bound}\right)  $. Take $G_{\mathcal{K}\rightarrow
\mathcal{\mathcal{LE}}}$ to be the \textquotedblleft small
angle\textquotedblright\ guess of Lemma
\ref{lemma construction of overlap guess}, with iterate $G_{\mathcal{K}%
\rightarrow\mathcal{\mathcal{LE}}}^{\left(  +\right)  }$ given by Theorem
\ref{Theorem iteration for overlap purifications}. By Lemma
\ref{lemma inequalitites for maximum overlap}, there exists a Stinespring
dilation $W_{\mathcal{K}\rightarrow\mathcal{\mathcal{LE}}}^{\text{opt}}$ of
$\mathcal{R}^{\text{opt}}$ such that the angle estimate $\left(
\ref{angle estimate for projective overlap guess}\right)  $ holds in the
equivalent form
\begin{equation}
\left\langle W^{\text{opt}},G\right\rangle _{\mu,\phi}\geq\frac{\left\langle
\phi_{\mathcal{\mathcal{LH}}}\right\vert \mathcal{R}_{\mathcal{K}%
\rightarrow\mathcal{L}}^{\text{opt}}\left(  \mu_{\mathcal{\mathcal{KH}}%
}\right)  \left\vert \phi_{\mathcal{\mathcal{LH}}}\right\rangle }{\left\Vert
\left(  \mathcal{R}^{\text{opt}}\right)  _{\mathcal{L}\rightarrow\mathcal{K}%
}^{\dag}\left(  \left\vert \phi_{\mathcal{LH}}\right\rangle \left\langle
\phi_{\mathcal{\mathcal{LH}}}\right\vert \right)  \right\Vert _{\infty}^{1/2}%
}\label{ineq terminal for MO proof}%
\end{equation}
given by $\left(  \ref{reduced angle estimate for maximum overlap}\right)  $.
But by lemma \ref{directional iterate lemma} and $\left(
\ref{zero order guess has norm 1}\right)  $,%
\begin{equation}
\left\Vert W^{\text{opt}}\right\Vert _{\mu,\phi}\geq\left\Vert G^{\left(
+\right)  }\right\Vert _{\mu,\phi}\geq\Lambda\left(  G\right)  \geq
\left\langle W^{\text{opt}},G\right\rangle _{\mu,\phi}\text{,}%
\label{ineq chain in MO proof}%
\end{equation}
where by $\left(  \ref{def of lamda in abstract case}\right)  $ and $\left(
\ref{zero order guess has norm 1}\right)  -\left(
\ref{inner prod of G and Gplus for zero order guess}\right)  $,
\begin{equation}
\Lambda\left(  G\right)  =\operatorname{Re}\left\langle G^{\left(  +\right)
},G/\left\Vert G\right\Vert _{\mu,\phi}\right\rangle _{\mu,\phi}%
=\Lambda\text{.}%
\end{equation}
The third inequality of $\left(  \ref{eq overlap bound}\right)  $ follows by
appending $\left(  \ref{ineq terminal for MO proof}\right)  $ to $\left(
\ref{ineq chain in MO proof}\right)  $. Squaring the first three quantities of
$\left(  \ref{ineq chain in MO proof}\right)  $ proves the first two
inequalities of $\left(  \ref{eq overlap bound}\right)  $. The final
inequality follows from the fact that $\mathcal{R}^{\text{opt}}$ is a quantum operation.
\end{proof}

\subsection{Estimates for quantum conditional
min-entropy\label{section estimates for conditional min-entropy}}

Theorem $\ref{Theorem two sided overlap estimates}$ has the following corollary:

\begin{corollary}
\label{corollary bounding min entropy}Let $\mathcal{H}_{A}$ and $\mathcal{H}%
_{B}$ be finite-dimensional, and let $\rho_{AB}$ be a density matrix on
$\mathcal{H}_{A}\otimes\mathcal{H}_{B}$. Then for any $s\in\mathbb{R}$ the
conditional min-entropy $\left(  \ref{eq defining conditional minentropy}%
\right)  $ of $A$ given $B$ satisfies the bounds%
\begin{equation}
-\log_{2}\left(  \sqrt{\operatorname*{Tr}\rho_{A}^{s}}\times\operatorname*{Tr}%
_{B}\sqrt{\operatorname*{Tr}_{A}\rho_{AB}\rho_{A}^{-s}\rho_{AB}}\right)  \leq
H_{\text{min}}\left(  A|B\right)  _{\rho}\leq-\log_{2}\left(  \frac{\left(
\operatorname*{Tr}_{B}\sqrt{\operatorname*{Tr}_{A}\rho_{AB}\rho_{A}^{-s}%
\rho_{AB}}\right)  ^{2}}{\operatorname*{Tr}\rho_{A}^{1-s}}\right)
\label{minentropy bounds}%
\end{equation}
where $\rho_{A}=\operatorname*{Tr}_{B}\rho_{AB}$. Here any non-positive powers
are evaluated as in equation $\ref{eq defining minus 1/2 plus exponent}$. (In
particular%
\begin{equation}
\rho_{A}^{0}:=\Pi_{+}\left(  \rho_{A}\right)  \text{,}%
\label{hidden projection equation}%
\end{equation}
where the RHS is the positive projection $\left(
\ref{eq defining positive projection}\right)  $.)
\end{corollary}

\noindent\textbf{Remarks:}

\begin{enumerate}
\item The $s=0$ case of $\left(  \ref{minentropy bounds}\right)  $ is
particularly simple:%
\begin{equation}
-\log\left(  \operatorname*{Tr}_{B}\sqrt{\operatorname*{Tr}_{A}\rho_{AB}^{2}%
}\right)  -\frac{1}{2}\log\left(  \operatorname{rank}\left(  \rho_{A}\right)
\right)  \leq H_{\text{min}}\left(  A|B\right)  _{\rho}\leq-2\log_{2}\left(
\operatorname*{Tr}_{B}\sqrt{\operatorname*{Tr}_{A}\rho_{AB}^{2}}\right)
\text{.}%
\end{equation}

\item If $s=1/2$ and $\rho_{AB}=\left\vert \psi_{AB}\right\rangle \left\langle
\psi_{AB}\right\vert $ is pure then a simple calculation$^{\text{\cite{algebra
step for frederic}}}$ shows that the upper and lower bounds of $\left(
\ref{minentropy bounds}\right)  $ agree, yielding the known \cite{Konig Renner
Schaffner Operational meaning of min and max entropy} expression%
\begin{equation}
H_{\text{min}}\left(  A|B\right)  _{\left\vert \psi_{AB}\right\rangle }%
=-2\log_{2}\left(  \operatorname*{Tr}_{A}\sqrt{\rho_{A}}\right)  \text{.}%
\end{equation}

\item More generally, if $\rho_{AB}=\left\vert \psi_{AB}\right\rangle
\left\langle \psi_{AB}\right\vert $ is pure then the upper bound of $\left(
\ref{minentropy bounds}\right)  $ is exact and independent of $s$.

\item If $\rho_{AB}$ is a maximally-entangled pure state then the lower bound
of $\left(  \ref{minentropy bounds}\right)  $ is also exact and independent of
$s$.
\end{enumerate}

\begin{proof}
Let%
\[
\rho_{A}=%
{\displaystyle\sum}
\lambda_{k}\left\vert k\right\rangle _{A}\left\langle k\right\vert
\]
be a spectral decomposition. Then by equation $\ref{eq min entropy as overlap}%
$ one has the identity%
\begin{equation}
2^{-H_{\text{min}}\left(  A|B\right)  _{\rho}}=\left\langle \psi
_{s}\right\vert _{AA^{\ast}}\max_{\mathcal{R}}\left(  \mathcal{R}%
_{B\rightarrow A^{\ast}}\left(  \mu_{s}\right)  _{AB}\right)  \left\vert
\psi_{s}\right\rangle _{AA^{\ast}}\text{,}%
\label{put min entropy step in nice form}%
\end{equation}
where%
\begin{align}
\left\vert \psi_{s}\right\rangle _{AA^{\ast}} &  =\frac{\left.  \left\vert
\rho_{A}^{s/2}\right\rangle \!\right\rangle _{AA^{\ast}}}{\left\Vert \rho
_{A}^{s/2}\right\Vert _{2}}=%
{\displaystyle\sum_{k\text{ with }\lambda_{k}>0}}
\frac{\lambda_{k}^{s/2}}{\sqrt{\operatorname*{Tr}\rho_{A}^{s}}}\left\vert
k\right\rangle _{A}\left\vert \bar{k}\right\rangle _{A^{\ast}}\\
\left(  \mu_{s}\right)  _{AB} &  =\left\Vert \rho_{A}^{s/2}\right\Vert
_{2}^{2}\rho_{A}^{-s/2}\rho_{AB}\rho_{A}^{-s/2}\text{.}%
\end{align}
The bounds $\left(  \ref{minentropy bounds}\right)  $ follow by Theorem
\ref{Theorem two sided overlap estimates}.
\end{proof}

\bigskip

\noindent\textbf{Remark:} Using the fourth term of inequality
\ref{eq overlap bound}, one may tighten the lower bound of $\left(
\ref{minentropy bounds}\right)  $ in cases where one can estimate%
\[
\left\Vert \left(  \mathcal{R}^{\text{opt}}\right)  _{B\rightarrow
\mathcal{A}^{\ast}}^{\dag}\left(  \left\vert \psi_{s}\right\rangle
\left\langle \psi_{s}\right\vert \right)  \right\Vert _{\infty}\text{,}%
\]
where $\mathcal{R}^{\text{opt}}$ is a maximizer of $\left(
\ref{put min entropy step in nice form}\right)  $. Appendix C shows that this
works in the case that $\rho_{AB}$ is a \textquotedblleft
quantum-classical\textquotedblright\ state.

\section{\label{section approximate channel reversals}Approximate Channel
Reversals}

This section applies Theorem \ref{Theorem two sided overlap estimates} to
estimate the reversibility of an arbitrary quantum operation $\mathcal{A}%
:B^{1}\left(  \mathcal{H}\right)  \rightarrow B^{1}\left(  \mathcal{K}\right)
$, as measured by entanglement fidelity $\max_{\mathcal{R}_{\mathcal{K}%
\rightarrow\mathcal{H}}}F_{e}\left(  \rho,\mathcal{R}\circ\mathcal{A}\right)
$. (Note that \emph{more generally Theorem
\ref{Theorem two sided overlap estimates} gives estimates when the input state
of }$\mathcal{A}$\emph{\ and target output state of }$\mathcal{R}%
$\emph{\ differ}, but we focus on this special case.)

In order to express the our reversibility estimates in a more intuitive form
(and to understand the relationship of the corresponding reversal with that of
Barnum and Knill), it is useful to introduce a method for applying functions
to CP maps.

\subsection{\label{section introducing reweightings}The $\rho$-functional
calculus for CP maps}

One may tailor the Kraus decomposition \cite{Kraus Decomp probably} a CP map
to a given input density matrix $\rho$:

\begin{definition}
Let $\mathcal{A}:B^{1}\left(  \mathcal{H}\right)  \rightarrow B^{1}\left(
\mathcal{K}\right)  $ be a completely-positive map and let $\rho$ be a density
matrix on $\mathcal{H}$. A $\rho$\textbf{-Kraus decomposition} of the
restriction of $\mathcal{A}$ to $B^{1}\left(  \operatorname*{supp}\left(
\rho\right)  \right)  $ is a decomposition of the form%
\begin{equation}
\mathcal{A}\left(  \mu\right)  =%
{\displaystyle\sum}
p_{k}E_{k}\mu E_{k}^{\dag}\text{, \ \ }\operatorname*{supp}\left(  \mu\right)
\subseteq\operatorname*{supp}\left(  \rho\right) \label{rho kraus decomp}%
\end{equation}
where

\begin{enumerate}
\item The vectors $E_{k}\left\vert \psi_{\rho}\right\rangle \in\mathcal{K}%
\otimes\mathcal{H}^{\ast}$ are orthonormal, where $\psi_{\rho}$ is the
purification $\left(  \ref{eq for canonical purification}\right)  $.

\item The $p_{k}$ are non-negative.

\item The $\rho$\textbf{-Kraus operators} $E_{k}:\mathcal{H}\rightarrow
\mathcal{K}$ have supports in $\operatorname*{supp}\left(  \rho\right)  $.
\end{enumerate}
\end{definition}

\noindent\textbf{Remarks:}

\begin{enumerate}
\item Existence of a $\rho$-Kraus decomposition follows by an easy
modification of standard techniques. (See Proposition
\ref{simple proposition for rho kraus and rho functional cal}, below.)

\item When $\mathcal{A}$ is trace-preserving, one interprets $\mathcal{A}$ as
acting on $\rho$ by randomly sending the purification $\left\vert \psi_{\rho
}\right\rangle _{\mathcal{HH}^{\ast}}\mathstrut$ into one of the orthonormal
vectors $E_{k}\left\vert \psi_{\rho}\right\rangle $, which are
classically-distinguishable by the observer with access to
$\mathcal{\mathcal{KH}}^{\ast}$.

\item By equations \ref{eq for canonical purification} and
\ref{eq double ket is an isometry}, the orthonormality of the $E_{k}\left\vert
\psi_{\rho}\right\rangle _{\mathcal{HH}^{\ast}}$ is equivalent to the
condition%
\begin{equation}
\operatorname*{Tr}\left(  E_{k}^{\dag}E_{\ell}\rho\right)  =\delta_{k\ell
}\text{.}\label{equiv orthonorm condition in rho kraus}%
\end{equation}
If $\rho$ is maximally-mixed one therefore obtains the usual orthogonality
conditions \cite{Arrighi and Patricot On quantum operations as quantum states}
sometimes required for the Kraus operators.
\end{enumerate}

Given a state $\rho$, there is a natural notion of applying functions to
completely positive maps:

\begin{definition}
[$\rho$-functional calculus for CP maps]\label{def of rho functional calc}Let
$\mathcal{A}:B^{1}\left(  \mathcal{H}\right)  \rightarrow B^{1}\left(
\mathcal{K}\right)  $ be a completely positive map with the $\rho$-Kraus
decomposition $\left(  \ref{rho kraus decomp}\right)  $. For $f:\left[
0,\infty\right)  \rightarrow\left[  0,\infty\right)  $ one defines the CP map
$f_{\rho}\left(  \mathcal{A}\right)  :B^{1}\left(  \operatorname*{supp}\left(
\rho\right)  \right)  \rightarrow B^{1}\left(  \mathcal{K}\right)  $ by%
\begin{equation}
\left(  f_{\rho}\left(  \mathcal{A}\right)  \right)  \left(  \mu\right)  =%
{\displaystyle\sum}
f\left(  p_{k}\right)  E_{k}\mu E_{k}^{\dag}.\label{eq for f sub rho of A}%
\end{equation}
The \textbf{quadratic reweighting $\mathcal{A}^{\left(  2,\rho\right)  }$} of
$\mathcal{A}$ corresponds to the case $f\left(  p\right)  =p^{2}$:%
\begin{equation}
\mathcal{A}^{\left(  2,\rho\right)  }\left(  \mu\right)  =%
{\displaystyle\sum}
p_{k}^{2}E_{k}\mu E_{k}^{\dag}\text{.}\label{eq explicit quadratic weighting}%
\end{equation}

\end{definition}

The following proposition shows that the CP maps $f_{\rho}\left(
\mathcal{A}\right)  $ and $\mathcal{A}^{\left(  2,\rho\right)  }$ are
well-defined and independent of the decomposition $\left(
\ref{rho kraus decomp}\right)  $:

\begin{proposition}
\label{simple proposition for rho kraus and rho functional cal}Let $f$,
$\mathcal{A},$ and $\rho$ be as in Definition \ref{def of rho functional calc}%
. Then

\begin{enumerate}
\item A $\rho$-Kraus decomposition $\left(  \ref{rho kraus decomp}\right)  $
of $\mathcal{A}$ exists. Furthermore,
\[%
{\displaystyle\sum}
p_{k}=\operatorname*{Tr}\mathcal{A}\left(  \rho\right)  ,
\]
so that $\left\{  p_{k}\right\}  $ is a probability distribution if
$\mathcal{A}$ is trace preserving.

\item One has the identity%
\begin{equation}
\left(  f_{\rho}\left(  \mathcal{A}\right)  \right)  \left(  \mu\right)
=\operatorname*{Tr}_{\mathcal{H}^{\ast}}\left(  \left(  \overline
{\rho^{-1/2^{+}}\mu^{\dag}\rho^{-1/2^{+}}}\right)  _{\mathcal{H}^{\ast
}\rightarrow\mathcal{H}^{\ast}}\times f\left(  \mathcal{A}\left(  \left\vert
\psi_{\rho}\right\rangle _{\mathcal{HH}^{\ast}}\left\langle \psi_{\rho
}\right\vert \right)  \right)  \right)  \text{,}%
\label{recover fsub rho of A from rho-choi}%
\end{equation}
for all $\mu\in B^{1}\left(  \operatorname*{supp}\left(  \rho\right)  \right)
$, where barred product of operators acts on $\mathcal{H}^{\ast}$ as indicated
by equations $\ref{stupid def of bar}$-$\ref{stupid def of partial transpose}
$. Here one applies $f$ to a self-adjoint operator using the functional
calculus \cite{Reed and Simon I}: Given a spectral decomposition
\begin{equation}
A=%
{\displaystyle\sum}
\lambda_{i}\Pi_{i}\text{,}%
\end{equation}
one writes%
\begin{equation}
f\left(  A\right)  =%
{\displaystyle\sum}
f\left(  \lambda_{i}\right)  \Pi_{i}.\label{def of fun calc}%
\end{equation}

\item In particular, the CP map $f_{\rho}\left(  \mathcal{A}\right)  $ is
independent of the choice of decomposition $\left(  \ref{rho kraus decomp}%
\right)  $, and%
\begin{equation}
\mathcal{A}^{\left(  2,\rho\right)  }\left(  \mu\right)  =\operatorname*{Tr}%
_{\mathcal{H}^{\ast}}\left(  \left(  \overline{\rho^{-1/2^{+}}\mu^{\dag}%
\rho^{-1/2^{+}}}\right)  _{\mathcal{H}^{\ast}\rightarrow\mathcal{H}^{\ast}%
}\times\left(  \mathcal{A}\left(  \left\vert \psi_{\rho}\right\rangle
_{\mathcal{HH}^{\ast}}\left\langle \psi_{\rho}\right\vert \right)  \right)
^{2}\right)  \text{,}\label{eq for quadratic weighting as it appears in proof}%
\end{equation}
for $\mu\in B^{1}\left(  \operatorname*{supp}\left(  \rho\right)  \right)  $.
\end{enumerate}
\end{proposition}

\noindent\textbf{Remarks:}

\begin{enumerate}
\item If $f\left(  p\right)  =p$ is the identity function and if $\rho$ is
maximally mixed then equation $\ref{recover fsub rho of A from rho-choi}$
reduces to the usual procedure for recovering a channel from its Choi matrix
$\left(  \ref{eq for basis free choi matrix}\right)  $.

\item Equation $\left(
\ref{eq for quadratic weighting as it appears in proof}\right)  $ gives the
form of $\mathcal{A}^{\left(  2,\rho\right)  }$ which appears when one applies
Theorem \ref{Theorem two sided overlap estimates} to obtain reversibility
estimates. The transpose of $\mu$ becomes a partial transpose $\left(
\ref{stupid def of partial transpose}\right)  $ when $\mathcal{A}^{\left(
2,\rho\right)  }$ is applied to the state of a composite quantum system.

\item A prescription for computing $\mathcal{A}^{\left(  2,\rho\right)  }$
from an arbitrary set of Kraus operators of $\mathcal{A}$ appears in the next section.
\end{enumerate}

\begin{proof}
By equations \ref{basic eq for double kets},
\ref{eq for canonical purification}, \ref{eq defining minus 1/2 plus exponent}%
, and \ref{eq defining positive projection},%
\begin{equation}
\overline{\mu^{1/2}\rho^{-1/2^{+}}}\left\vert \psi_{\rho}\right\rangle
_{\mathcal{HH}^{\ast}}=\left.  \left\vert \rho^{1/2}\rho^{-1/2^{+}}\mu
^{1/2}\right\rangle \!\right\rangle _{\mathcal{HH}^{\ast}}=\left.  \left\vert
\Pi_{+}\left(  \rho\right)  \mu^{1/2}\right\rangle \!\right\rangle
_{\mathcal{HH}^{\ast}}=\left\vert \psi_{\mu}\right\rangle
_{\mathcal{\mathcal{HH}^{\ast}}}%
.\label{eq relationship between psi sub rho and psi sub mu}%
\end{equation}
It therefore follows by equation \ref{eq psi sub rho purifies rho} that for
densities $\mu\in B^{1}\left(  \mathcal{H}\right)  $%
\begin{equation}
\mathcal{A}\left(  \mu\right)  =\operatorname*{Tr}_{\mathcal{H}^{\ast}%
}\mathcal{A}\left(  \left\vert \psi_{\mu}\right\rangle _{\mathcal{HH}^{\ast}%
}\left\langle \psi_{\mu}\right\vert \right)  =\operatorname*{Tr}%
_{\mathcal{H}^{\ast}}\overline{\mu^{1/2}\rho^{-1/2^{+}}}\mathcal{A}\left(
\left\vert \psi_{\rho}\right\rangle _{\mathcal{HH}^{\ast}}\left\langle
\psi_{\rho}\right\vert \right)  \overline{\rho^{-1/2^{+}}\mu^{1/2}}%
\text{.}\label{get A back from rho-kraus}%
\end{equation}
Taking a spectral decomposition%
\begin{equation}
\mathcal{A}\left(  \left\vert \psi_{\rho}\right\rangle \left\langle \psi
_{\rho}\right\vert \right)  =%
{\displaystyle\sum}
p_{k}\left.  \left\vert F_{k}\right\rangle \!\right\rangle
_{\mathcal{\mathcal{KH}}^{\ast}}\left\langle \!\left\langle F_{k}\right\vert
\right.  \text{,}%
\end{equation}
it follows from equations \ref{get A back from rho-kraus},
\ref{basic eq for double kets}, and \ref{eq used to construct canonical purif}
that%
\begin{equation}
\mathcal{A}\left(  \mu\right)  =\operatorname*{Tr}_{\mathcal{H}^{\ast}}%
{\displaystyle\sum}
p_{k}\left.  \left\vert F_{k}\rho^{-1/2^{+}}\mu^{1/2}\right\rangle
\!\right\rangle _{\mathcal{\mathcal{\mathcal{\mathcal{KH}}}}^{\ast}%
}\left\langle \!\left\langle F_{k}\rho^{-1/2^{+}}\mu^{1/2}\right\vert \right.
=%
{\displaystyle\sum}
p_{k}F_{k}\rho^{-1/2^{+}}\mu\rho^{-1/2^{+}}F_{k}^{\dag}\text{.}%
\end{equation}
Setting
\begin{equation}
E_{k}=F_{k}\rho_{k}^{-1/2^{+}}%
\end{equation}
gives the desired $\rho$-Kraus decomposition, where the desired condition
$\left(  \ref{equiv orthonorm condition in rho kraus}\right)  $ follows from
the orthonormality of the $\left.  \left\vert F_{k}\right\rangle
\!\right\rangle \in\mathcal{K}\otimes\mathcal{H}^{\ast}$ using equation
\ref{eq double ket is an isometry}.

If $\mathcal{A}$ is trace-preserving then%
\[
1=\operatorname*{Tr}_{\mathcal{\mathcal{KH}}^{\ast}}\mathcal{A}\left(
\left\vert \psi_{\rho}\right\rangle _{\mathcal{HH}^{\ast}}\left\langle
\psi_{\rho}\right\vert \right)  =\operatorname*{Tr}_{\mathcal{\mathcal{KH}%
}^{\ast}}%
{\displaystyle\sum}
p_{k}E_{k}\left\vert \psi_{\rho}\right\rangle _{\mathcal{HH}^{\ast}%
}\left\langle \psi_{\rho}\right\vert E_{k}^{\dag}=%
{\displaystyle\sum}
p_{k}\text{,}%
\]
by the defining orthonormality condition on the $\left\{  E_{k}\right\}  $,
proving that $\left\{  p_{k}\right\}  $ is a probability distribution.

Now suppose that we are given an arbitrary $\rho$-Kraus decomposition $\left(
\ref{rho kraus decomp}\right)  $ and that $\mu\in B^{1}\left(
\operatorname*{supp}\left(  \rho\right)  \right)  $. Note that since both
sides of $\left(  \ref{recover fsub rho of A from rho-choi}\right)  $ are
linear in $\mu$ we may assume without loss of generality that $\mu$ is
positive semidefinite. Then by equations \ref{eq psi sub rho purifies rho}$,$
\ref{eq relationship between psi sub rho and psi sub mu},
\ref{def of fun calc}, \ref{rho kraus decomp}, and cyclicity of the trace
\begin{align*}%
{\displaystyle\sum}
f\left(  p_{k}\right)  E_{k}\mu E_{k}^{\dag} &  =\operatorname*{Tr}%
_{\mathcal{H}^{\ast}}%
{\displaystyle\sum}
f\left(  p_{k}\right)  E_{k}\left\vert \psi_{\mu}\right\rangle _{\mathcal{HH}%
^{\ast}}\left\langle \psi_{\mu}\right\vert E_{k}^{\dag}\\
&  =\operatorname*{Tr}_{\mathcal{H}^{\ast}}\left[  \overline{\mu^{1/2}%
\rho^{-1/2^{+}}}\left(
{\displaystyle\sum}
f\left(  p_{k}\right)  E_{k}\left\vert \psi_{\rho}\right\rangle _{\mathcal{HH}%
^{\ast}}\left\langle \psi_{\rho}\right\vert E_{k}^{\dag}\right)
\overline{\rho^{-1/2^{+}}\mu^{1/2}}\right] \\
&  =\operatorname*{Tr}_{\mathcal{H}^{\ast}}\left[  \overline{\mu^{1/2}%
\rho^{-1/2^{+}}}f\left(
{\displaystyle\sum}
p_{k}E_{k}\left\vert \psi_{\rho}\right\rangle _{\mathcal{HH}^{\ast}%
}\left\langle \psi_{\rho}\right\vert E_{k}^{\dag}\right)  \overline
{\rho^{-1/2^{+}}\mu^{1/2}}\right] \\
&  =\operatorname*{Tr}_{\mathcal{H}^{\ast}}\left[  \overline{\rho^{-1/2^{+}%
}\mu\rho^{-1/2^{+}}}\times f\left(  \mathcal{A}\left(  \left\vert \psi_{\rho
}\right\rangle _{\mathcal{HH}^{\ast}}\left\langle \psi_{\rho}\right\vert
\right)  \right)  \right]  \text{,}%
\end{align*}
as desired.
\end{proof}

\subsection{Quadratic quantum error recovery}

\begin{theorem}
\label{theorem quadratic recovery estimates}Let $\mathcal{A}:B^{1}\left(
\mathcal{H}\right)  \rightarrow B^{1}\left(  \mathcal{K}\right)  $ be a
quantum operation, and let $\rho$ be a density matrix on $\mathcal{H}$. Then
one has the following bounds on the optimal entanglement fidelity of recovery
\begin{equation}
\frac{\Lambda^{2}}{\operatorname*{Tr}\mathcal{A}\left(  \rho\right)  }\leq
F_{e}\left(  \rho,\mathcal{R}^{\text{QR}}\circ\mathcal{A}\right)  \leq
\sup_{\mathcal{R}_{\mathcal{K}\rightarrow\mathcal{H}}}F_{e}\left(
\rho,\mathcal{R}\circ\mathcal{A}\right)  \leq\Lambda\text{,}%
\label{eq recovery channel bounds}%
\end{equation}
where the supremum is over quantum operations $\mathcal{R}:B^{1}\left(
\mathcal{K}\right)  \rightarrow B^{1}\left(  \mathcal{H}\right)  $, where%
\begin{equation}
\Lambda=\operatorname*{Tr}_{\mathcal{K}}\sqrt{\mathcal{A}^{\left(
2,\rho\right)  }\left(  \rho^{2}\right)  }\leq\operatorname*{Tr}%
\mathcal{A}\left(  \rho\right)  \text{,}\label{eq for gamma of recovery}%
\end{equation}
where $\mathcal{A}^{\left(  2,\rho\right)  }$ is given by Definition
\ref{def of rho functional calc} (see also equations
\ref{eq for quadratic weighting as it appears in proof} and
\ref{eq for quadratic weighting given Kraus operators}, below), and where
\textbf{quadratic recovery operation} is given by%
\begin{equation}
\mathcal{R}^{\text{QR}}\left(  \upsilon\right)  =\rho\left(  \mathcal{A}%
^{\left(  2,\rho\right)  }\right)  ^{\dag}\left(  \left(  \mathcal{A}^{\left(
2,\rho\right)  }\left(  \rho^{2}\right)  \right)  ^{-1/2^{+}}\upsilon
_{\mathcal{K}\rightarrow\mathcal{K}}\left(  \mathcal{A}^{\left(
2,\rho\right)  }\left(  \rho^{2}\right)  \right)  ^{-1/2^{+}}\right)
\rho\text{.}\label{eq quadratic recovery channel}%
\end{equation}
Here the adjoint $\left(  \bullet\right)  ^{\dag}$ is from Definition
$\ref{definition of adjoint as an actual displayed defintion}$.
\end{theorem}

\noindent\textbf{Remark: }In the case that $\mathcal{A}$ is trace-preserving,
one may plug the square of the last inequality of $\left(
\ref{eq recovery channel bounds}\right)  $ into the first inequality $\left(
\ref{eq recovery channel bounds}\right)  $, giving%
\begin{equation}
\frac{F_{e}\left(  \rho,\mathcal{R}^{\text{QR}}\circ\mathcal{A}\right)  }%
{\sup_{\mathcal{R}_{\mathcal{K}\rightarrow\mathcal{H}}}F_{e}\left(
\rho,\mathcal{R}\circ\mathcal{A}\right)  }\geq\frac{\Lambda^{2}}%
{\sup_{\mathcal{R}_{\mathcal{K}\rightarrow\mathcal{H}}}F_{e}\left(
\rho,\mathcal{R}\circ\mathcal{A}\right)  }\geq\sup_{\mathcal{R}_{\mathcal{K}%
\rightarrow\mathcal{H}}}F_{e}\left(  \rho,\mathcal{R}\circ\mathcal{A}\right)
\text{.}\label{eq quad reversal also satisfy barnum knill bounds}%
\end{equation}
\textit{In particular, both of the lower bounds of }$\left(
\ref{eq recovery channel bounds}\right)  $ \textit{are sufficiently tight to
also satisfy the tightness relation} $\left(  \ref{Barnum Knill estimate}%
\right)  $ \textit{of Barnum and Knill} \cite{Barnum Knill UhOh}. (Note,
however, that Barnum and Knill also produce estimates for \textit{average}
entanglement fidelity, under certain commutativity assumptions.) Furthermore,
by expressing the bounds $\left(  \ref{eq recovery channel bounds}\right)  $
in terms of the infidelity $1-F_{e}\left(  \rho,\mathcal{R}\circ
\mathcal{A}\right)  $ one obtains the fact that $\mathcal{R}^{\text{QR}}$,
like $\mathcal{R}^{\text{BK}}$, has an infidelity of recovery within a factor
of two of the optimal.

\medskip

\begin{proof}
Let $\mathcal{H}_{\text{in}}$ be a copy of $\mathcal{H}$, let $\left\vert
\psi_{\rho}\right\rangle _{\mathcal{HH}^{\ast}}$ be the canonical purification
$\left(  \ref{eq for canonical purification}\right)  $ of $\rho,$ and set%
\[
\mu_{\mathcal{\mathcal{KH}}^{\ast}}=\mathcal{A}_{\mathcal{H}_{\text{in}%
}\rightarrow\mathcal{K}}\left(  \left\vert \psi_{\rho}\right\rangle
_{\mathcal{H}_{\text{in}}\mathcal{H}^{\ast}}\left\langle \psi_{\rho
}\right\vert \right)  \text{.}%
\]
Using the replacements $\left(  \mathcal{H},\mathcal{K},\mathcal{L}\right)
\rightarrow\left(  \mathcal{H}^{\ast},\mathcal{K},\mathcal{H}\right)  $ and
$\left\vert \phi\right\rangle _{\mathcal{\mathcal{LH}}}\rightarrow\left\vert
\psi_{\rho}\right\rangle _{\mathcal{HH}^{\ast}}$, Theorem
$\ref{Theorem two sided overlap estimates}$ gives estimates of the form%
\begin{equation}
\frac{\tilde{\Lambda}^{2}}{\operatorname*{Tr}\hat{\mu}_{\mathcal{\mathcal{KH}%
}^{\ast}}}\leq F_{e}\left(  \rho,\mathcal{R}^{\text{QO}}\circ\mathcal{A}%
\right)  \leq\sup_{\mathcal{R}_{\mathcal{K}\rightarrow\mathcal{H}}}%
F_{e}\left(  \rho,\mathcal{R}\circ\mathcal{A}\right)  \leq\tilde{\Lambda
}\text{.}%
\end{equation}
We claim that $\operatorname*{Tr}\hat{\mu}_{\mathcal{\mathcal{KH}}^{\ast}%
}=\operatorname*{Tr}\mathcal{A}\left(  \rho\right)  $, $\tilde{\Lambda
}=\Lambda$, and $\mathcal{R}^{\text{QO}}=\mathcal{R}^{\text{QR}}$.

\textbf{First claim:} Note that%
\begin{align}
\hat{\mu}_{\mathcal{\mathcal{KH}}^{\ast}}  &  =\Pi_{+}\left(
\operatorname*{Tr}_{\mathcal{H}}\left\vert \psi_{\rho}\right\rangle
_{\mathcal{HH}^{\ast}}\left\langle \psi_{\rho}\right\vert \right)
\mu_{\mathcal{\mathcal{KH}}^{\ast}}\Pi_{+}\left(  \operatorname*{Tr}%
_{\mathcal{H}}\left\vert \psi_{\rho}\right\rangle _{\mathcal{HH}^{\ast}%
}\left\langle \psi_{\rho}\right\vert \right) \label{z use muhat def}\\
&  =\Pi_{+}\left(  \bar{\rho}_{\mathcal{H}^{\ast}}\right)  \mathcal{A}%
_{\mathcal{H}_{\text{in}}\rightarrow\mathcal{K}}\left(  \left.  \left\vert
\sqrt{\rho}\,\right\rangle \!\right\rangle _{\mathcal{\mathcal{H}}_{\text{in}%
}\mathcal{\mathcal{H}^{\ast}}}\left\langle \!\left\langle \sqrt{\rho
}\,\right\vert \right.  \right)  \Pi_{+}\left(  \bar{\rho}_{\mathcal{H}^{\ast
}}\right) \label{z use dual purif and def canon}\\
&  =\mathcal{A}_{\mathcal{H}_{\text{in}}\rightarrow\mathcal{K}}\left(  \left.
\left\vert \sqrt{\rho}\,\Pi_{+}\left(  \rho\right)  \right\rangle
\!\right\rangle _{\mathcal{\mathcal{H}}_{\text{in}}\mathcal{\mathcal{H}^{\ast
}}}\left\langle \!\left\langle \sqrt{\rho}\,\Pi_{+}\left(  \rho\right)
\right\vert \right.  \right) \label{z use basic}\\
&  =\mathcal{A}_{\mathcal{H}_{\text{in}}\rightarrow\mathcal{K}}\left(
\left\vert \psi_{\rho}\right\rangle _{\mathcal{\mathcal{H}}_{\text{in}%
}\mathcal{\mathcal{H}^{\ast}}}\left\langle \psi_{\rho}\right\vert \right)
=\mu_{\mathcal{\mathcal{KH}}^{\ast}}\text{,}\label{mu is muhat for reversal}%
\end{align}
where the first three equalities used $\left(
\ref{def muhat after switch to phis}\right)  $ \& $\left(
\ref{phi sub H defined}\right)  $, $\left(
\ref{eq psi sub rho purifies rhobar}\right)  $, and $\left(
\ref{basic eq for double kets}\right)  $. It follows by equation
$\ref{eq psi sub rho purifies rho}$ that%
\[
\operatorname*{Tr}_{\mathcal{\mathcal{KH}}^{\ast}}\hat{\mu}%
_{\mathcal{\mathcal{KH}}^{\ast}}=\operatorname*{Tr}_{\mathcal{K}}%
\mathcal{A}_{\mathcal{H}_{\text{in}}\rightarrow\mathcal{K}}\left(
\operatorname*{Tr}_{\mathcal{H}^{\ast}}\left\vert \psi_{\rho}\right\rangle
_{\mathcal{\mathcal{H}}_{\text{in}}\mathcal{\mathcal{H}^{\ast}}}\left\langle
\psi_{\rho}\right\vert \right)  =\operatorname*{Tr}\mathcal{A}\left(
\rho\right)  \text{,}%
\]
as desired.

\textbf{Second claim: }One computes%
\begin{align}
&  \left\langle \psi_{\rho}\right\vert _{\mathcal{\mathcal{HH}^{\ast}}}\left(
\hat{\mu}_{\mathcal{\mathcal{KH}}^{\ast}}\right)  ^{2}\left\vert \psi_{\rho
}\right\rangle _{\mathcal{\mathcal{HH}^{\ast}}}\nonumber\\
&  =\left\langle \psi_{\rho}\right\vert _{\mathcal{\mathcal{HH}^{\ast}}%
}\left(  \mathcal{A}_{\mathcal{H}_{\text{in}}\rightarrow\mathcal{K}}\left(
\left\vert \psi_{\rho}\right\rangle _{\mathcal{H}_{\text{in}}\mathcal{H}%
^{\ast}}\left\langle \psi_{\rho}\right\vert \right)  \right)  ^{2}\left\vert
\psi_{\rho}\right\rangle _{\mathcal{\mathcal{HH}^{\ast}}}\label{c uses mu hat}%
\\
&  =\operatorname*{Tr}_{\mathcal{\mathcal{HH}^{\ast}}}\left[  \left(
\mathcal{A}_{\mathcal{H}_{\text{in}}\rightarrow\mathcal{K}}\left(  \left\vert
\psi_{\rho}\right\rangle _{\mathcal{H}_{\text{in}}\mathcal{H}^{\ast}%
}\left\langle \psi_{\rho}\right\vert \right)  \right)  ^{2}\left\vert
\psi_{\rho}\right\rangle _{\mathcal{\mathcal{HH}^{\ast}}}\left\langle
\psi_{\rho}\right\vert \right] \label{c uses cyclicity}\\
&  =\operatorname*{Tr}_{\mathcal{\mathcal{H}^{\ast}}}\left[  \left(
\mathcal{A}_{\mathcal{H}_{\text{in}}\rightarrow\mathcal{K}}\left(  \left\vert
\psi_{\rho}\right\rangle _{\mathcal{H}_{\text{in}}\mathcal{H}^{\ast}%
}\left\langle \psi_{\rho}\right\vert \right)  \right)  ^{2}\bar{\rho
}_{\mathcal{H}^{\ast}}\right] \label{c dual purity}\\
&  =\mathcal{A}_{\mathcal{H}\rightarrow\mathcal{K}}^{\left(  2,\rho\right)
}\left(  \rho^{2}\right)  ,\label{give M in chan rev proof}%
\end{align}
where $\mathcal{H}_{\text{in}}$ is a copy of $\mathcal{H}$ and where our steps
(in sequence) used $\left(  \ref{mu is muhat for reversal}\right)  $,
cyclicity of the trace, $\left(  \ref{eq psi sub rho purifies rhobar}\right)
$, and $\left(  \ref{eq for quadratic weighting as it appears in proof}%
\right)  $. That $\tilde{\Lambda}^{\ }=\Lambda$ now follows from equation
$\ref{gamma for overlap bound}$.

\textbf{Third claim:} By equation $\ref{eq for PGO}$%
\begin{equation}
\mathcal{R}_{\mathcal{K}\rightarrow\mathcal{L}}^{\text{QO}}\left(
\upsilon\right)  =\operatorname*{Tr}_{\mathcal{KH}^{\ast}}\left[  \left(
\hat{\mu}_{\mathcal{KH}^{\ast}}\right)  ^{2}\left(  \left(  Y^{-1/2^{+}%
}\upsilon Y^{-1/2^{+}}\right)  _{\mathcal{K}\rightarrow\mathcal{K}}%
\otimes\left\vert \psi_{\rho}\right\rangle _{\mathcal{H\mathcal{H}}^{\ast}%
}\left\langle \psi_{\rho}\right\vert \right)  \right]
\label{recovery surprisingly tricky to proceed}%
\end{equation}
where by $\left(  \ref{M in maximum overlap}\right)  $ and $\left(
\ref{give M in chan rev proof}\right)  $%
\begin{equation}
Y_{\mathcal{K}\rightarrow\mathcal{K}}=\mathcal{A}^{\left(  2,\rho\right)
}\left(  \rho^{2}\right)  \text{,}\label{found Y in recovery}%
\end{equation}
But by $\left(  \ref{eq for canonical purification}\right)  $ \& $\left(
\ref{basic eq for double kets}\right)  $, $\left(
\ref{eq for quadratic weighting as it appears in proof}\right)  $, and
$\left(  \ref{Basis free PT identity}\right)  $ one has%
\begin{align}
&  \operatorname*{Tr}_{\mathcal{H}^{\ast}}\left[  \left(  \hat{\mu
}_{\mathcal{KH}^{\ast}}\right)  ^{2}\left\vert \psi_{\rho}\right\rangle
_{\mathcal{H\mathcal{H}}^{\ast}}\left\langle \psi_{\rho}\right\vert \right]
\nonumber\\
&  =\operatorname*{Tr}_{\mathcal{H}^{\ast}}\left[  \left(  \mathcal{A}%
_{\mathcal{H}_{\text{in}}\rightarrow\mathcal{K}}\left(  \left\vert \psi_{\rho
}\right\rangle _{\mathcal{H}_{\text{in}}\mathcal{H}^{\ast}}\left\langle
\psi_{\rho}\right\vert \right)  \right)  ^{2}\bar{\rho}_{\mathcal{H}^{\ast}%
}^{\,-1/2^{+}}\left.  \left\vert \rho\right\rangle \!\right\rangle
_{\mathcal{H\mathcal{H}}^{\ast}}\left\langle \!\left\langle \rho\right\vert
\right.  \bar{\rho}_{\mathcal{H}^{\ast}}^{\,-1/2^{+}}\right] \nonumber\\
&  =\left(  \mathcal{A}_{\mathcal{H}_{\text{in}}\rightarrow\mathcal{K}%
}^{\left(  2,\rho\right)  }\otimes%
\openone
_{\mathcal{H}}\right)  \left(  \operatorname*{PT}_{B^{2}\left(  \mathcal{H}%
^{\ast}\right)  \rightarrow B^{2}\left(  \mathcal{H}_{\text{in}}\right)
}\left(  \left.  \left\vert \rho\right\rangle \!\right\rangle
_{\mathcal{H\mathcal{H}}^{\ast}}\left\langle \!\left\langle \rho\right\vert
\right.  \right)  \right) \nonumber\\
&  =\left(  \mathcal{A}_{\mathcal{H}_{\text{in}}\rightarrow\mathcal{K}%
}^{\left(  2,\rho\right)  }\otimes%
\openone
_{\mathcal{H}}\right)  \left(  \rho_{\mathcal{H}\rightarrow\mathcal{H}%
_{\text{in}}}\otimes\rho_{\mathcal{H}_{\text{in}}\rightarrow\mathcal{H}%
}\right)  ,\label{tricky substitution}%
\end{align}
where $\mathcal{H}_{\text{in}}$ is a copy of $\mathcal{H}$. So setting
\begin{equation}
X_{\mathcal{K}\rightarrow\mathcal{K}}=Y^{-1/2^{+}}\upsilon Y^{-1/2^{+}%
}=\left(  \mathcal{A}^{\left(  2,\rho\right)  }\left(  \rho^{2}\right)
\right)  ^{-1/2^{+}}\upsilon\left(  \mathcal{A}^{\left(  2,\rho\right)
}\left(  \rho^{2}\right)  \right)  ^{-1/2^{+}}%
,\label{X operator to simplify recovery}%
\end{equation}
and substituting $\left(  \ref{tricky substitution}\right)  $ into $\left(
\ref{recovery surprisingly tricky to proceed}\right)  $ gives
\begin{align}
\mathcal{\tilde{R}}_{\mathcal{K}\rightarrow\mathcal{L}}\left(  \upsilon
\right)   &  =\operatorname*{Tr}_{\mathcal{\mathcal{KH}}^{\ast}}\left(
\left(  \hat{\mu}_{\mathcal{KH}^{\ast}}\right)  ^{2}\left(  X_{\mathcal{K}%
\rightarrow\mathcal{K}}\otimes\left\vert \psi_{\rho}\right\rangle
_{\mathcal{\mathcal{HH}}^{\ast}}\left\langle \psi_{\rho}\right\vert \right)
\right) \nonumber\\
&  =\operatorname*{Tr}_{\mathcal{\mathcal{K}}}\left(  \operatorname*{Tr}%
_{\mathcal{H}^{\ast}}\left(  \left(  \hat{\mu}_{\mathcal{KH}^{\ast}}\right)
^{2}\left\vert \psi_{\rho}\right\rangle _{\mathcal{\mathcal{HH}}^{\ast}%
}\left\langle \psi_{\rho}\right\vert \right)  X_{\mathcal{K}\rightarrow
\mathcal{K}}\right) \nonumber\\
&  =\operatorname*{Tr}_{\mathcal{\mathcal{K}}}\left(  \mathcal{A}%
_{\mathcal{H}_{\text{in}}\rightarrow\mathcal{K}}^{\left(  2,\rho\right)
}\left(  \rho_{\mathcal{H}\rightarrow\mathcal{H}_{\text{in}}}\otimes
\rho_{\mathcal{H}_{\text{in}}\rightarrow\mathcal{H}}\right)  X_{\mathcal{K}%
\rightarrow\mathcal{K}}\right) \nonumber\\
&  =\operatorname*{Tr}_{\mathcal{H}_{\text{in}}}\left[  \left(  \rho
_{\mathcal{H}\rightarrow\mathcal{H}_{\text{in}}}\otimes\rho_{\mathcal{H}%
_{\text{in}}\rightarrow\mathcal{H}}\right)  \left(  \left(  \mathcal{A}%
^{\left(  2,\rho\right)  }\right)  _{\mathcal{K}\rightarrow\mathcal{H}%
_{\text{in}}}^{\dag}\left(  X\right)  \right)  \right] \nonumber\\
&  =\rho_{\mathcal{H}_{\text{in}}\rightarrow\mathcal{H}}\left(  \left(
\mathcal{A}^{\left(  2,\rho\right)  }\right)  _{\mathcal{K}\rightarrow
\mathcal{H}_{\text{in}}}^{\dag}\left(  X\right)  \right)  \rho_{\mathcal{H}%
\rightarrow\mathcal{H}_{\text{in}}}\text{.}%
\end{align}
This proves the claim.

The inequality $\Lambda\leq\operatorname*{Tr}\mathcal{A}\left(  \rho\right)  $
follows from $\left(  \ref{eq recovery channel bounds}\right)  $.
\end{proof}

\bigskip

The following proposition puts our recovery bounds into a form closer to the
nearly simultaneously-appearing bounds of B\'{e}ny and Oreshkov (Theorem
\ref{corrolary 3 of beny}, above):

\begin{proposition}
\label{prop form near beny}Suppose that $\rho$ is a density on $\mathcal{H}$
and that the quantum operation $\mathcal{A}:B^{1}\left(  \mathcal{H}\right)
\rightarrow B^{1}\left(  \mathcal{K}\right)  $ has a Kraus decomposition of
the form%
\begin{equation}
\mathcal{A}\left(  \mu\right)  =%
{\displaystyle\sum}
F_{k}\mu F_{k}^{\dag}\text{,}%
\end{equation}
where the $F_{k}$ are not constrained to satisfy any orthogonality conditions.
Then for $\mu\in B^{1}\left(  \operatorname*{supp}\left(  \rho\right)
\right)  $ one has%
\begin{equation}
\mathcal{A}^{\left(  2,\rho\right)  }\left(  \mu\right)  =%
{\displaystyle\sum_{k\ell}}
F_{k}\mu F_{\ell}^{\dag}\times\operatorname*{Tr}\left(  F_{k}^{\dag}F_{\ell
}\rho\right) \label{eq for quadratic weighting given Kraus operators}%
\end{equation}

\end{proposition}

\begin{proof}
Since both sides of $\left(
\ref{eq for quadratic weighting given Kraus operators}\right)  $ are linear in
$\mu$, we may assume without loss of generality that $\mu$ is positive
semidefinite. The conclusion follows using equations
\ref{eq for quadratic weighting as it appears in proof},
\ref{eq relationship between psi sub rho and psi sub mu},
\ref{eq for canonical purification} \& \ref{basic eq for double kets}, and
\ref{eq double ket is an isometry} \&
\ref{eq used to construct canonical purif} (in said order):
\begin{align*}
\mathcal{A}^{\left(  2,\rho\right)  }\left(  \mu\right)   &
=\operatorname*{Tr}_{\mathcal{H}^{\ast}}\left(  \overline{\mu^{1/2}%
\rho^{-1/2^{+}}}\times\left(
{\displaystyle\sum_{k}}
F_{k}\left\vert \psi_{\rho}\right\rangle _{\mathcal{HH}^{\ast}}\left\langle
\psi_{\rho}\right\vert F_{k}^{\dag}\right)  ^{2}\times\overline{\rho
^{-1/2^{+}}\mu^{1/2}}\right) \\
&  =\operatorname*{Tr}_{\mathcal{H}^{\ast}}\left(
{\displaystyle\sum_{k\ell}}
F_{k}\left\vert \psi_{\mu}\right\rangle _{\mathcal{HH}^{\ast}}\left\langle
\psi_{\rho}\right\vert F_{k}^{\dag}F_{\ell}\left\vert \psi_{\rho}\right\rangle
_{\mathcal{HH}^{\ast}}\left\langle \psi_{\mu}\right\vert F_{\ell}^{\dag
}\right) \\
&  =\operatorname*{Tr}_{\mathcal{H}^{\ast}}\left(
{\displaystyle\sum_{k\ell}}
\left.  \left\vert F_{k}\sqrt{\mu}\right\rangle \!\right\rangle _{\mathcal{KH}%
^{\ast}}\left\langle \!\left\langle F_{k}\sqrt{\rho}\right\vert \right.
\left.  \left\vert F_{\ell}\sqrt{\rho}\right\rangle \!\right\rangle
_{\mathcal{KH}^{\ast}}\left\langle \!\left\langle F_{\ell}\sqrt{\mu
}\right\vert \right.  \right) \\
&  =%
{\displaystyle\sum_{k\ell}}
F_{k}\mu F_{\ell}^{\dag}\times\operatorname*{Tr}\left(  F_{k}^{\dag}F_{\ell
}\rho\right)  \text{.}%
\end{align*}

\end{proof}

\subsection{The relationship between the Quadratic Recovery and Barnum and
Knill's reversal \label{subsection comparison with barnum knill reversal}}

As we have already seen in equation
\ref{eq quad reversal also satisfy barnum knill bounds}, the quadratic
reversal $\mathcal{R}^{QR}$ and the simple lower bound of $\left(
\ref{eq recovery channel bounds}\right)  $ are both sufficiently accurate to
also satisfy the tightness relation $\left(  \ref{Barnum Knill estimate}%
\right)  $ of Barnum and Knill. This section makes a brief comparison of the
quadratic reversal operation with the reversal map of Barnum and Knill (for
the special case of non-average entanglement fidelity) in light of the
relationship between the quadratic measurement and the PGM.

Re-expressing the elements of the ensemble $\left(
\ref{a priori normed ensemble to distinguish}\right)  $ as $\rho_{k}=p_{k}%
\hat{\rho}_{k}$, where $\operatorname*{Tr}\left(  \hat{\rho}_{k}\right)  =1$
and $p_{k}=\operatorname*{Tr}\rho_{k} $ is the chance that $\hat{\rho}_{k}$
appears, the formulas for the \textquotedblleft pretty good\textquotedblright%
\ and quadratically-weighted measurements become%
\begin{align*}
M_{k}^{\text{PGM}} &  =\left(
{\displaystyle\sum}
p_{\ell}\hat{\rho}_{\ell}\right)  ^{-1/2^{+}}p_{k}\hat{\rho}_{k}\left(
{\displaystyle\sum}
p_{\ell}\hat{\rho}_{\ell}\right)  ^{-1/2^{+}}\\
M_{k}^{\text{QW}} &  =\left(
{\displaystyle\sum}
p_{\ell}^{2}\hat{\rho}_{\ell}^{2}\right)  ^{-1/2^{+}}p_{k}^{2}\hat{\rho}%
_{k}^{2}\left(
{\displaystyle\sum}
p_{\ell}^{2}\hat{\rho}_{\ell}^{2}\right)  ^{-1/2^{+}}\text{.}%
\end{align*}
In particular, to get from the pretty-good measurement to the quadratic
measurement, one replaces all probabilities and density matrices by their squares.

A simple examination of the formulas $\left(  \ref{Barnum Knill reversal}%
\right)  $ and $\left(  \ref{eq quadratic recovery channel}\right)  $ shows
that a similar relationship exists between the entanglement fidelity case of
the Barnum-Knill reversal $\mathcal{R}^{\text{BK}}$ and the
quadratically-weighted reversal $\mathcal{R}^{\text{QR}}$. Note that the
corresponding probabilities $p_{k}$, which must be replaced by their squares,
are viewed as being hidden in the $\rho$-Kraus decomposition $\left(
\ref{rho kraus decomp}\right)  $ of the reversed map $\mathcal{A}$.

In \cite{Tyson Error rates of quantum Belavkin measurements} various
weightings for Belavkin pure-state square-root measurements were compared, and
it was argued that Holevo's quadratically-weighted measurement had qualitative
and quantitative advantages over the linearly weighted PGM. Based on analogy,
we conjecture that $\mathcal{R}^{\text{QR}}$ will typically (but not always)
outperform $\mathcal{R}^{\text{BK}}$.

\subsubsection{Depolarizing noise and the quadratic transpose channel
\label{section the case of depolarizing noise}}

It is perhaps interesting to quantitatively compare the actions of the
Barnum-Knill reversal $\mathcal{R}^{\text{BK}}$ with the quadratic recovery
$\mathcal{R}^{\text{QR}}$ in the simplest special case, in which depolarizing
noise
\begin{equation}
\left(  \mathcal{A}_{p}\right)  _{\mathcal{H}\rightarrow\mathcal{H}}\left(
\mu\right)  =p\times\frac{%
\openone
}{\dim\mathcal{H}}\operatorname*{Tr}\left(  \mu\right)  +\left(  1-p\right)
\times\mu,\text{\ \ \ \ }p\in\left[  0,1\right]
\end{equation}
acts on half of a maximally-entangled state. (Note that the $\rho=%
\openone
/\dim\mathcal{H}$ case of $\mathcal{R}^{\text{BK}}$, also known as the
\textit{transpose channel }\cite{Petz quantum entropy and its use}, has
recently \cite{Ng Mandayam simple approach to approximate QEC} been employed
in the study of approximate quantum error correction.) For $\dim\mathcal{H}%
>1$, one easily obtains%
\begin{align}
\mathcal{R}_{p}^{\text{BK}} &  =\mathcal{A}_{p}\\
\mathcal{R}_{p}^{\text{QR}} &  =\mathcal{A}_{f\left(  p,\dim\mathcal{H}%
\right)  }\text{,}%
\end{align}
where%
\begin{equation}
f\left(  p,\dim\mathcal{H}\right)  :=\frac{p^{2}}{\left(  1-p\right)
^{2}\left(  \dim\mathcal{H}\right)  ^{2}+\left(  2-p\right)  p}%
\end{equation}
satisfies%
\begin{equation}
f\left(  p,\dim\mathcal{H}\right)  \leq f\left(  p,1\right)  =p^{2}\text{.}%
\end{equation}
In particular, both recovery operations \textquotedblleft
correct\textquotedblright\ depolarization errors by committing further
depolarization. Fortunately, however, when $p<1$ \textit{the quadratic
recovery depolarizes with lower probability than the transpose channel},
\textit{especially when} $\dim\mathcal{H}$ \textit{is large or }$p$ \textit{is
small.}

A more detailed quantitative comparison of $\mathcal{R}^{\text{BK}}$,
$\mathcal{R}^{\text{QR}}$, and of reversals of other possible weightings
(perhaps generalizing the cubically-weighted measurement of \cite{Wehner
thesis, Wehner State discrimination with post-measurement information}) will
be left for future work.

\section{Conclusion and future directions}

We have generalized the iterative schemes of
Je\v{z}ek-\v{R}eh\'{a}\v{c}ek-Fiur\'{a}\v{s}ek \cite{Jezek Rehacek and
Fiurasek Finding optimal strategies for minimum error quantum state
discrimination}, Je\v{z}ek-Fiur\'{a}\v{s}ek-Hradil \cite{Jezek Fiurasek Hradil
Quantum inference of states and processes, Hradil et al Maximum Likelihood
methods in quantum mechanics}, and Reimpell-Werner \cite{Reimpell Werner,
Reimpell Thesis}. Using an abstract framework, \textquotedblleft small
angle\textquotedblright\ guesses were employed to construct concise two-sided
bounds for minimum-error quantum detection, maximum overlap, quantum
conditional min-entropy, and the reversibility of quantum dynamics. An
approximately-optimal channel reversal and overlap operation were derived. The
resulting bounds were sufficiently tight to also satisfy the tightness
relations of Barnum and Knill \cite{Barnum Knill UhOh}, although our methods
more generally allowed the target state and the input state to differ. Our
recovery operation was found to be a significant improvement of the transpose
channel in the simple case of depolarizing noise acting on half of a
maximally-entangled state.

As a direction for future study, we note that Barnum and Knill constructed an
approximate reversal operation in the more general sense of \textit{average}
entanglement fidelity, albeit with commutativity assumptions of unknown
necessity. A remaining open question is whether one can generalize our
quadratic reversal construction to this case of \emph{average} entanglement
fidelity, and whether these commutativity assumptions may be removed. More
generally, one may ask how to obtain estimates for the maximum overlap problem
without our assumed purity of the target state. The principle difficulty in
answering both of these questions is in finding an appropriate
\textquotedblleft small angle guess,\textquotedblright\ in the sense of lemma
\ref{directional iterate lemma}.

Another future direction, in which we have made recent progress \cite{Tyson in
preparation}, is to employ matrix monotonicity to obtain bounds for the
maximum overlap problem, including its special cases of channel reversibility
and quantum conditional min-entropy.%

\section*{Appendix A: Canonical Stinespring dilations
\label{stinespring appendix}
\addcontentsline{toc}{section}{Appendix A: Canonical Stinespring dilations}
}%

Using only the square root function and the natural isomorphisms of Section
\ref{section basis free}, one may construct Stinespring dilations which are
independent of any choice of a basis:

\begin{definition}
Let $\mathcal{R}:$ $B^{1}\left(  \mathcal{K}\right)  \rightarrow B^{1}\left(
\mathcal{L}\right)  $ be a completely positive map, with $\mathcal{K}$
finite-dimensional. The \textbf{canonical environment }is given by%
\begin{equation}
\mathcal{E=L}_{\mathcal{E}}^{\ast}\otimes\mathcal{K}_{\mathcal{E}%
},\label{can env 2}%
\end{equation}
where $\mathcal{L}_{\mathcal{E}}^{\ast}$ and $\mathcal{K}_{\mathcal{E}}$ are
copies of $\mathcal{L}^{\ast}$ and $\mathcal{K}$, respectively. The
\textbf{canonical Stinespring dilation }$U_{\mathcal{R}}$ of $\mathcal{R}$ is
the linear transformation $U_{\mathcal{R}}:\mathcal{K}\rightarrow
\mathcal{L}\otimes\mathcal{E}$ such that $\left.  \left\vert U_{\mathcal{R}%
}\right\rangle \!\right\rangle _{\mathcal{\mathcal{LE}K}^{\ast}}=\left.
\left\vert U\right\rangle \!\right\rangle _{\mathcal{LK}^{\ast}\mathcal{L}%
_{\mathcal{E}}^{\ast}\mathcal{K}_{\mathcal{E}}}$ is the canonical purification
$\left(  \ref{eq for canonical purification}\right)  $ of the Choi matrix
$\mathcal{\tilde{R}}=\mathcal{R}\left(  \left.  \left\vert
\openone
\right\rangle \!\right\rangle _{\mathcal{KK}^{\ast}}\left\langle
\!\left\langle
\openone
\right\vert \right.  \right)  $.
\end{definition}

That $U_{\mathcal{R}}$ is a bona fide purification of $\mathcal{R}$ follows
from the following lemma:

\begin{lemma}
\label{lemma dilation is purification}Let $\mathcal{R}:B^{1}\left(
\mathcal{K}\right)  \rightarrow B^{1}\left(  \mathcal{L}\right)  $ be a
quantum operation, with $\mathcal{K}$ finite-dimensional. Then $U_{\mathcal{K}%
\rightarrow\mathcal{\mathcal{LE}}}$ is a Stinespring dilation of $\mathcal{R}$
iff $\left.  \left\vert U\right\rangle \!\right\rangle
_{\mathcal{L\mathcal{\mathcal{E}K}}^{\ast}}$ is a purification of the Choi
matrix $\mathcal{\tilde{R}}$.
\end{lemma}

\begin{proof}
Suppose that $U$ dilates $\mathcal{R}$. Then by equation
$\ref{basic eq for double kets}$%
\[
\mathcal{\tilde{R}}=\operatorname*{Tr}_{\mathcal{E}}U_{\mathcal{K}%
\rightarrow\mathcal{\mathcal{LE}}}\left.  \left\vert
\openone
\right\rangle \!\right\rangle _{\mathcal{KK}^{\ast}}\left\langle
\!\left\langle
\openone
\right\vert \right.  \left(  U_{\mathcal{K}\rightarrow\mathcal{\mathcal{LE}}%
}\right)  ^{\dag}=\operatorname*{Tr}_{\mathcal{E}}\left.  \left\vert
U\right\rangle \!\right\rangle _{\mathcal{\mathcal{\mathcal{LE}K}}^{\ast}%
}\left\langle \!\left\langle U\right\vert \right.  \text{,}%
\]
so $\left.  \left\vert U\right\rangle \!\right\rangle $ purifies
$\mathcal{\tilde{R}}$.

Conversely, suppose that $\left.  \left\vert U\right\rangle \!\right\rangle
_{\mathcal{\mathcal{\mathcal{LE}K}}^{\ast}}$ purifies $\mathcal{\tilde{R}}$,
and let $\upsilon\in B^{1}\left(  \mathcal{K}\right)  $ be a density matrix.
Then by equations $\ref{eq used to construct canonical purif}$,
$\ref{basic eq for double kets}$, \ref{eq for canonical purification}, and
\ref{eq psi sub rho purifies rho},%
\begin{align}
\operatorname*{Tr}_{\mathcal{E}}\left(  U_{\mathcal{K}\rightarrow
\mathcal{\mathcal{LE}}}\upsilon_{\mathcal{K}}\left(  U_{\mathcal{K}%
\rightarrow\mathcal{\mathcal{LE}}}\right)  ^{\dag}\right)   &
=\operatorname*{Tr}_{\mathcal{EK}^{\ast}}\left(  \left(  \bar{\upsilon}^{\dag
}\right)  _{\mathcal{K}^{\ast}}^{1/2}\left.  \left\vert U\right\rangle
\!\right\rangle _{\mathcal{\mathcal{LE}K}^{\ast}}\left\langle \!\left\langle
U\right\vert \right.  \left(  \bar{\upsilon}^{\dag}\right)  _{\mathcal{K}%
^{\ast}}^{1/2}\right) \nonumber\\
&  =\operatorname*{Tr}_{\mathcal{K}^{\ast}}\left[  \left(  \bar{\upsilon
}^{\dag}\right)  _{\mathcal{K}^{\ast}}^{1/2}\mathcal{R}\left(  \left.
\left\vert
\openone
\right\rangle \!\right\rangle _{\mathcal{KK}^{\ast}}\left\langle
\!\left\langle
\openone
\right\vert \right.  \right)  \left(  \bar{\upsilon}^{\dag}\right)
_{\mathcal{K}^{\ast}}^{1/2}\right] \nonumber\\
&  =\mathcal{R}\left(  \operatorname*{Tr}_{\mathcal{K}^{\ast}}\left\vert
\psi_{\upsilon}\right\rangle \left\langle \psi_{\upsilon}\right\vert \right)
=\mathcal{R}\left(  \upsilon\right)  \text{.}%
\end{align}

\end{proof}%

\section*{Appendix B: Reimpell-Werner iteration as directional iteration
\label{appendix Reimpell Werner iteration}
\addcontentsline{toc}{section}%
{Appendix B: Reimpell-Werner iteration as directional iteration}
}%

The purpose of this section is to verify that Reimpell-Werner iteration
(introduced in section \ref{section reimpell werner iteration}) for CP maps
corresponds to directional iteration of the corresponding Stinespring dilations.

One may re-express the maximized functional $f\left(  \mathcal{R}\right)  $ of
equation $\left(  \ref{eq reimpell functional represented by F}\right)  $ as%

\begin{equation}
f\left(  \mathcal{R}\right)  =\left\Vert U_{\mathcal{R}}\right\Vert _{F}%
^{2}\text{,}%
\end{equation}
where $U_{\mathcal{R}}$ is a Stinespring dilation of the CP map $\mathcal{R}%
:B^{1}\left(  \mathcal{K}\right)  \rightarrow B^{1}\left(  \mathcal{L}\right)
$ and the seminorm is defined by

\begin{definition}
Let $\mathcal{E}=\mathcal{L}_{\mathcal{E}}^{\ast}\otimes\mathcal{K}%
_{\mathcal{E}}$ be the canonical environment $\left(  \ref{can env 2}\right)
$ for quantum operations from $\mathcal{K}$ to $\mathcal{L}$. For operators
$U,W:\mathcal{K}\rightarrow\mathcal{L}\otimes\mathcal{E}$, define the
semidefinite inner product
\begin{equation}
\left\langle U,W\right\rangle _{F}=\left\langle \!\left\langle U\right\vert
\right.  _{\mathcal{LKE}^{\ast}}F_{\mathcal{LK}^{\ast}\rightarrow
\mathcal{LK}^{\ast}}\left.  \left\vert W\right\rangle \!\right\rangle
_{\mathcal{LKE}^{\ast}}\text{,}\label{Reimpell F inner product}%
\end{equation}
where $A\mapsto\left.  \left\vert A\right\rangle \!\right\rangle $ is the
isomorphism of equation \ref{stupid def of double ket}. Let $V_{F}=\left\{
\left.  U~\right\vert ~\left\Vert U\right\Vert _{F}<\infty\right\}  $, on
which $\left\langle \bullet,\bullet\right\rangle _{F}$ is a well-defined
semidefinite inner product. Let $S=\left\{  \left.  U~\right\vert ~\left\Vert
U\right\Vert \leq1\right\}  $.
\end{definition}

\begin{theorem}
[Reimpell-Werner iteration is directional iteration]Suppose that
$U_{\mathcal{K}\rightarrow\mathcal{\mathcal{LE}}}$ is a Stinespring dilation
of a CP map $\mathcal{R}:B^{1}\left(  \mathcal{K}\right)  \rightarrow
B^{1}\left(  \mathcal{L}\right)  $. Then $U\in V_{F}$ has a directional
iterate $U^{\left(  +\right)  }\in S$ which dilates the Reimpell-Werner
iterate $\mathcal{R}^{\oplus}$ of Def. \ref{def of reimpell werner iterates}.
\end{theorem}

\begin{proof}
Let $X_{\mathcal{K}\rightarrow\mathcal{\mathcal{LE}}}$ be the operator defined
by%
\[
\left.  \left\vert X_{\mathcal{K}\rightarrow\mathcal{\mathcal{LE}}%
}\right\rangle \!\right\rangle =F_{\mathcal{LK}^{\ast}\rightarrow
\mathcal{LK}^{\ast}}\left.  \left\vert U\right\rangle \!\right\rangle
_{\mathcal{\mathcal{\mathcal{LE}K}}^{\ast}}\text{.}%
\]
Then by equations $\ref{Reimpell F inner product},$
$\ref{equation sup trace A dag U}$, and $\ref{unitary maximizer equation}$,%
\[
\max_{W\in S}\operatorname{Re}\left\langle W,U\right\rangle _{F}=\max_{W\in
S}\operatorname{Re}\operatorname*{Tr}\left(  W_{\mathcal{K}\rightarrow
\mathcal{\mathcal{LE}}}\right)  ^{\dag}X_{\mathcal{K}\rightarrow
\mathcal{\mathcal{LE}}}=\left\Vert X_{\mathcal{K}\rightarrow
\mathcal{\mathcal{LE}}}\right\Vert _{1},
\]
with maximizer $W=U^{\left(  +\right)  }$ given by%
\[
U^{\left(  +\right)  }=X\left(  X^{\dag}X\right)  ^{-1/2^{+}}\text{.}%
\]
By equations \ref{basic eq for double kets},
\ref{eq inner product out first factor}, and
\ref{Reimpell def of gamma for iterate} one has%
\begin{align*}
\left.  \left\vert U^{\left(  +\right)  }\right\rangle \!\right\rangle
_{\mathcal{\mathcal{\mathcal{LE}K}}^{\ast}} &  =\left(  \overline{X^{\dag}%
X}\right)  _{\mathcal{K}^{\ast}\rightarrow\mathcal{K}^{\ast}}^{-1/2^{+}%
}\left.  \left\vert X\right\rangle \!\right\rangle
_{\mathcal{\mathcal{\mathcal{LE}K}}^{\ast}}\\
&  =\left(  \operatorname*{Tr}_{\mathcal{\mathcal{LE}}}\left.  \left\vert
X\right\rangle \!\right\rangle _{\mathcal{\mathcal{\mathcal{LE}K}}^{\ast}%
}\left\langle \!\left\langle X\right\vert \right.  \right)  ^{-1/2}\left.
\left\vert X\right\rangle \right\rangle _{\mathcal{\mathcal{\mathcal{LE}K}%
}^{\ast}}\\
&  =\left(  \operatorname*{Tr}_{\mathcal{\mathcal{LE}}}F_{\mathcal{LK}^{\ast
}\rightarrow\mathcal{LK}^{\ast}}\left.  \left\vert U\right\rangle
\!\right\rangle _{\mathcal{\mathcal{\mathcal{LE}K}}^{\ast}}\left\langle
\!\left\langle U\right\vert \right.  F_{\mathcal{LK}^{\ast}\rightarrow
\mathcal{LK}^{\ast}}\right)  ^{-1/2^{+}}F_{\mathcal{LK}^{\ast}\rightarrow
\mathcal{LK}^{\ast}}\left.  \left\vert U\right\rangle \!\right\rangle
_{\mathcal{\mathcal{\mathcal{LE}K}}^{\ast}}\\
&  =\Gamma^{-1/2^{+}}F_{\mathcal{LK}^{\ast}\rightarrow\mathcal{LK}^{\ast}%
}U_{\mathcal{K}\rightarrow\mathcal{\mathcal{LE}}}\left.  \left\vert
\openone
\right\rangle \!\right\rangle _{\mathcal{K\mathcal{K}}^{\ast}}\text{.}%
\end{align*}
It follows from Lemma $\ref{lemma dilation is purification}$ and Eq.
\ref{Reimpell def of iterate} that $U^{\left(  +\right)  }$ dilates
$\mathcal{R}^{\oplus}$.
\end{proof}

\pagebreak%

\section
*{Appendix C: The relationship between overlap bounds and state distinguishability
\label{appendix relationship to the measurement bounds}
\addcontentsline{toc}{section}%
{Appendix C: The relationship between overlap bounds and state distinguishability}
}%

As remarked in section \ref{section the restricted maximum overlap}, Theorems
1 and 2 of \cite{Konig Renner Schaffner Operational meaning of min and max
entropy} (see also equations
\ref{premaximized identity of koenig renner shafner}%
-\ref{form of Ropt used to kill root m}, below) imply that minimum-error
distinguishability of a finite collection of quantum states $\mathcal{E}%
=\left\{  \rho_{k}\right\}  _{k=1,\ldots,m}$ may be expressed in terms of
restricted maximum overlap:%
\begin{equation}
P_{\text{succ}}\left(  M^{\text{opt}}\right)  =m\times\max_{\mathcal{R}%
_{\mathcal{H}\rightarrow\left(  \mathbb{C}^{m}\right)  ^{\ast}}}\left\langle
\phi_{\mathbb{C}^{m}\left(  \mathbb{C}^{m}\right)  ^{\ast}}\right\vert
\mathcal{R}_{\mathcal{H}\rightarrow\left(  \mathbb{C}^{m}\right)  ^{\ast}%
}\left(  \mu_{\mathcal{H}\mathbb{C}^{m}}\right)  \left\vert \phi
_{\mathbb{C}^{m}\mathbb{C}^{m\ast}}\right\rangle
.\label{KRS version of optimal success}%
\end{equation}
Here the vector $\phi\in\mathbb{C}^{m}\otimes\left(  \mathbb{C}^{m}\right)
^{\ast}$ is the maximally-mixed state%
\begin{equation}
\left\vert \phi_{\mathbb{C}^{m}\left(  \mathbb{C}^{m}\right)  ^{\ast}%
}\right\rangle =\frac{1}{\sqrt{m}}\left.  \left\vert
\openone
\right\rangle \!\right\rangle _{\mathbb{C}^{m}\mathbb{C}^{m\ast}}:=\frac
{1}{\sqrt{m}}%
{\displaystyle\sum}
\left\vert k\right\rangle _{\mathbb{C}^{m}}\left\vert \bar{k}\right\rangle
_{\mathbb{C}^{m\ast}}\text{,}\label{KRS max entangled}%
\end{equation}
and $\mu\in B^{1}\left(  \mathcal{H}\otimes\mathbb{C}^{m}\right)  $ is the
\textquotedblleft quantum-classical\textquotedblright\ state%
\begin{equation}
\mu_{\mathcal{H}\otimes\mathbb{C}^{m}}=%
{\displaystyle\sum_{k=1}^{m}}
\rho_{k}\otimes\left\vert k\right\rangle _{\mathbb{C}^{m}}\left\langle
k\right\vert \text{,}\label{class quantum state}%
\end{equation}
where the $\rho_{k}$ are normalized as in Definition
\ref{definition containing ensemble E}.

If one applies the overlap bounds of Theorem
\ref{Theorem two sided overlap estimates} (or the $s=0$ case of Corollary
\ref{corollary bounding min entropy} combined with Eq.
$\ref{eq min entropy as overlap}$), one obtains%
\begin{align}
\left(  \operatorname*{Tr}\sqrt{%
{\displaystyle\sum}
\rho_{k}^{2}}\right)  ^{2}\leq P_{\text{succ}}\left(  M^{\text{QW}}\right)
\leq P_{\text{succ}}\left(  M^{\text{opt}}\right)   &  \leq\sqrt{m}\left\Vert
\left(  \mathcal{R}^{\text{opt}}\right)  ^{\dag}\left(  \left\vert
\phi\right\rangle \left\langle \phi\right\vert \right)  \right\Vert _{\infty
}\times\operatorname*{Tr}\sqrt{%
{\displaystyle\sum}
\rho_{k}^{2}}\nonumber\\
&  \leq\sqrt{m}\times\operatorname*{Tr}\sqrt{%
{\displaystyle\sum}
\rho_{k}^{2}}.\label{overlap applied to meas}%
\end{align}
In particular, if one neglects the $\left\Vert \mathcal{R}^{\dag}\right\Vert
_{\infty}$ factor in the fourth expression of this estimate then one picks up
a spurious factor of $\sqrt{m}$ not appearing in the bounds of Theorem
\ref{Theorem gen Holevo Curlander bounds}. (Weakness of the upper bound is not
surprising, since $\phi$ and $\mu$ are generally not \textquotedblleft
reasonably overlappable.\textquotedblright)

In order to show how one may apply the fourth term of the overlap estimate
$\left(  \ref{overlap applied to meas}\right)  $, we give another proof of
Theorem \ref{Theorem gen Holevo Curlander bounds}. It is hoped that similar
methods may lead to sharper upper in other instances of maximum overlap or
conditional min-entropy.\bigskip

\begin{proof}
[An \textquotedblleft overlap proof\textquotedblright\ of Theorem
\ref{Theorem gen Holevo Curlander bounds}]We restrict consideration to the
case $\mathcal{E}=\left\{  \rho_{k}\right\}  _{k=1,\ldots,m}$. Given a quantum
operation $\mathcal{R}_{\mathcal{H}\rightarrow\mathbb{C}^{M}}$ one has the
identity%
\begin{equation}
m\times\left\langle \phi_{\mathbb{C}^{m}\left(  \mathbb{C}^{m}\right)  ^{\ast
}}\right\vert \mathcal{R}_{\mathcal{H}\rightarrow\left(  \mathbb{C}%
^{m}\right)  ^{\ast}}\left(  \mu_{\mathcal{H}\mathbb{C}^{m}}\right)
\left\vert \phi_{\mathbb{C}^{m}\mathbb{C}^{m\ast}}\right\rangle
=P_{\text{succ}}\left(  M^{\mathcal{R}}\right)
,\label{premaximized identity of koenig renner shafner}%
\end{equation}
where $\mu\in B^{1}\left(  \mathcal{H}\otimes\mathbb{C}^{m}\right)  $ and
$\phi\in\mathbb{C}^{m}\otimes\mathbb{C}^{m\ast}$ are as in equations $\left(
\ref{KRS max entangled}\right)  $-$\left(  \ref{class quantum state}\right)  $
and where the POVM $M^{\mathcal{R}}$ corresponding to the operation
$\mathcal{R}$ is given by%
\begin{equation}
M_{k}^{\mathcal{R}}:=\left(  \mathcal{R}\right)  _{\left(  \mathbb{C}%
^{m}\right)  ^{\ast}\rightarrow\mathcal{H}}^{\dag}\left(  \left\vert \bar
{k}\right\rangle _{\mathbb{C}^{m\ast}}\left\langle \bar{k}\right\vert \right)
,\ \ \ k=1,\ldots,m\text{.}\label{POVM corresponding to operator}%
\end{equation}
Since any given POVM $M$ may be expressed in the form of $\left(
\ref{POVM corresponding to operator}\right)  $ for the quantum operation
$\mathcal{R}=\mathcal{R}^{M}$ given by
\begin{equation}
\mathcal{R}_{\mathcal{H}\rightarrow\left(  \mathbb{C}^{m}\right)  ^{\ast}}%
^{M}\left(  \rho\right)  :=%
{\displaystyle\sum_{k=1}^{m}}
\left\vert \bar{k}\right\rangle _{\mathbb{C}^{m\ast}}\left\langle \bar
{k}\right\vert \times\operatorname*{Tr}\left(  M_{k}\rho\right)
\text{,}\label{operation implementing subPOVM}%
\end{equation}
maximization of $\left(  \ref{premaximized identity of koenig renner shafner}%
\right)  $ over operations $\mathcal{R}$ gives the identity $\left(
\ref{KRS version of optimal success}\right)  $. Taking $M^{\text{opt}}$ to be
some optimal measurement, it follows that a maximizer of the LHS\ of $\left(
\ref{premaximized identity of koenig renner shafner}\right)  $ is given by
\begin{equation}
\mathcal{R}^{\text{opt}}=\mathcal{R}^{M^{\text{opt}}}\text{,}%
\label{form of Ropt used to kill root m}%
\end{equation}
where $M^{\text{opt}}$ is an optimal measurement. One estimates%
\begin{equation}
\left\Vert \left(  \mathcal{R}^{\text{opt}}\right)  _{\mathbb{C}^{m\ast
}\rightarrow\mathcal{H}}^{\dag}\left(  \left\vert \phi\right\rangle
_{\mathbb{C}^{m}\mathbb{C}^{m\ast}}\left\langle \phi\right\vert \right)
\right\Vert _{\infty}=\left\Vert \frac{1}{m}%
{\displaystyle\sum}
\left\vert k\right\rangle _{\mathbb{C}^{m}}\left\langle k\right\vert \otimes
M_{k}^{\text{opt}}\right\Vert _{\infty}\leq\frac{1}{m}\text{.}%
\label{resutling estimate from form of R opt}%
\end{equation}
Applying the bounds $\left(  \ref{eq overlap bound}\right)  $ to $\left(
\ref{premaximized identity of koenig renner shafner}\right)  $ yields the
chain of inequalities%
\begin{equation}
\left(  \sqrt{%
{\displaystyle\sum}
\rho_{k}^{2}}\right)  ^{2}\leq P_{\text{succ}}\left(  M^{\text{QW}}\right)
\leq P_{\text{succ}}\left(  M^{\text{opt}}\right)  \leq\sqrt{%
{\displaystyle\sum}
\rho_{k}^{2}}\leq\sqrt{m}\times\sqrt{%
{\displaystyle\sum}
\rho_{k}^{2}}\text{.}\label{yuck root m}%
\end{equation}

\end{proof}

\bigskip

\noindent\textbf{Acknowledgements:} We would like to thank Arthur Jaffe, Peter
Shor, and Chris King for their encouragement, Arthur Jaffe for suggesting a
change in presentation, Andrew Fletcher, Fr\'{e}d\'{e}ric Dupuis, Cedric
B\'{e}ny, Ognyan Oreshkov \& Renato Renner for valuable discussions, Stephanie
Wehner for pointing out the work of Ogawa and Nagaoka, the anonymous referee
for valuable comments, and a previous referee for \cite{Tyson Holevo Curlander
Bounds} for pointing out the connection to min-entropy.

\bigskip

\noindent\textbf{Note Added: }Private communication from the authors of
\cite{Beny and Oreshkov general conditions for approximate quanutm error
recovery and near optimal recovery channels} indicates that they have obtained
the quadratic recovery channel by alternative means \cite{Beny and Oreshkov in
preparation}.


\begin{thebibliography}{999}
\bibitem {subpovm}One usually requires that a POVM satisfy $%
{\textstyle\sum}
M_{k}=%
\openone
$. The relaxed condition $%
{\textstyle\sum}
M_{k}\leq%
\openone
$ allows the possibility that the POVM may fail to return an result.
(Alternatively, one may augment $\left\{  M_{k}\right\}  $ with the operator
$M_{\text{error}}=%
\openone
-%
{\textstyle\sum}
M_{k}$. The operator $M_{\text{error }}$ could perhaps serve as a useful error
flag, i.e. an indication that a state orthogonal to $\operatorname*{span}%
\left(  \mathcal{E}\right)  $ has been detected.)
\bibitem {maximum overlap as tomography}The maximum-overlap problem $\left(
\ref{eq defining maximum overlap}\right)  $ is equivalent to the
one-data-point special case of maximum-likelihood quantum process tomography.
The specific iteration is defined by equations 14 through 17 of \cite{Jezek
Fiurasek Hradil Quantum inference of states and processes}.
\bibitem {barnum knill mistake}Barnum and Knill incorrectly assert on page
2103 of \cite{Barnum Knill UhOh} that the asymptotically-optimal measurement
$\left(  \ref{holevo's pure state meas}\right)  $ introduced by Holevo in
\cite{Holovo Assym Opt Hyp Test} is equal to the \textquotedblleft pretty
good\textquotedblright\ measurement $\left(  \ref{PGM}\right)  $.
\bibitem {footnote citing antiparticles}Indeed, in Quantum Field Theory it is
natural to use dual states $\bar{\psi}\in\mathcal{H}^{\ast}$ to represent
antiparticles. See, for example, \cite{Baez and Huerta the algebra of grand
unified theories}.
\bibitem {footnote how to get Holevo cost}Indeed, if $\mathcal{E}=\left\{
p_{k}\left\vert \psi_{k}\right\rangle \left\langle \psi_{k}\right\vert
\right\}  $ is an ensemble of linearly-independent pure states $\psi_{k}$
spanning $\mathcal{H}$ then the maximizer $E=G^{\left(  +\right)  }$ of
$\operatorname{Re}\left\langle G,E\right\rangle _{\mathcal{E}}$ is of the form
$E_{k}=\left\vert \psi_{k}\right\rangle \left\langle e_{k}\right\vert $, with
$\left\{  e_{k}\right\}  $ orthornormal.$^{\text{\cite{rank 1 is on}}}$ One
may therefore express
\[
\sup_{E}\operatorname{Re}\left(  \left\langle G,E\right\rangle _{\mathcal{E}%
}\right)  =\sup_{\left\{  e_{k}\right\}  }\operatorname{Re}\left(
{\displaystyle\sum}
p_{k}\left\langle e_{k},\psi_{k}\right\rangle \right)  =1-\inf_{\left\{
e_{k}\right\}  }C^{\text{Holevo}}\left(  \left\{  e_{k}\right\}  \right)
\text{,}%
\]
where optimization is over orthonormal $\left\{  e_{k}\right\}  $. Here
$C^{\text{Holevo}}=\frac{1}{2}%
{\displaystyle\sum}
p_{k}\left\Vert \psi_{k}-e_{k}\right\Vert ^{2}$ is the \textquotedblleft
approximate cost function\textquotedblright\ which was minimized by Holevo
(\cite{Holovo Assym Opt Hyp Test}, equation 8) in his construction of
$M^{\text{QW}}$ in the case of pure states. (See also \cite{Tyson Holevo
Curlander Bounds, EldarSquareRootMeasurement}.)
\bibitem {rank 1 is on}Note that if $\left\{  \left\vert e_{k}\right\rangle
\left\langle e_{k}\right\vert \right\}  _{k=1,\ldots,\dim\mathcal{H}}$ is a
POVM then the $e_{k}$ are orthonormal. In particular, normality of the $e_{k}
$ follows from the inequalities $\left\Vert e_{k}\right\Vert ^{2}%
\leq\left\Vert
{\displaystyle\sum}
\left\vert e_{\ell}\right\rangle \left\langle e_{\ell}\right\vert \right\Vert
_{\infty}=1$ and the identity $%
{\displaystyle\sum}
\left\Vert e_{\ell}\right\Vert ^{2}=\operatorname*{Tr}%
{\displaystyle\sum}
\left\vert e_{\ell}\right\rangle \left\langle e_{\ell}\right\vert
=\dim\mathcal{H}$. Orthogonality then follows from the identity $\left\langle
e_{\ell}\right\vert \left(
{\displaystyle\sum}
\left\vert e_{k}\right\rangle \left\langle e_{k}\right\vert \right)
\left\vert e_{\ell}\right\rangle =\left\Vert e_{\ell}\right\Vert ^{2}$.
\end{thebibliography}
\end{document}